%% file: main.tex
\documentclass[11pt]{article}
\usepackage[margin=1in]{geometry}
\usepackage[T1]{fontenc}
\usepackage[utf8]{inputenc}
\usepackage{microtype}
\usepackage{amssymb,amsthm}
\usepackage{mathtools,bbm,authblk}
\usepackage{xcolor}
\usepackage{comment}
\usepackage[inline]{enumitem}
\usepackage{hyperref}
\usepackage{framed}

\usepackage{caption}
\hypersetup{
colorlinks=true,
allcolors=blue,
pdfauthor={Chapman et al.},
pdftitle={SPSE},
}
\usepackage{tikz}
\usetikzlibrary{quantikz}
\usepackage[capitalize,nameinlink]{cleveref}
\usepackage{soul}

\newcommand{\C}{\mathbb C}
\newcommand{\Log}{\operatorname{Log}}
\newcommand{\ii}{\mathrm i}
\newcommand{\dd}{\,\mathrm d}
\newcommand{\e}{\mathrm e}
\newcommand{\mnorm}{m_\Omega^{(k)}}
\newcommand{\mnormhat}{\widehat{m}_\Omega^{(k)}}
\newcommand{\Mell}{\mathcal{M}_\ell}
\newcommand{\Tr}{\operatorname{Tr}}
\newcommand{\ntr}{\operatorname{ntr}}
\newcommand{\rank}{\mathrm{rank}}
\newcommand{\sinc}{\operatorname{sinc}}
\newcommand{\I}{\mathbb I}
\newcommand{\ketbra}[2]{|#1\rangle\!\langle#2|}
\theoremstyle{plain}
\newtheorem{theorem}{Theorem}
\newtheorem{proposition}[theorem]{Proposition}
\newtheorem{lemma}[theorem]{Lemma}
\newtheorem{corollary}[theorem]{Corollary}
\newtheorem*{theorem*}{Theorem}
\newtheorem*{corollary*}{Corollary}
\theoremstyle{definition}
\newtheorem{definition}[theorem]{Definition}
\theoremstyle{remark}
\newtheorem*{remark}{Remark}

\title{Single-Particle Spectral Estimation}
\author[1]{Adrian Chapman}
\author[1]{Charles Derby}
\author[1,2]{Steven T. Flammia}
\author[1,3]{Yeongwoo Hwang}
\author[1]{Joel Klassen}
\author[1,4]{Calum McCartney}
\affil[1]{Phasecraft}
\affil[2]{Virginia Tech}
\affil[3]{Harvard University}
\affil[4]{University College London}
\date{October 1, 2026}

\renewcommand{\Re}{\operatorname{Re}}
\newcommand{\supp}{\operatorname{supp}}
\newcommand{\reg}[1]{\mathsf{#1}}

\newcommand{\ProofOf}[1]{\par\medskip\noindent\textbf{\Cref{#1}.}\par\nopagebreak\smallskip}

\begin{document}

\maketitle

\begin{abstract}
    We study Hamiltonian learning in a novel setting, where we are promised a physical structure that is obfuscated by a global unitary. 
    In particular, we consider the task of learning the weights \(E_i\) of a Hamiltonian \(H=U^{\dagger}(\sum_{i \in I} E_i n_i + H_O)U\) comprised of a set of non-interacting modes \(I\) and an arbitrary ``spectator'' \(H_O\) acting on a disjoint set of modes \(O\). 
    The non-interacting structure is hidden by an unknown global unitary \(U\), and our goal is to recover the \(E_i\) without reconstructing \(U\).
    Such Hamiltonians are well-motivated in the setting of materials modeling, where they find application for classical methods such as mean-field theory, Green’s functions, and density functional theory. 
    We show that this problem is DQC1-hard, despite the apparent simplification given by the non-interacting structure. 
    We then introduce SPICES (single-particle inference via contour estimators), which combines a Hadamard test with a novel classical postprocessing procedure to efficiently recover the spectral distribution in Wasserstein distance under certain natural assumptions.
    The method exploits the correspondence between single-particle energies of the individual quasiparticles and the Fourier frequencies governing zeros of the partition function at complex inverse temperature.
    As quasiparticle spectra feature ubiquitously in downstream materials calculations, and with SPICES being simple and practical to execute, we expect that SPICES is a promising route to quantum applications in materials simulation.
\end{abstract}

\section{Introduction}
In this work we consider the notion of a \textit{free spectrum} -- wherein the eigenvalues of a Hamiltonian are decomposable into distinct sums of a polynomial number of ``single particle'' energies \(\{E_i\}\).
A motivating example of such systems are non-interacting fermionic systems and spin systems. 
Our focus is on the setting where this spectral freedom is not necessarily manifest in the form of the Hamiltonian -- for example when it is obscured under a unitary transformation or holds either approximately or exactly in some interacting system -- and furthermore when it may be embedded in a more general non-free system. 
In this setting we are given a number-conserving \(M\)-mode (either fermionic or spin) many-body Hamiltonian \(H\) in a fixed, classically describable basis with a promise that there exists a global unitary \(U\) and a partition of modes
\(\{1,\dots,M\}=I\sqcup O\) such that
\begin{align}
    U H U^{\dagger} = \sum_{i \in I} E_i n_i + H_O,
    \qquad U N U^{\dagger} = \sum_{j=1}^{M} n_j ,
    \label{eq:ffansatz}
\end{align}
where \(N\) is the total occupation observable and \(H_O\) acts only on the modes \(O\). 
We call the modes \(I\) \emph{protected} and the modes \(O\) the \emph{spectator sector}.
The task is then to learn the effective single-particle energies \(E_i\) associated with the submodes \(I\) \emph{without knowing \(U\) a priori} given query access to \(H\) and \(N\) and some guiding input state such as a thermal state. 
We call the problem of inferring the energies \(\{E_i\}\) within a given energy window \(\Omega\) the \textit{Single-Particle Spectral Estimation} (SPSE) problem, which we formulate as retrieving a spectral distribution that is close in Wasserstein distance to a thermally weighted Dirac comb with peaks situated at the energies \(E_i\).

\begin{definition}[Single-Particle Spectral Estimation, SPSE (informal version of \cref{def:spse-formal})]
    Given a Hamiltonian \(H\) and number operator \(N\) described by \cref{eq:ffansatz} such that \([H, N] = 0\), recover a distribution over the single-particle \(\{E_i\}\) to accuracy \(\epsilon\) in Wasserstein-1 distance with high probability.
\end{definition}

We formalize this definition, with the aim of operationalizing concepts from many-body physics from a learning perspective -- and introduce an algorithm which targets this formal notion of the problem. 
We call this algorithm SPICES: Single-Particle Inference via Contour Estimators. 
The SPICES protocol makes use of the technical insight that in the case of this effectively free model the \textit{complex grand canonical partition function} \(Z(s,z)\) -- where here \(z\) is the complex fugacity \(z = \e^{\beta \mu + \ii \vartheta}\) and \(s\) is the complex time -- admits a factorization into a free part \(Z_I\) and a spectator part \(Z_O\)
\begin{equation}
    Z(s,z)
    \coloneqq 
    \Tr\bigl\{z^{N} \exp(-s H)\bigr\} =Z_I Z_O =  \prod_{i\in I} \bigl(1+z \e^{-s E_i}\bigr) Z_O ,
\end{equation}
where \(Z_O(s,z)=\Tr_O\bigl\{z^{N_O}\e^{-sH_O}\bigr\}\) with \(N_O=\sum_{j\in O}n_j\); the roots of \(Z_I\) correspond precisely to the free energies. 
The SPICES protocol proceeds in the spirit of Lee--Yang and Fisher zero analysis~\cite{Yang1952, Lee1952, Fisher1965}. 
Namely, it estimates a pair of circular logarithmic Stieltjes contour integrals over the inverse fugacity \(z^{-1}\) so that the first root moment \(\sum_{i\in I} \e^{-s E_i}\) can be isolated and retrieved. 
Samples along the points of these contour integrals can be produced by Hadamard test circuits acting on a Gibbs state
\begin{align}
    \rho_{\beta, \mu} = \e^{-\beta(H - \mu N)}/Z(\beta, \e^{\beta\mu}).
\end{align}
Applying a time-domain Fourier transform to the root moment retrieves the thermally weighted peaks at the energies \(E_i\). We emphasize here that although in this work we present the SPICES algorithm in a formal capacity, we have tested this technique extensively numerically and find that it works well as a practical algorithm. 

The energy window \(\Omega\) enters SPICES geometrically.
With \(w=-1/z\) and \(P(s,w)=w^{M}Z(s,-1/w)\), each protected mode contributes a root \(w=\e^{-sE_i}\) of modulus \(\e^{-\beta E_i}\), so higher-energy modes lie closer to the origin.
The circles \(|w|=\e^{-\beta\Omega_{\max}}\) and \(|w|=\e^{-\beta\Omega_{\min}}\) are probed by Gibbs states with chemical potentials \(\mu=\Omega_{\max}\) and \(\mu=\Omega_{\min}\), respectively, and bound the annulus
\begin{equation*}
    A_\Omega=\bigl\{w\in\mathbb C:\ \e^{-\beta\Omega_{\max}}<|w|<\e^{-\beta\Omega_{\min}}\bigr\},
\end{equation*}
and the difference of the two contour integrals is the root moment of the roots inside it (\cref{fig:spices-contour-certificates}).
This is the root moment of exactly the protected modes provided the spectator polynomial \(P_O(s,w)=w^{|O|}Z_O(s,-1/w)\) has no roots in \(A_\Omega\) at every point \(s=\beta+\ii t\) with \(t\in[0,t_{\max}]\), where \(t_{\max}=O(1/\epsilon)\) is the largest time sampled by SPICES (\cref{thm:spices-recovery}).
Elsewhere the spectator roots may lie anywhere in the plane, provided they stay radially separated from the two circles (\cref{def:sc}).
Since the roots of \(Z\) depend only on \((H,N)\), the target \(\nu_\Omega\) is the same for every decomposition satisfying the promise, so the non-uniqueness of \cref{eq:ffansatz} is harmless.

\begin{definition}[Stability Criteria] \label{def:sc}
    Let \(R_\mu=\e^{-\beta\mu}\) for \(\mu\in\{\Omega_{\min},\Omega_{\max}\}\) be the radii of the two contours, and let \(\theta_j=2\pi j/S\) (with \(S\) the number of quadrature points on each circle) and \(t_\ell=\pi\ell/L\) be the quadrature angles and sample times used by SPICES (\cref{thm:spices-recovery}), with \(L=\Omega_{\max}-\Omega_{\min}\) and \(s_\ell=\beta+\ii t_\ell\).
    We say that an SPSE instance satisfies the \emph{stability criteria} (the ``contour certificates'' of \cref{thm:spices-recovery}) with constants \(0<q<1\) and \(0<b\le 1\) if, at every sample point \(s_\ell\):
    \begin{enumerate}[label=\arabic*.]
        \item \textbf{Radial separation:} every root \(\zeta\) of \(P(s_\ell,\cdot)\) is bounded away from both contours, 
        \begin{align}
            \label{eq:radial_separation}
            \min\{|\zeta|/R_\mu,\ R_\mu/|\zeta|\}\le q
        \end{align} for both \(\mu\in\{\Omega_{\min},\Omega_{\max}\}\);
        \item \textbf{Signal floor:} for both \(\mu\), every quadrature angle \(\theta_j\), and every sample time \(t_{\ell}\)
        \begin{align}
            \left|\Tr\!\left(\rho_{\beta, \mu} \e^{-\ii t_\ell H + \ii(\pi - \theta_j)N}\right)\right| \geq b .
        \end{align}
    \end{enumerate}
\end{definition}
In this work, an SPSE instance satisfying the above stability criteria is called \emph{stable SPSE} \footnote{While the stability criteria are fairly natural physical assumptions, we do not rule out the possibility that other natural notions of stability for this problem exist.}.
We prove that SPICES can efficiently and robustly solve stable SPSE given: query access (either by an operator expansion or sparse access) to the Hamiltonian; the total occupation observable; and a sufficiently low temperature quantum input state.

\begin{figure}[!t]
    \centering
    \resizebox{\textwidth}{!}{\input{annulus_recovery.tex}}
    \caption{Contour geometry for \cref{thm:spices-recovery}. 
    The annulus \(A_\Omega=\{R_{\mathrm{in}}<|w|<R_{\mathrm{out}}\}\), where \(R_{\mathrm{in}}=\e^{-\beta \Omega_{\max}}\) and \(R_{\mathrm{out}}=\e^{-\beta \Omega_{\min}}\), contains exactly the three protected roots \(\e^{-s_\ell E_a}\); other roots lie outside it.
    The radial separation condition \eqref{eq:spices-radial-separation} excludes every root \(\zeta\) within the regions \(qR<|w|<R/q\) around both boundaries of the annulus; this is depicted at a logarithmic scale in the inset on the right. 
    Squares mark five equally spaced quadrature endpoints on each circle, as used in \eqref{eq:midpoint-quadrature-rule}. 
    The choice \(S=5\) is for illustration only.}
    \label{fig:spices-contour-certificates}
\end{figure}
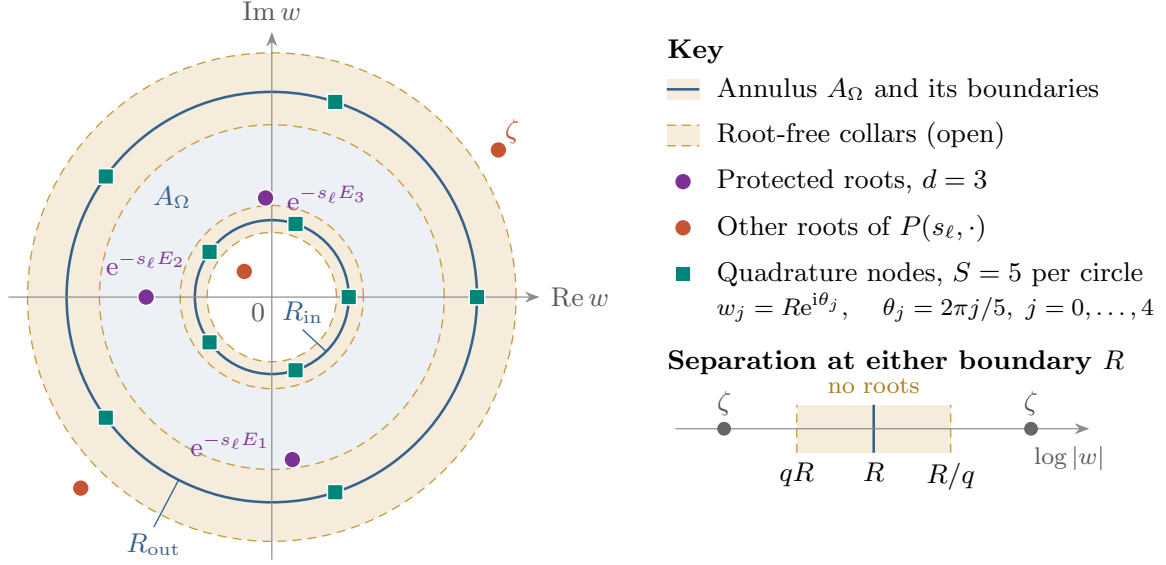

\begin{theorem}[SPICES recovery from contour certificates (informal version of \cref{thm:spices-recovery}) ]
    \label{thm:spices-informal}
    Let \((H,N,\Omega,\beta,\epsilon,\delta)\) be an instance of stable SPSE with window width \(L=\Omega_{\max}-\Omega_{\min}\), accuracy \(0<\epsilon<L/2\), and inverse-polynomial failure probability \(0<\delta<1\).
    Suppose that the Gibbs preparations \(\rho_{\beta,\mu}\) and controlled joint evolutions \(\e^{-\ii tH+\ii\vartheta N}\) are available to any required inverse-polynomial accuracy with polynomial overhead; and that \(L\), \(\epsilon^{-1}\), \(b^{-1}\), \((1-q)^{-1}\) and \(\e^{2\beta L}\) are bounded by polynomials in \(M\).
    Then, with probability at least \(1-\delta\), SPICES returns a probability measure \(\widehat\nu\) with \(W_1(\widehat\nu,\nu_\Omega)\le\epsilon\), using polynomially many gates and samples and polynomial classical postprocessing.
\end{theorem}
Notably we do not require knowledge of the basis in which this free spectrum is manifest. We give some natural sufficient conditions for these criteria to hold, motivating the utility of SPICES. 
These include: a promise that when the energies of modes \(I\) are bounded away from the contour by energy \(\Delta\) and the spectator roots are not inside the contour, then a temperature of \(\beta \geq \Delta^{-1} \log(3M)\) suffices to satisfy the signal-floor criterion of \cref{def:sc} (\cref{cor:log-inv-temp-suffices}).
Furthermore, that the spectator roots remain outside of the contour up to a maximum time provided the interactions between spectator modes are sufficiently weak (\cref{thm:finite-time-spectator-stability,cor:log-gap-perturbative-stability}). 

Finally, we show that SPSE remains DQC1-hard even under the stability criteria -- evidence that quantum resources are necessary for this task.

\begin{theorem}[SPICES solves DQC1-hard instances (informal version of \cref{thm:ff-hardness,cor:ff-spices-hard-family})]
    Work in the input model of \Cref{thm:spices-informal}.
    There exists a family of instances in this model such that for each instance defined by a Hamiltonian $H$ and unitary $W$,
    \begin{enumerate}[label=(\arabic*)]
        \item in the basis of \(W\) (that is, with \(U=W^\dagger\)) the Hamiltonian has the form of \cref{eq:ffansatz} with a single protected mode, whose energy \(E_\star\in\{16,17\}\) is the only protected energy in \(\Omega\), and all remaining modes form an interacting spectator sector;
        \item determining \(E_\star\) to within \(1/2\) is DQC1-hard;
        \item there are parameters $\beta, q, b$ such that the stability criteria of \cref{def:sc} hold and for \(\epsilon, \delta \in 1/\textsf{poly}(M)\), SPICES returns with probability at least \(1-\delta\) a probability measure \(\widehat\nu\) with \(W_1(\widehat\nu,\delta_{E_\star})\le\epsilon\), using polynomially many gates and samples and polynomial classical postprocessing.
    \end{enumerate}
\end{theorem}

The SPSE problem has strong motivations both from a physics and computer science perspective. We briefly discuss both here before proceeding with an overview of the rest of the paper.

\paragraph{Computer Science Motivation}
Learning a Hamiltonian model for an unknown quantum system is a fundamental task in quantum information processing.
Protocols restricted to single-copy measurements from the dynamics of a local many-body Hamiltonian can achieve the standard quantum limit of \(O(\epsilon^{-2})\) total time evolution, whereas protocols utilizing entanglement with an ancilla can reach the Heisenberg limit of \(O(\epsilon^{-1})\) total time evolution to reach a corresponding accuracy \(\epsilon\) \cite{huang2023learning}.
Related results establish that only a polynomial number of measurements are needed to learn a Hamiltonian from access to its Gibbs state~\cite{anshu2021sample}, and this sample complexity is provably optimal at sufficiently high temperature~\cite{haah2024learning}.
Naturally, methods for Hamiltonian learning are directly applicable to the certification, benchmarking, and control of quantum devices \cite{wiebe2014hamiltonian}. Notably, the SPICES protocol constitutes a rare case of a practical learning algorithm. 

An important special case of Hamiltonian learning concerns systems with a special structure, such as geometric locality \cite{arunachalam2025testing}, sparsity \cite{ma2024learning}, or symmetry \cite{zhou2024learning}.
In addition to providing an advantage for Hamiltonian learning, these families of structure, and the certification thereof, can provide valuable insight into the physics of the systems we wish to model.
The SPSE problem constitutes a hitherto unexplored form of Hamiltonian learning.  It is worth noting that the inclusion of an obfuscating basis via the unitary \(U\) makes SPSE substantially different from the conventional setting of Hamiltonian tomography.
Rather, SPSE and its algorithm, SPICES, aim to utilize the particular spectral structure of \(H\) as a collection of subset sums from a polynomially large base set of single-particle energies.

Free fermion systems have found a surprising degree of application to computer science.
Matchgate quantum circuits were originally proposed by Leslie Valiant as an example of a nontrivial classically simulable family of quantum circuits \cite{valiant2002quantum}, and these were later discovered to constitute the quantum circuit family of fermionic linear optics \cite{knill2001fermionic,terhal2002classical,bravyi2002fermionic,bravyi2005lagrangian,jozsa2008matchgates}.
More recently, mappings to free fermions were discovered to have played a role in the proof of the sensitivity conjecture \cite{huang2019induced,gu2019majorana}.
We expect that our characterization, by revealing hidden structure in these settings, may find application in all of these areas.

Free models are also known to be computationally hard in a different access model, where the number of modes is exponentially large and the system is specified through efficient oracles.
Estimating the kinetic energy of \(2^n\) coupled classical oscillators, whose masses and spring constants can be queried efficiently, is BQP-complete, and any classical algorithm needs \(2^{\Omega(n)}\) queries~\cite{babbush2023exponential}.
An analogous result holds for non-interacting fermions: deciding whether a given mode is appreciably occupied after a system of \(2^n\) modes has evolved under a sparse quadratic Hamiltonian, given by oracles for its sparse description, is BQP-complete~\cite{stroeks2024solving}.
These access models differ substantially from ours, but they show that a free structure on its own does not make a problem easy.
Here we show that hardness persists in a far more modest access model.
The number of modes is polynomial, the Hamiltonian is a short explicit Pauli sum, and the only source of difficulty is that the free structure is hidden by a change of basis that is not given to the algorithm; recovering a single protected energy is then DQC1-hard (\cref{thm:ff-hardness}).

\paragraph{Physics Motivation for SPSE}
The particle is among the oldest and most successful paradigms in physics: the energy of a system is accounted for by the energies of its constituents.
Interacting quantum systems are not atomistic in this sense -- their eigenenergies do not generically decompose into sums of single-particle contributions.
It is striking that the particle picture remains central to our understanding of quantum matter, from orbital energies in Hartree--Fock and Kohn--Sham theory~\cite{kohn1965self} to band structures and Landau Fermi-liquid theory (FLT)~\cite{landau1957theory,landau1957oscillations,landau1959theory,baym1991landau}.
One reason for this success is that the non-interacting limit, where the particle picture is well defined, serves as a natural starting point from which to perturbatively express the effect of interactions, so for weakly interacting systems a particle picture suffices.
What is remarkable is that this picture often survives strong interactions as well: at low energies, many strongly interacting systems behave as collections of long-lived \emph{quasiparticles} whose parameters are renormalized by the interactions. For example, in heavy-fermion compounds quasiparticle masses can reach hundreds of times the bare electron mass~\cite{stewart1984heavy}.

In FLT, adiabatically switching on interactions maps the low-lying eigenstates of a free system onto those of the interacting one, and their energies are expanded in the deviations \(\delta n_i\) of the dressed-particle occupations from the ground state~\cite{interacting_electrons},
\begin{equation}\label{eq:flt_energy_draft}
    \lambda(\delta n) = \lambda_0 + \sum_i E_i\,\delta n_i + \sum_{ij} f_{ij}\,\delta n_i\,\delta n_j + O(\delta n^3).
\end{equation}
To first order the spectrum is free, with quasiparticle energies \(E_i\).

The structural promise~\eqref{eq:ffansatz} is designed to target the dominant features of this picture.
First, the adiabatic dressing map is an unknown unitary that is, crucially, not a single-particle basis rotation: dressed operators \(U^\dagger c_j^\dagger U\) acquire multi-particle components, so no choice of bare single-particle operators captures them.
This is the role of the obfuscating basis.
While physical dressing maps are typically quasi-local~\cite{hastings2005quasiadiabatic,bachmann2012automorphic}, SPICES makes no use of this, and we allow \(U\) to be arbitrary.
Second, quasiparticles are well defined only in a low-energy window about the chemical potential and at low temperature; this is the role of the window \(\Omega\), the protected modes \(I\), and the thermal input \(\rho_{\beta,\mu}\), while the remaining degrees of freedom form the spectator sector \(O\).
Third, physical quasiparticles have finite lifetimes; correspondingly, our stability guarantees for interacting spectators hold only up to a finite time horizon (\cref{thm:finite-time-spectator-stability}). A two-mode example shows that when the spectator modes are only weakly separated from the window, some such restriction cannot be avoided for general interacting spectators (\cref{subsec:interacting-long-time-obstruction}); we do not claim that it is necessary at larger gaps.
We note that for simplicity our structural promise~\eqref{eq:ffansatz} keeps only the leading term of~\eqref{eq:flt_energy_draft}: the Landau interaction \(f_{ij}\) couples quasiparticles to one another, whereas we require the protected modes to decouple -- we aim to address generalizations of this in future work.
Finally, quasiparticles are a property of the low-energy states at finite density, not of the few-particle sectors of \(H\); the dressing \(U\) and the decoupling in~\eqref{eq:ffansatz} should be understood as holding only near the relevant filling, with the global promise an idealization.
SPICES respects this,  querying \(H\) only through Hadamard tests on thermal states \(\rho_{\beta,\mu}\) with \(\mu\) at the edges of \(\Omega\), so its signal is weighted toward exactly this low-energy, finite-density subspace.

The standard probe of quasiparticle energies is the single-particle Green's function \(G_{ij}(\omega) = \langle c_i(\omega - H + \lambda_0 + \ii\eta)^{-1}c_j^\dagger\rangle\), whose poles sit at excitation energies.
Peak strengths are set by the overlap of the bare excitation with the dressed one, and are reduced by strong dressing.
Resolving quasiparticles from \(G\) therefore requires approximate knowledge of the dressed quasi-particle operator.
\(GW\) methods supplement this through a mean-field reference, which works well when the choice of reference is good but offers no a priori guarantees otherwise~\cite{hedin1965new,the_GW_method,GW_starting_points,GW_starting_points_II}. SPICES avoids this dependence entirely. SPICES employs similar quantum resources to quantum Green's function algorithms~\cite{somma2002} -- Hadamard test circuits, time evolution and low-energy input states. For these reasons we view SPICES as complementary to quantum Green's function methods, particularly in regimes where quasiparticle weights may be suppressed by dressing.

\paragraph{Overview of the remainder of the paper.} We give a brief conclusion in \cref{sec:conclusion}.
In \Cref{sec:prelim} we overview the SPICES protocol and establish definitions which will be useful for the remainder of the paper, and in \Cref{sec:spse} we state the SPSE problem formally. We establish that this protocol is robust to noise (\Cref{sec:noise}), can be implemented efficiently (\Cref{sec:free-gapped}), and that it works outside the perfectly-free setting, notably in the setting of ansatz established in \Cref{eq:ffansatz} (\Cref{sec:perturbative-stability}). Finally, we establish in \Cref{sec:ff-hardness} that SPICES solves a DQC1-hard problem, providing evidence that it offers an advantage over entirely classical methods.

\section{Discussion}\label{sec:conclusion}

We suspect that the rigorous promises that SPICES requires can likely be relaxed in practice.
We use logarithmic inverse temperature dependence to guarantee these, but it is not clear \emph{a priori} that this requirement is ultimately necessary.
In experiments one is often satisfied with coarse-grained features of the single-particle energy density, such as bandgaps, rather than individual energies.
We would find it plausible that a more forgiving scaling may suffice in this relaxed setting.
Additionally, our algorithmic results assume controlled evolution and Gibbs preparation at inverse temperature \(O(\log M)\), and the hardness result shows that the number operator cannot be straightforwardly removed from the access model.
It would be interesting to identify if certain physical systems allow for a weaker form of access to \(N\) for the same task.

There are a number of directions we intend to elaborate on in near term future work. 

Firstly, we have been engaged in testing numerically the SPICES method and have seen that it performs well. The present paper is meant to present the theoretical foundations of the method with a practical survey to follow.

Secondly, the SPICES protocol may be generalized to a larger family of good quantum numbers beyond total charge $N$. 
For example these techniques may be extended to total momentum in order to extract single particle momenta, i.e. band structure. 
From this perspective we may think of the SPICES protocol as a technique for identifying where global symmetries decompose into local symmetries. 

Thirdly, the SPICES estimator can be used as a primitive in a swept estimator, wherein the radial values of the contours are swept across an energy domain, yielding a regularized energy distribution. We have seen that this works well numerically -- which suggests that the condition that roots be bounded away from the contours may potentially be relaxed in this setting. 

Fourthly, we would like to understand better where in practice we can expect the benefits of the SPICES approach to be most impactful and how the performance of the method compares to and can supplement Green's functions for practical applications.
In line with this, we would like to understand in detail how the higher-order parameters in the Landau Fermi Liquid hypothesis given in equation \ref{eq:flt_energy_draft} impact the estimator produced by SPICES.

Finally, we would like to show that SPICES can also be employed to witness the \emph{absence} of a free model, improving its usefulness as a diagnostic method.

\paragraph{Acknowledgements} We thank Raul Santos, Toby Cubitt, Evan Sheridan, and Glenn Jones for helpful discussions throughout this project.

\paragraph{Author Contributions} JK formulated the research direction. SF and JK provided technical leadership. AC, CD, SF and JK developed the SPICES algorithm. AC, JK, CD, SF, YH, CM developed the proofs of correctness and noise robustness. YH and AC developed the proofs of DQC1 hardness. 

\paragraph{AI Disclosure.} The formulation of the research direction, problem statement, algorithm and hardness construction was done by the authors. Proof sketches were formulated by the authors with AI used to assist in proof completion. AI was used in optimizing the bounds and resource estimates for the algorithm. AI was used to assist with paper drafting, editing and the design of figures and diagrams.

\paragraph{Note added.} During the final stages of drafting of this document we became aware of contemporaneous work employing similar algorithmic techniques to the ones used here in Ref.~\cite{mcintyre2026one}. We note that our setting is distinct from this one in that we allow for the presence of an interacting spectator Hamiltonian and includes models which are free-fermion in an unknown basis. 

\tableofcontents

\section{The SPICES Protocol}\label{sec:prelim}

In this section we describe the SPICES protocol. To motivate the protocol we consider a purely free system \begin{equation}
    H_0 = \sum_{j=1}^{M} E_j n_j,
    \label{eq:free-model}
\end{equation} and introduce the inverse-fugacity polynomial $P(s,w)$ of the complex time $s=\beta+it$ and the inverse fugacity $w$, both defined in \Cref{sec:free-model}. We show in \Cref{sec:free_hamiltonians} that in the case of a free system the polynomial factorizes as,
\[
    P(s, w) = \prod_{j=1}^M (w - \zeta_j)\,,\qquad \zeta_j = e^{-s E_j}\,,
\]
and the roots of this polynomial correspond precisely to exponentials $e^{-s E_j}$ of the single-particle energies. 
Retrieving roots directly from a polynomial is a notoriously unstable problem. 
Instead, SPICES will obtain information by sampling a signal along chosen contours. 
We introduce \emph{root moments} to connect contour integrals to polynomial roots. Formally, root moments are defined by taking a Stieltjes integral of the log of the inverse-fugacity polynomial over an annulus \(A\)
\begin{equation}
        \mathcal{M}_A^{(k)} \coloneqq \frac{1}{2\pi \ii} \oint_{\partial A} w^k \dd\bigl(\log P(w)\bigr) =  \sum_{\zeta_j \in A} \zeta_j^k\,. \label{eq:stieljes}
    \end{equation} 
In the case of an exactly free system, taking a time-domain Fourier transform of the root moment at constant inverse-temperature \(\beta\) constitutes the signal we wish to retrieve, a thermally weighted distribution over energies $E_j$.

It remains to estimate the root moments.
They can be estimated from a normalized proxy \(\Xi\), 
\[\Xi(s,z)
    \coloneqq
    \Tr\{\rho(\beta,|z|)\exp(-\ii t H + \ii \vartheta N)\}\]
the expectation value of \(\e^{-\ii tH+\ii\vartheta N}\) in a thermal state, which can be measured by a Hadamard test circuit. On circular contours the logarithm of \(\Xi\) yields the same root moments as the logarithm of \(P\). 

Thus the SPICES protocol is given by the following:
\begin{framed}
\textbf{SPICES Protocol} Fix an energy annulus $A$ (a difference of two discs) and an inverse temperature $\beta$, then estimate $\Xi(s,z)$ at a series of points $z$ along the perimeter of the annulus. Use these points to approximate $\mathcal{M}_A^{(k)}$ for a fixed value of $t$. Repeat over a sequence of $t$ values to retrieve a time series signal $\mathcal{M}^{(k)}_A(t)$. Apply a time domain Fourier transform to retrieve a spectral distribution. 
\end{framed}
We discuss this in more detail in the remaining portions of this section. 

In our formal problem statement of SPSE \cref{def:spse-formal} the ansatz Hamiltonian is only free on a subset of modes and allows for a spectator Hamiltonian $H_O$. The contour integration used by SPICES can exclude roots outside of the chosen contour, including roots introduced by spectator modes. We discuss in \Cref{sec:perturbative-stability} how the SPICES protocol behaves in this setting, and when it can robustly retrieve the free roots. 

\subsection{Grand canonical model}\label{sec:free-model}

Motivated by Landau Fermi liquid theory, we are formalizing what follows in terms of fermionic modes.
However, nothing in our paper depends on the canonical fermionic anticommutation relations, and all of what follows could be formulated using qubits with the identification \(n_j = (\I+Z_j)/2\), where \(Z_j\) is the Pauli \(Z\) operator acting on the \(j\)th qubit.

Consider \(M\) fermionic modes with occupation operators \(n_j\) and total occupation \(N\):
\begin{equation}
    n_j^2=n_j, \qquad
    [n_j,n_k]=0, \qquad
    N=\sum_{j=1}^{M}n_j
    \label{eq:number-operator-relations}
\end{equation}
Given a Hamiltonian \(H\) which is number conserving, \([H,N] = 0\), define the \textit{complex grand canonical partition function}
\begin{equation}
    Z(s,z) \coloneqq \Tr\bigl\{z^{N} \exp(-s H)\bigr\}
    \label{eq:general-partition}
\end{equation}
This is implicitly a function of the Hamiltonian \(H\), and is explicitly a map from \(\C \times \C \mapsto \C\).

In the usual physical case, we have \(s = \beta\) and \(z = e^{\beta \mu}\) for inverse temperature \(\beta\) and chemical potential \(\mu\), yielding the standard grand canonical partition function
\begin{equation}
    Z(\beta,e^{\beta \mu}) = \Tr\bigl\{ \exp(-\beta (H-\mu N )\bigr\}.
    \label{eq:free-partition}
\end{equation}
Here \(z\) is known as the \textit{fugacity} and effectively measures the cost of adding an additional particle to the system.
\cref{eq:general-partition} can be understood as the analytic continuation of this physical case to complex temperature and complex chemical potential.

We can write the complex fugacity as
\(z = \e^{\beta \mu + \ii \vartheta}\),
where \(|z| = \e^{\beta \mu}\) is the standard fugacity and \(\vartheta\) is the physical fugacity angle or imaginary chemical potential.
For our results, it is frequently more natural to use the \textit{(signed) inverse-fugacity} coordinate
\begin{equation}
    w \coloneqq -\frac{1}{z}
    \label{eq:inverse-fugacity-coordinate}
\end{equation}
with
\begin{equation}
    P(s,w)\coloneqq w^M Z\bigl(s,-\tfrac1w\bigr).
    \label{eq:inv-partition-poly}
\end{equation}
This is the (signed) reciprocal polynomial of the grand canonical partition function in the fugacity variable.
The minus sign is only a convention to remove a minus sign from the roots of the polynomial used below.
We treat \(w\) as merely a change of variables, but one physical intuition for \(-w\) is that it is the fugacity of a hole instead of a particle; large \(|w|\) (small \(|z|\)) means adding \textit{holes} is thermodynamically favored, and small \(|w|\) (large \(|z|\)) means adding holes is suppressed.
In the \(w\) coordinate we write
\begin{equation}
    R = \e^{-\beta\mu}, \quad w = R\e^{\ii\theta}.
\end{equation}
Thus \(|w|=\e^{-\beta\mu}\) is a Gibbs-weight scale, while
\(\theta=\pi-\vartheta\) modulo \(2\pi\).
We also explicitly split the complex inverse-temperature \(s\) into a real inverse temperature \(\beta\) and treat the imaginary part as a time coordinate \(t\),
\begin{equation}
    s \coloneqq \beta + \ii t.
\end{equation}
We then treat the complex grand canonical partition function \(Z\) as a function of two \textit{independent} variables, \(Z(s,z)\) (or equivalently \(Z(s,-1/w)\)) with \(\Re(s) > 0\).
This is justified since the physical chemical potential can always be recovered from a nonzero fugacity for any finite temperature via \(\mu = \beta^{-1} \log |z|\).

\subsection{Free Hamiltonians}
\label{sec:free_hamiltonians}

From now, we almost exclusively consider a \textit{free} Hamiltonian \(H_0 = \sum_{j=1}^{M} E_j n_j\) (\cref{eq:free-model}), where the \(E_j\) are called the \textit{single-particle energies}.
We return to the ansatz of \cref{eq:ffansatz}, for which \(H_O \neq 0\), in \cref{sec:perturbative-stability}.

As we now prove, in the free model \cref{eq:free-model}, \(P(s, w)\) is monic, with roots
\begin{equation}
    \zeta_j \coloneqq \e^{-s E_j}
    \label{eq:zeta-zeros}
\end{equation}
which directly encode the thermally weighted single-particle energies.

\begin{proposition}[Free Factorization]
    \label[proposition]{prop:free-factorization}
    The free grand canonical partition function \(Z(s,z)\) is a degree-\(M\) polynomial in the fugacity \(z\) with factorization
    \begin{equation}
        Z(s,z) = \prod_{j=1}^{M} \bigl(1+z \e^{-s E_j}\bigr) = \prod_{j=1}^{M} \bigl(1 + z \zeta_j\bigr).
        \label{eq:Z-factorization}
    \end{equation}
    Equivalently,
    \begin{equation}
        P(s,w) = w^M Z\!\left(s,-\frac1w\right) =
        \prod_{j=1}^{M}(w-\zeta_j).
        \label{eq:inverse-fugacity-polynomial}
    \end{equation}
    Thus the zeros of \(P\) in the inverse fugacity \(w\) are \(\zeta_j=\e^{-sE_j}\), with
\(|\zeta_j|=\e^{-\beta E_j}\).
\end{proposition}

\subsection{The normalized signal and the single-particle gap}

Write the complex fugacity in polar coordinates with \(|z|=\e^{\beta \mu}\) being the standard physical fugacity, so that \(z=|z|\e^{\ii\vartheta}\).
For a general number-conserving Hamiltonian \(H\), the Gibbs state of this ensemble can be written as
\begin{equation}
    \rho(\beta,|z|) \coloneqq \frac{|z|^N \exp(-\beta H)}{Z(\beta,|z|)}.
    \label{eq:gibbs-state}
\end{equation}
This is the state \(\rho_{\beta,\mu}\) of the introduction, with \(|z|=\e^{\beta\mu}\).
The real time dynamics of the Hamiltonian and number operator is given by the unitary operator
\begin{equation}
    \mathcal{J}(t,\vartheta) \coloneqq \exp(-\ii t H + \ii \vartheta N).
    \label{eq:unitary}
\end{equation}
Using \([H,N]=0\), the trace inner product between the state and the unitary is
\begin{equation}
    \Xi(s,z)
    \coloneqq
    \Tr\{\rho(\beta,|z|)\mathcal{J}(t,\vartheta)\}
    =
    \frac{Z(\beta+\ii t,|z|\e^{\ii\vartheta})}{Z(\beta,|z|)}
    =
    \frac{Z(s,z)}{Z\bigl(\Re(s),|z|\bigr)}.
    \label{eq:normalized-signal}
\end{equation}

On a circle \(|w|=R\) in the inverse-fugacity coordinate, the corresponding physical fugacity has modulus \(|z|=R^{-1}\) and angle \(\vartheta=\pi-\theta\).

We call \(\Xi\) the \textit{normalized} partition function.
Because it can be written in the form \(\Tr(\rho \mathcal{J})\),
the Hadamard test~\cite{NielsenChuang2010} can be used to estimate the original
analytically continued partition function \(Z\)
up to the strictly positive normalization \(Z(\beta,|z|)\).
Since \(\lvert\Xi\rvert = \lvert \Tr(\rho \mathcal{J})\rvert\) does not have a nontrivial lower bound in general, we need some assumptions about our parameter choices to make more refined complexity statements.
For a fixed free Hamiltonian \(H_0\), define the
\textit{single-particle gap} \(\Delta\) as a function of the chemical potential \(\mu\) as
\begin{equation}
    \Delta(\mu) \coloneqq \min_k |E_k - \mu| \,.
    \label{eq:gap-def}
\end{equation}
Note that this is the gap between the \(M\) \textit{single-particle} energies and the chemical potential.
We caution the reader that the terminology ``gap'' is often used to refer to the smallest energy gap of the full many-body spectrum, possibly with respect to a reference state like the ground state.
Our notion of gap has no relation in general to the many-body energy gap; it can be either larger or smaller depending on the Hamiltonian and chemical potential.
\subsection{Root moments}
\label{sec:root_moments}

In this subsection we fix a value of \(s\) and consider the dependence of the partition function on the inverse-fugacity coordinate
\[
w=R\e^{\ii\theta},
\]
which the reader will recall is related to the physical fugacity by \(z=-1/w\).
Because \(s\) is fixed, we suppress the \(s\) dependence and write the inverse-fugacity polynomial as
\begin{equation}
    P(w) \coloneqq P(s,w) = w^M Z\!\left(s,-\frac1w\right) = w^M Z\!\left(\beta,\tfrac{1}{|w|}\right) \Xi\!\left(s,-\frac1w\right) .
    \label{eq:root-polynomial-w}
\end{equation}
By \cref{eq:inverse-fugacity-polynomial}, the roots of \(P\) are exactly \(\zeta_j=\e^{-sE_j}\) from \cref{eq:zeta-zeros}.
Recall the root moments introduced in \cref{sec:prelim}; we now define them for a general polynomial and region.

\begin{definition}[Root Moments]
    \label[definition]{def:root-moments}
    Let \(A\) be a compact region in the complex plane whose piecewise smooth boundary \(\partial A\) is positively oriented (outer boundary counterclockwise, inner boundaries clockwise, so that \(A\) lies to the left). 
    Let \(P(w)\) be a complex polynomial in the variable \(w\), and assume that \(P\) has no zeros on \(\partial A\). 
    For \(k\in\mathbb Z_{\ge0}\), the \textit{\(k\)th root moment} of \(P(w)\) with respect to the region \(A\) is 
    \begin{equation}
        \mathcal{M}_A^{(k)} \coloneqq \frac{1}{2\pi \ii} \oint_{\partial A} w^k \dd\bigl(\log P(w)\bigr) \,.
        \label{eq:root-moment}
    \end{equation}
\end{definition}

    The expression \(\dd(\log P(w))\) in \cref{eq:root-moment} should not be read as requiring a single-valued logarithm of \(P\) on all of \(A\).
    We only need \(P\) to be nonzero on the contour \(\partial A\).
    The integral should be interpreted as a Stieltjes integral along the parametrized boundary.
    Concretely, if \(\gamma:[0,1]\to\partial A\) is one oriented boundary component and \(L\) is a continuous lift of \(\log P(\gamma(x))\), then the contribution to the integral of that component is
\[
    \int_0^1 \gamma(x)^k\,\dd L(x).
\]
    The differential of this lifted logarithm is independent of the additive constant in the branch.
    We use this interpretation since our estimator for the root moments (\cref{eq:midpoint-quadrature-rule} below) follows this structure.
    Furthermore, the lifted logarithm on a closed contour need not return to the same value after one circuit.
    Its endpoint change is \(2\pi\ii\) times the winding number of \(P(\partial A)\) about the origin, and this is precisely what detects zeros inside the contour.
    For a fixed \(s\), this only requires phase unwrapping along the inverse-fugacity contour.
    If one later varies \(s\), a consistent logarithm for finite differences in \(s\) requires the same kind of phase unwrapping in the \(s\)-direction as well.
    We will discuss sufficient conditions for phase unwrapping below in \cref{sec:noise}.
    
    When \(P\) is nonzero on \(\partial A\) and the boundary is piecewise smooth, this agrees with the more familiar expression as an ordinary contour integral
    \begin{equation}
        \mathcal{M}_A^{(k)}
        =
        \frac{1}{2\pi \ii} \oint_{\partial A} w^k \dd\bigl(\log P(w)\bigr)
        =
        \frac{1}{2\pi \ii} \oint_{\partial A} \frac{w^k}{P(w)} \frac{\partial P}{\partial w} \dd w \,.
        \label{eq:alt-root-moment}
    \end{equation}

\begin{remark}[Relation to the signal \(\Xi\)]
    On a fixed circle \(|w|=R\), the normalizing factor \(Z\!\left(\beta,\tfrac{1}{R}\right)\) in \cref{eq:root-polynomial-w} is independent of \(\theta\), while the monomial factor \(w^M\) in \cref{eq:root-polynomial-w} is known explicitly.
    In particular, on regions \(A\) whose boundaries are oriented circles we have, for every \(s\) and for \(k>0\)
    \begin{equation}
        \mathcal{M}_A^{(k)} = \frac{1}{2\pi \ii} \oint_{\partial A} w^k \dd\bigl(\log P(w)\bigr) = \frac{1}{2\pi \ii} \oint_{\partial A} w^k \dd\bigl(\log \Xi(s,-w^{-1})\bigr).
        \label{eq:moment-from-signal}
    \end{equation}
    This provides the connection between root moments and the measurable signal \(\Xi\) for discs and annuli.
\end{remark}

\begin{proposition}[Power sum formula]
    \label[proposition]{prop:power-sum-formula}
    Let \(A\) be a compact region in the complex plane with positively oriented piecewise smooth boundary, and suppose that the complex polynomial \(P(w)\) has zeros \(\zeta_j\), listed with multiplicity, none of which lie on \(\partial A\). 
    Then
    \begin{equation}
        \mathcal{M}_A^{(k)} = \sum_{\zeta_j \in A} \zeta_j^k.
        \label{eq:power-sum}
    \end{equation}
\end{proposition}

\subsection{Discs, annuli and energy windows}

To see the importance of the root moments, let us plug back in the expression for \(\zeta_j\) in terms of the single-particle energies.
We reintroduce an explicit dependence on \(s = \beta + \ii t\).
From equations~\cref{eq:power-sum} and \cref{eq:zeta-zeros}, we have
\begin{equation}
    \mathcal{M}_A^{(k)}(s) = \sum_{\zeta_j(s) \in A} \zeta_j^k(s) = \sum_{\e^{-sE_j} \in A} \e^{-skE_j} = \sum_{\e^{-sE_j} \in A} \e^{-\beta k E_j} \e^{-\ii t k E_j}.
    \label{eq:root-moment-explicit}
\end{equation}

This equation suggests taking a Fourier transform in \(t\) in order to estimate the energies whose roots lie in \(A\) from the peaks in the spectrum.
This idea provides the correct intuition, but we need to address some technical caveats to make this work.

Foremost, if we know the \(k\)th root moment for times \(t \in [0,t_{\max}]\), we can only learn the single-particle energies to precision about \(k^{-1}/t_{\max}\), or about \(1/t_{\max}\) for fixed \(k > 0\).
Thus, if our goal is an algorithm that learns \(E_j\) to \textit{arbitrary} precision, we need to be able to observe root moments at arbitrarily long times.

The next caveat is that the boundary of \(A\) must be a union of centered circles.
Otherwise, \(A\) will not generally contain the same zeros for all times.
Indeed,
\[
|\e^{-(\beta+\ii t)E_j}|=\e^{-\beta E_j}
\]
is independent of \(t\), while the phase rotates at angular velocity \(-E_j\).
For a general region, a zero \(\e^{-sE_j}\) will fluctuate inside and outside of the region as time evolves, essentially convolving its contribution to the spectrum with a window function.
This will set a lower bound on the precision with which one can infer \(E_j\) due to spectral broadening (even if shot noise were no concern).

The third and worst caveat is that a zero fluctuating through the boundary necessarily implies that there is a time \(t\) when the zero is \textit{on} the boundary, which is expressly forbidden in the definition of the root moments and destroys the ability of any shot-based estimator to achieve good relative accuracy at these crossing times.

Because of these considerations we restrict ourselves to regions bounded by centered circles.
We will use discs and annuli, with annuli represented as differences of two discs with opposite boundary orientations.
These regions are merely sufficient for our purposes, not necessary.
However, we do not know of any good reason to consider more general regions, and this choice greatly simplifies our subsequent calculations.

It is convenient to overload the notation for root moments:
if \(D(R) = \{ w \in \C : |w| \le R\}\) is a closed disc of radius \(R\) in the complex plane, then we write
\[
\mathcal{M}_{R}^{(k)}(s) \coloneqq \mathcal{M}_{D(R)}^{(k)}(s).
\]
When \(R = \e^{-\beta \mu}\) is itself an inverse-fugacity radius and no \(E_j\) equals \(\mu\), inclusion of a zero in the sum now simplifies to
\[
\mathcal{M}_{R}^{(k)}(s) = \sum_{\e^{-\beta E_j}
< R} \e^{-\beta k E_j} \e^{-\ii t k E_j}
= \sum_{E_j > \mu} \e^{-\beta k E_j} \e^{-\ii t k E_j}.
\]
For disc root moments in the inverse-fugacity plane, the chemical potential thus acts like a hard high-pass filter, and only energies above it appear in the power sum.

Similarly, for an annulus
\begin{equation}
    A(R_{\mathrm{in}},R_{\mathrm{out}})
    =
    \{w\in\C:R_{\mathrm{in}}< |w|< R_{\mathrm{out}}\},
    \qquad
    R_{\mathrm{in}}=\e^{-\beta\mu_+},
    \qquad
    R_{\mathrm{out}}=\e^{-\beta\mu_-},
    \label{eq:annulus-radii-conditions}
\end{equation}
with \(\mu_-<\mu_+\) and no energy on either boundary,
\[
\mathcal{M}_{A(R_{\mathrm{in}},R_{\mathrm{out}})}^{(k)}(s)
=
\mathcal{M}_{R_{\mathrm{out}}}^{(k)}(s)-\mathcal{M}_{R_{\mathrm{in}}}^{(k)}(s)
=
\sum_{\mu_-<E_j<\mu_+}
\e^{-\beta k E_j}\e^{-\ii t k E_j}.
\]
Thus root moments on annular regions act as hard band-pass filters.

Fix an energy window \(\Omega\) and a bandwidth \(L\)
\[
\Omega=[\Omega_{\min},\Omega_{\max}],
\qquad
L=\Omega_{\max}-\Omega_{\min}>0.
\]
Assume that no \(E_j\) is an endpoint of \(\Omega\), and use the annulus
\begin{equation}
R_{\mathrm{out}}=\e^{-\beta \Omega_{\min}},
\qquad
R_{\mathrm{in}}=\e^{-\beta \Omega_{\max}},
\qquad
A_\Omega=A(R_{\mathrm{in}},R_{\mathrm{out}}).
\label{eq:window-annulus}
\end{equation}
Its contour gap is
\[
\Delta_\Omega
\coloneqq
\min\{\Delta(\Omega_{\min}),\Delta(\Omega_{\max})\}.
\]
Thus \(\Delta_\Omega>0\) is precisely the condition that both contour circles avoid the roots.
Recall that the definition of gap includes all energies, not just those inside the band.

A word on conventions: \(A_\Omega\) is an open annulus, whereas root moments are defined for compact regions.
The root moment \(\mathcal M_{A_\Omega}^{(k)}\) always means the one over the closed annulus \(\overline{A_\Omega}\), whose boundary is the two contour circles with the orientation above.
Statements that a polynomial has no roots in \(A_\Omega\) are always accompanied by control of the two circles (the gap \(\Delta_\Omega>0\), radial separation, or the Rouch\'e condition), so no root lies on either circle and the open and closed annulus contain the same roots.

\subsection{Circular quadrature}

In practice, the root moments have to be estimated from finitely many
values of the partition function on the contour.

For a circle of radius \(R>0\), choose an integer \(S>2\) and set
\[
\theta_j=\frac{2\pi j}{S},
\qquad
w_j=R\e^{\ii\theta_j},
\qquad
w_{j+1/2}=R\e^{\ii\theta_{j+1/2}},
\qquad
\theta_{j+1/2}=\frac{2\pi(j+1/2)}{S}.
\]
Define
\begin{equation}
    \sinc(x)
    \coloneqq
    \begin{cases}
        \frac{\sin(\pi x)}{\pi x} & \text{for } x\neq 0, \\
        1 & \text{for } x = 0.
    \end{cases}
    \label{eq:sinc-factor}
\end{equation}
The sinc-corrected logarithmic-increment midpoint rule for the \(k\)th root
moment (for integers \(0<k<S\)) on the circle is
\begin{equation}
    Q_{R,S}^{(k)}(s)
    \coloneqq
    \frac{1}{\sinc(k/S)}
    \frac{1}{2\pi\ii}
    \sum_{j=0}^{S-1}
    w_{j+1/2}^{\,k}
    \Log
    \left(
    \frac{P(w_{j+1})}{P(w_j)}
    \right).
    \label{eq:midpoint-quadrature-rule}
\end{equation}
Here \(\Log\) means the logarithmic increment obtained
from a continuous lift of \(\log P(R\e^{\ii\theta})\) along the circle.
Equivalently, if \(L(\theta)\) is such a lift, then the \(j\)th increment is
\[
L(\theta_{j+1})-L(\theta_j).
\]
The value is independent of the additive constant in the lift.
Since
\[
P(w)=w^M Z\biggl(s,-\frac1w\biggr),
\]
the same increments can be formed from normalized values on the
circle by using \cref{eq:root-polynomial-w} to obtain
\begin{equation}
    \label{eq:exchange-P-Xi-quadrature}
    \Log\frac{P(w_{j+1})}{P(w_j)}
    =
    M\Log\frac{w_{j+1}}{w_j}
    +
    \Log
    \frac{\Xi(s,-1/w_{j+1})}{\Xi(s,-1/w_j)}.
\end{equation}
The normalizing factor \(Z(\beta,R^{-1})\) is constant on the circle and
cancels, while the monomial term cancels inside the quadrature sum for \(0<k<S\), and contributes only a fixed known offset when \(k=0\).
Thus for practical purposes, \(P\) and the observable signal \(\Xi\) are interchangeable inputs for the quadrature rule on a circle.

\begin{remark}[Zeroth moment]
    There is no error due to quadrature aliasing in the \(k=0\) case:
\[
    Q_{R,S}^{(0)}(s)
    =
    \mathcal{M}_R^{(0)}(s).
\]
    Indeed, the sum of the unwrapped logarithmic increments is the total change
    of \(\log P\) around the circle, hence \(2\pi\ii\) times the number of
    zeros inside.
    The same statement holds for annuli by taking oriented differences of the
    two boundary circles.
    The deterministic quadrature estimates below therefore focus on \(0<k<S\).
\end{remark}

For \(0<k<S\), the factor \(\sinc(k/S)^{-1}\) is included in the definition (and is finite, since \(\sinc(k/S)=0\) exactly when \(S\) divides \(k\)) because
midpoint quadrature retains the \(k\)th Fourier mode
with the sinc factor \(\sinc(k/S)\).
Dividing by this factor removes an otherwise harmless \(O(k^2/S^2)\) bias
and leaves only the exponentially small aliasing terms controlled below.
\begin{lemma}[Single-root aliasing]
    \label[lemma]{lem:single-root-aliasing}
    Let \(k\) be a positive integer with \(k<S\), let \(\zeta\ne0\),
\(\alpha=\frac{\zeta}{R}\), and 
    assume \(|\alpha|\ne1\).
    For the single-root factor \((w-\zeta)\), let
\(q_{R,S}^{(k)}(\zeta)\) denote its contribution to
\(Q_{R,S}^{(k)}\).
    Finally, set
\[
    a
    \coloneqq
    \min\{|\alpha|,|\alpha|^{-1}\}<1.
\]
    If \(S\ge 2k\), then the sinc-corrected single-root error
    satisfies
    \begin{equation}
        \left|
        \zeta^k\mathbf{1}_{\{|\zeta|<R\}}
        -
        q_{R,S}^{(k)}(\zeta)
        \right|
        \le
        \frac{2k}{S}
        |\zeta|^k
        \frac{a^S}{1-a^S}.
        \label{eq:single-root-disc-bound}
    \end{equation}
\end{lemma}

\section{The SPSE Problem}\label{sec:spse}
In the previous section we introduced the SPICES algorithm and discussed how it retrieves single-particle energies in a purely free system. However, motivated by the physics of interacting systems, we would like for SPICES to be able to solve the more general problem wherein the free model is embedded in a larger system that may not be globally free.

Here we formalize the general problem of estimating free spectra of a free Hamiltonian embedded in an interacting system. 

\begin{definition}[Single-Particle Spectral Estimation]
    \label[definition]{def:spse-formal}
    \emph{Input.} Access, as in \cref{def:ff-operational-access}, to a Hamiltonian \(H\) and an occupation observable \(N\) on \(M\) modes; an energy window \(\Omega=[\Omega_{\min},\Omega_{\max}]\) of width \(L=\Omega_{\max}-\Omega_{\min}\); an inverse temperature \(\beta>0\); a time horizon \(t_{\max}\); an accuracy \(\epsilon>0\); and a failure probability \(\delta>0\).

    \emph{Promise.} There exist a unitary \(U\) and a partition of the modes \(\{1,\dots,M\}=I\sqcup O\) with \(I\neq\emptyset\) such that
    \begin{equation*}
        UHU^\dagger=\sum_{i\in I}E_i\,n_i+H_O,
        \qquad
        UNU^\dagger=\sum_{j=1}^{M}n_j ,
    \end{equation*}
    where \(H_O\) is supported on the modes \(O\) and commutes with \(N_O=\sum_{j\in O}n_j\) (so that \([H,N]=0\)). Moreover:
    \begin{enumerate}[label=(\roman*)]
        \item every protected energy lies in the window, \(E_i\in(\Omega_{\min},\Omega_{\max})\) for all \(i\in I\); and
        \item the spectator sector has no roots in the window: writing \(Z_O(s,z)=\Tr_O\{z^{N_O}\e^{-sH_O}\}\) and \(P_O(s,w)=w^{|O|}Z_O(s,-1/w)\), the polynomial \(P_O(\beta+\ii t,\cdot)\) has no root in the annulus \(A_\Omega\) of \cref{eq:window-annulus} for every \(t\in[0,t_{\max}]\).\end{enumerate}

    \emph{Output.} With probability at least \(1-\delta\), a finitely represented probability measure \(\widehat\nu\) with \(W_1(\widehat\nu,\nu_\Omega)\le\epsilon\), where
    \begin{equation*}
        \nu_\Omega=\frac{\sum_{i\in I}\e^{-\beta E_i}\,\delta_{E_i}}
        {\sum_{i\in I}\e^{-\beta E_i}} ,
    \end{equation*}
    counted with multiplicity, and \(W_1\) is the Wasserstein-1 distance.
\end{definition}
The requirement that the spectator roots lie outside of the energy annulus is included so that the problem has a unique solution, and it is also motivated by the feature of SPICES that roots outside of the annulus are not included in the signal.
When \(H_O\) is itself free, condition (ii) says exactly that no single-particle energy of \(H_O\) lies in \(\Omega\), and in this case it holds for all times \(t\).
When \(H_O\) is interacting there is no notion of single-particle energies for the spectator sector, and (ii) serves as the natural generalization of this free condition.

\paragraph{Input model.}
The input to SPSE is given in the following operational form, which is the form used by the algorithmic results and by the hardness results of \cref{sec:ff-hardness}.
The symbolic description of \(N\) matches the promise: since \(UNU^\dagger=\sum_{j}n_j\), we have \(N=WN_{\mathrm{std}}W^\dagger\) with \(W=U^\dagger\).
Fix a global polynomial \(p_\text{enc}\). Write \(\ell_{\mathrm{enc}}\) for the total
bit length of the instance description and \(\ell_{\mathrm{qry}}\) for the total bit length of the query parameters.

\begin{definition}[Public operational access to a number operator]
    \label[definition]{def:ff-operational-access}
    An operational description of \((M, H,N)\) consists of:
    \begin{enumerate}
        \item a parameter \(M\), specifying the total dimension \(2^M\);
        \item an explicit Pauli expansion of an \(M\)-qubit Hamiltonian \(H\);
        \item a finite symbolic definition of an exact Hermitian operator \(N\),
        \[
        N=W N_{\mathrm{std}}W^\dagger,
        \qquad
        N_{\mathrm{std}}=\sum_{j=1}^{M}n_j;
        \]
        \item a uniform classical compiler which, given a positive integer precision parameter \(r\) in unary and binary-encoded rational query parameters, outputs circuits for:
        \[
        \text{measurement of }N,\qquad \text{controlled }\e^{-\ii tH+\ii\vartheta N}, \qquad
        \rho_{\beta,\mu} = \frac{\e^{-\beta(H-\mu N)}} {\Tr\e^{-\beta(H-\mu N)}}.
        \]
    \end{enumerate}
    Here \(\beta\ge0\), \(\mu,t,\vartheta\in\mathbb R\); efficiently computable real parameters
    may instead be supplied to sufficient precision, with their approximation cost included.
    Preparation circuits may use measurements and classical control, with the stated bound applying
    to every execution. The state-preparation error is measured in trace distance and the channel error
    in diamond distance. For error at most \(2^{-r}\), the compiler's classical running time,
    ancilla count, and worst-case gate count are bounded by
\[
    \begin{cases}
        p_\text{enc}(M+\ell_{\mathrm{enc}}+r+\ell_{\mathrm{qry}}),
        &\text{for number measurement and joint evolution},\\
        p_\text{enc}(M+\ell_{\mathrm{enc}}+r+\ell_{\mathrm{qry}}+\lceil\e^\beta\rceil),
        &\text{for Gibbs preparation}.
    \end{cases}
\]
    Thus preparation may have exponential dependence on \(\beta\), but has
    polynomial cost whenever \(\beta=O(\log(M+\ell_{\mathrm{enc}}))\) and the
    remaining parameters are polynomially bounded.
\end{definition}

We first record that the promise factorizes the partition function.

\begin{lemma}[Factorization under the promise]
    \label[lemma]{lem:spse-factorization}
    Let \(U\) be a unitary and \(\{1,\dots,M\}=I\sqcup O\) a partition of the modes such that \(UHU^\dagger=\sum_{i\in I}E_in_i+H_O\) and \(UNU^\dagger=\sum_{j=1}^Mn_j\), where \(H_O\) is supported on the modes \(O\) and commutes with \(N_O=\sum_{j\in O}n_j\).
    Then for all \(s,z\in\mathbb C\),
    \begin{equation*}
        Z(s,z)=Z_I(s,z)Z_O(s,z),
        \qquad
        Z_I(s,z)=\prod_{i\in I}\bigl(1+z\e^{-sE_i}\bigr),
        \qquad
        Z_O(s,z)=\Tr_O\bigl\{z^{N_O}\e^{-sH_O}\bigr\},
    \end{equation*}
    and consequently \(P(s,w)=w^MZ(s,-1/w)=P_I(s,w)P_O(s,w)\) with \(P_I(s,w)=\prod_{i\in I}(w-\e^{-sE_i})\) and \(P_O(s,w)=w^{|O|}Z_O(s,-1/w)\).
\end{lemma}

\begin{proof}
    The partition function is invariant under unitary conjugation, so
    \[
        Z(s,z)=\Tr\bigl\{z^{UNU^\dagger}\e^{-sUHU^\dagger}\bigr\}.
    \]
    Write \(H_I=\sum_{i\in I}E_in_i\), supported on the modes \(I\), and \(N_I=\sum_{i\in I}n_i\).
    Since \(H_I\) and \(H_O\) act on different modes they commute, and \(N=N_I+N_O\) in the same way, so on \(\mathcal F_I\otimes\mathcal F_O\)
    \[
        z^{UNU^\dagger}\e^{-sUHU^\dagger}=\bigl(z^{N_I}\e^{-sH_I}\bigr)\otimes\bigl(z^{N_O}\e^{-sH_O}\bigr).
    \]
    Taking the trace gives \(Z=Z_IZ_O\), where \(Z_I=\Tr_I\{z^{N_I}\e^{-sH_I}\}\) is the free partition function of \cref{prop:free-factorization} for the modes in \(I\).
    The statement for \(P\) follows by setting \(z=-1/w\) and multiplying by \(w^M\).
\end{proof}

\begin{proposition}[The promise determines the protected roots]
    \label[proposition]{prop:promise-protected-roots}
    Suppose \((H,N,\Omega,\beta,t_{\max})\) satisfies the promise of \cref{def:spse-formal} with protected energies \(\{E_i\}_{i\in I}\), and write \(P(s,w)=w^MZ(s,-1/w)\).
    Then for every \(t\in[0,t_{\max}]\) and \(s=\beta+\ii t\), the roots of \(P(s,\cdot)\) in \(A_\Omega\) are precisely \(\zeta_i=\e^{-sE_i}\), \(i\in I\), counted with multiplicity.
    In particular, if \(t_{\max}\ge\pi T/L\) (the largest time sampled in \cref{thm:spices-recovery}), the first hypothesis of that theorem holds with \(d=|I|\), and the measure \(\nu_\Omega\) of \cref{def:spse-formal} is the measure \(\nu_\Omega\) of \cref{thm:spices-recovery}.
\end{proposition}

\begin{proof}
    By \cref{lem:spse-factorization}, the roots of \(P(s,\cdot)\) are the \(\zeta_i\) together with the roots of \(P_O(s,\cdot)\).
    By condition (i), \(|\zeta_i|=\e^{-\beta E_i}\in(\e^{-\beta\Omega_{\max}},\e^{-\beta\Omega_{\min}})\), so every \(\zeta_i\) lies in \(A_\Omega\), and by condition (ii) \(P_O(s,\cdot)\) has no root there.
    Finally, the weights \(\e^{-\beta(E_i-\Omega_{\min})}\) of \cref{thm:spices-recovery} differ from the weights \(\e^{-\beta E_i}\) of \cref{def:spse-formal} by the common factor \(\e^{\beta\Omega_{\min}}\), which cancels in the normalization.
\end{proof}

The promise does not imply the remaining four hypotheses of \cref{thm:spices-recovery} (radial separation, the signal floor, operational access and polynomial parameters).
The first two of these are the stability criteria of \cref{def:sc}; the last two concern the input model and the resource scaling.
\Cref{sec:perturbative-stability} verifies them for a class of spectator perturbations, and \cref{sec:ff-hardness} for a family of hard instances.

\section{Robustness to noise}\label{sec:noise}
In this section, we clarify the conditions under which SPSE is formally solved by SPICES.
This is captured in the following theorem.

\begin{theorem}[SPICES recovery from contour certificates]
    \label[theorem]{thm:spices-recovery}
    Let \([H, N]=0\), fix \(\beta > 0\), and let \(P(s,w) = w^M Z(s,-1/w)\) have degree \(M\) at the
    sampled times below. Fix a window \(\Omega = [\Omega_{\min},\Omega_{\max}]\) of width \(L > 0\), a requested
    accuracy \(0 < \epsilon < L/2\), and an inverse-polynomial failure probability \(0 <
    \delta < 1\). Suppose that known bounds \(0 < q < 1\) and \(0 < b \leq 1\) satisfy the following
    five conditions at every sampled time \(s_\ell = \beta + \ii t_\ell\):
    \begin{enumerate}
        \item There is a fixed, nonempty multiset of protected energies \(E_1, \dots, E_d \in (\Omega_{\min},
        \Omega_{\max})\) such that the roots of \(P(s_\ell,\cdot)\) in the annulus \(A_\Omega\) of
        \cref{eq:window-annulus} are precisely \(\e^{-s_\ell E_a}\), \(1\le a\le d\), counted with
        multiplicity.
        \item For either boundary energy \(\mu \in \{\Omega_{\min}, \Omega_{\max}\}\), set \(R_\mu =
        \e^{-\beta\mu}\). Every root \(\zeta\) of \(P(s_\ell,\cdot)\) satisfies
        \begin{equation}
            \label{eq:spices-radial-separation}
            \min \left\{\frac{|\zeta|}{R_\mu},\frac{R_\mu}{|\zeta|}\right\} \leq q < 1.
        \end{equation}
        \item Set \(\rho_{\beta, \mu} = \e^{-\beta(H - \mu N)}/Z(\beta, \e^{\beta\mu})\). At every
        quadrature node \(w_j = R_\mu\e^{\ii\theta_j}\), with \(\theta_j = 2\pi j/S\), the
        normalized endpoint signal satisfies
        \[
        \left|\Xi(s_\ell, -1/w_j)\right| =\left| \frac{Z(s_\ell, \e^{\beta \mu + \ii(\pi -
        \theta_j)})} {Z(\beta, \e^{\beta\mu})}\right| =\left|\Tr\!\left(\rho_{\beta, \mu}
        \e^{-\ii t_\ell H + \ii(\pi - \theta_j)N}\right)\right| \geq b.
        \]
        \item The Gibbs preparations \(\rho_{\beta, \mu}\) and controlled joint evolutions \(\e^{-\ii
        t_\ell H + \ii(\pi - \theta_j)N}\) are available to any required inverse-polynomial
        operational accuracy with polynomial overhead.
        \item The quantities \(L\), \(\epsilon^{-1}\), \(b^{-1}\), \((1-q)^{-1}\), and
        \(\e^{2\beta L}\) are bounded by polynomials in \(M\).
    \end{enumerate}
    Then SPICES, using root-moment order \(k=1\), returns a finitely represented
    probability measure \(\widehat\nu\) with
\[
    W_1(\widehat\nu, \nu_\Omega) \leq \epsilon, \qquad \nu_\Omega = \frac{\sum_{a =
    1}^d\e^{-\beta(E_a - \Omega_{\min})} \delta_{E_a}} {\sum_{a = 1}^d \e^{-\beta(E_a - \Omega_{\min})}},
\]
    with probability at least \(1-\delta\), using polynomially many gates and
    samples and polynomial classical postprocessing.
    
    Explicitly, use the LP reconstruction (\cref{def:chen-lp-estimator}) with
\[
    n = \left\lceil\frac{4L}{\epsilon}\right\rceil, \qquad
    T = \left\lceil\frac{24L}{\epsilon}\right\rceil, \qquad
    \kappa = c\frac{\epsilon} {L\log(4L/\epsilon)},
\]
    where \(c>0\) is a sufficiently small universal constant, and sample at \(t_\ell = \pi\ell/L\), \(0
    \leq \ell \leq T\), so that the largest sampled time is \(t_{\max}=\pi T/L=O(\epsilon^{-1})\). Set
\[
    \tau=\min\left\{\frac14, \frac{\kappa}{32(1+\e^{\beta L})}\right\}, \qquad
    S = \max\left\{
    \left\lceil\frac{4M}{1-q}\right\rceil,\, 2,\, 1+\left\lceil \frac{\log(1 +
    4M/\tau)}{\log(1/q)}\right\rceil
    \right\}.
\]
    Estimating each normalized endpoint signal to relative error at most \(\tau/4\) suffices. The total
    number of Hadamard-test shots is
    \begin{equation}
        \label{eq:spices-certified-shot-count}
        Q = O\!\left((T + 1)S\, b^{-2} \tau^{-2} \log\frac{S(T+1)}{\delta} \right).
    \end{equation}
\end{theorem}

The remainder of this section establishes \cref{thm:spices-recovery}.
The theorem converts five conditions into a Wasserstein guarantee with an explicit shot count: the protected roots in the annulus, radial separation of all roots from the two contour circles (the parameter \(q\)), a floor on the normalized signal (the parameter \(b\)), operational access, and polynomially bounded parameters.
The supporting statements address the three places where noise can break the pipeline.
First, phase unwrapping: the logarithmic increments of \(P\) are recovered from endpoint ratios only if these stay away from the branch cut (\cref{lem:one-factor-phase-bound,prop:centered-product-phase-bound,prop:phase-stability-sampling}).
Second, shot noise: relative error in the endpoint signals must neither change the branch assignment nor spoil the quadrature (\cref{lem:endpoint-noise-branch-stability,lem:shot-noise-quadrature-perturbation,lem:hadamard-test}).
Third, inversion: noisy Fourier data from the root moments must be turned into a measure that is close in Wasserstein distance (\cref{lem:normalized-fourier-data,prop:wasserstein-window-recovery}).
The proof of \cref{thm:spices-recovery} at the end of the section assembles these pieces.
Proofs of the supporting statements are collected in \cref{app:proofs}.

The quadrature rule in \cref{eq:midpoint-quadrature-rule} assumes access to the logarithmic increments of \(P\) along the circle.
The exact contour integral only requires that no zero of \(P\) lie on the contour.
To recover the same increments from sampled endpoint values, however, the endpoint ratios must stay separated from the branch cut of the principal logarithm.
The statements below give simple sufficient conditions for that separation.
Throughout, \(\operatorname{Arg}\) denotes the principal argument in
\((-\pi,\pi]\).

For \(h=2\pi/S\), define the exact endpoint ratio
\begin{equation}
    \Gamma_j(R,s)
    \coloneqq
    \frac{P(w_{j+1})}{P(w_j)}
    =
    \prod_{\ell=1}^{M}
    \frac{\e^{\ii h}-\e^{-\ii\theta_j}\zeta_\ell/R}
    {1-\e^{-\ii\theta_j}\zeta_\ell/R},
    \qquad
    j=0,\ldots,S-1,
    \label{eq:endpoint-ratio-gamma}
\end{equation}
where \(w_j=R\e^{\ii\theta_j}\).
The relevant margin from the logarithm branch cut is the \textit{phase margin}
\begin{equation}
    \chi_S(R,s)
    \coloneqq
    \min_{0\le j<S}
    \left\{
    \pi-\left|\operatorname{Arg}\Gamma_j(R,s)\right|
    \right\}.
    \label{eq:branch-cut-margin}
\end{equation}
Thus \(\chi_S(R,s)>0\) means that every exact endpoint ratio avoids the cut
at argument \(\pm\pi\) by a uniform margin.

For a nonvanishing piecewise \(C^1\) path \(u:[0,h]\to\C \setminus \{0\}\), define its lifted argument increment by
\[
\Delta\arg u
\coloneqq
\Im
\int_0^h \frac{u'(\theta)}{u(\theta)}\,\dd\theta .
\]
Then \(u(h)/u(0)\) has argument \(\Delta\arg u\) modulo \(2\pi\).
Indeed, if \(u(\theta)=r(\theta)\e^{\ii\varphi(\theta)}\) with \(\varphi\) a continuous lift of the phase, then
\[
\frac{u'(\theta)}{u(\theta)}
=
\frac{r'(\theta)}{r(\theta)}+\ii\varphi'(\theta),
\]
so the imaginary part of the logarithmic derivative integrates to \(\varphi(h)-\varphi(0)\).
The endpoint ratio only determines this increment modulo \(2\pi\), while the integral records the continuous representative along the path.

We will apply the following estimates with
\[
r_\ell=\frac{|\zeta_\ell|}{R},
\qquad
\phi_{\ell,j}=\arg(\zeta_\ell)-\theta_j,
\]
so that each factor in \(\Gamma_j(R,s)\) has the form
\begin{equation} \label{eq:ratio_factor}
    \frac{\e^{\ii h}-\e^{-\ii\theta_j}\zeta_\ell/R}
    {1-\e^{-\ii\theta_j}\zeta_\ell/R} = \frac{\e^{\ii h}-r_\ell\e^{\ii\phi_{\ell,j}}}
    {1-r_\ell\e^{\ii\phi_{\ell,j}}}.
\end{equation}
The bounds are uniform in the phases, and hence uniform in the sample-point angle \(\theta_j\) and in the arguments of the roots.

\begin{lemma}[One-factor bound]
    \label[lemma]{lem:one-factor-phase-bound}
    Let \(0<h<\pi\), \(r\ge0\), \(r\ne1\), and \(\phi\in\mathbb R\).
    Let
\[
    u(\theta)=\e^{\ii\theta}-r\e^{\ii\phi},
    \qquad
    0\le\theta\le h.
\]
    Let \(\rho = \min\bigl\{r,r^{-1}\bigr\}\).
    Then
    \begin{equation}
        \left|\Delta\arg u-h\mathbf 1_{\{r<1\}}\right|
        \le
        h\,\frac{\rho}{1-\rho}.
        \label{eq:centered-single-factor-bound}
    \end{equation}
    Moreover,
\[
    \frac{\e^{\ii h}-r\e^{\ii\phi}}{1-r\e^{\ii\phi}}
\]
    has argument \(\Delta\arg u\) modulo \(2\pi\).
\end{lemma}

The lemma treats a single root.
A root inside the contour (\(r<1\)) advances the phase by nearly the full step \(h\), a root outside (\(r>1\)) barely moves it, and in both cases the deviation from these two values is at most \(h\rho/(1-\rho)\), which is small when the root is radially far from the circle.
The next statement multiplies the factors for all roots, so that only two numbers survive: the count \(K\) of roots inside the contour and the accumulated deviation \(B\).

\begin{proposition}[Product bound]
    \label[proposition]{prop:centered-product-phase-bound}
    Let \(0<h<\pi\), let \(r_1,\ldots,r_m\ge0\) with \(r_j\ne1\), and define
\(\rho_j = \min\bigl\{r^{\vphantom{-1}}_j,r^{-1}_j\bigr\}\). 
    Define also
    \begin{align} \label{eq:kbdef}
        K=\#\{j:r_j<1\},
        \qquad
        B=\sum_{j=1}^{m}\frac{\rho_j}{1-\rho_j}.
    \end{align}
    Then, uniformly over all phases \(\phi_1,\ldots,\phi_m\in\mathbb R\),
    \begin{equation}
        \pi-
        \left|
        \operatorname{Arg}
        \prod_{j=1}^{m}
        \frac{\e^{\ii h}-r_j\e^{\ii\phi_j}}
        {1-r_j\e^{\ii\phi_j}}
        \right|
        \ge
        \operatorname{dist}\!\left(Kh,(2\mathbb Z+1)\pi\right)-hB .
        \label{eq:centered-product-phase-bound}
    \end{equation}
    Here \(\operatorname{dist}(x,(2\mathbb Z+1)\pi)\) denotes the distance from \(x\) to any odd multiple of \(\pi\).
\end{proposition}

The bound \cref{eq:centered-product-phase-bound} is informative when \(Kh\) stays away from the odd multiples of \(\pi\) and \(hB\) is small.
For the free roots \(\zeta_j=\e^{-sE_j}\) and a contour of radius \(R=\e^{-\beta\mu}\) we have \(r_j=\e^{-\beta(E_j-\mu)}\), so \(K\) counts the energies above \(\mu\) and \(B\) is a sum of Boltzmann-type ratios that does not depend on \(t\).
The next proposition turns this into a sampling condition: taking enough quadrature points \(S\) (that is, a small enough step \(h=2\pi/S\)) secures a constant phase margin uniformly in time.

\begin{proposition}[Sufficient sampling for noiseless phase stability]
    \label[proposition]{prop:phase-stability-sampling}
    Let \(E_1,\ldots,E_M\in\mathbb R\), fix \(\mu\in\mathbb R\), and let \(\Delta(\mu)=\min_j |E_j-\mu|>0\).
    Let \(\beta>0\), set \(R=\e^{-\beta\mu}\), and let \(S>2\) be an integer.
    Set
\[
    K=\#\{j:E_j>\mu\},
    \qquad
    B
    =
    \sum_{j=1}^{M}
    \frac{\e^{-\beta|E_j-\mu|}}{
    1-\e^{-\beta|E_j-\mu|}
    }.
\]
    If
\[
    S
    \ge
    4(K+B),
\]
    then the phase margin \(\chi_S(R,s)\) satisfies \(\chi_S\ge\pi/2\), uniformly in \(t\).
\end{proposition}
The statements so far concern exact endpoint ratios.
The next two lemmas show that the margin survives, and that the quadrature remains accurate, when each endpoint signal is only known up to a relative shot-noise error \(\varsigma\).

\begin{lemma}[Phase margin stability under shot noise]
    \label[lemma]{lem:endpoint-noise-branch-stability}
    Suppose the exact endpoint ratios \(\Gamma_j(R,s)\) have branch-cut margin
    at least \(\chi_S>0\).
    Let \(\Xi_j=\Xi(s,-1/w_j)\), and suppose the endpoint estimates obey
\[
    \widehat \Xi_j=(1+\varsigma_j)\Xi_j,
    \qquad |\varsigma_j|\le\varsigma .
\]
    If
\[
    \varsigma<\sin\left(\frac{\chi_S}{2}\right),
\]
    then the noisy endpoint ratios preserve the same logarithm branch assignment
    as the exact endpoint ratios.
\end{lemma}

The first lemma says that a relative error below \(\sin(\chi_S/2)\) cannot push an endpoint ratio across the branch cut, so the noisy data are assigned the same logarithm branches as the exact data.
With the branches fixed, the second lemma shows that the quadrature error is linear in \(\varsigma\), with no dependence on the number of roots or on the radial separation.

\begin{lemma}[Shot-noise perturbation of the quadrature]
    \label[lemma]{lem:shot-noise-quadrature-perturbation}
    Assume \(0<k<S\), \(\chi_S(R,s)\ge\pi/2\), and
\(0<\varsigma<1/\sqrt2\).
    If \(|\varsigma_j|\le\varsigma\) for all \(j\), then
\[
    \left|
    \widehat Q_{R,S}^{(k)}(s)-Q_{R,S}^{(k)}(s)
    \right|
    \le
    2kR^k\varsigma .
\]
\end{lemma}

It remains to convert the relative error \(\varsigma\) assumed in these lemmas into a number of measurements.
Each endpoint signal is the expectation of a unitary in a state, and so it can be estimated with the Hadamard test.

Given a preparation of any state \(\rho\) and an implementation of a controlled unitary
\[
\Lambda_{\mathcal{J}} = \lvert0\rangle\!\langle0\rvert\otimes\I
+
\lvert1\rangle\!\langle1\rvert\otimes \mathcal{J},
\]
we can prepare an ancilla in the state \(|+\rangle = \tfrac{1}{\sqrt{2}}(|0\rangle + |1\rangle)\) and evolve to \(\Lambda_{\mathcal{J}} (|+\rangle\!\langle+|\otimes \rho) \Lambda_{\mathcal{J}}^\dagger\).
After the controlled unitary, the reduced state of the ancilla is
\[
\rho_a
=
\frac12
\begin{bmatrix}
    1&\Tr(\rho \mathcal{J})^*
    \\
    \Tr(\rho \mathcal{J})&1
\end{bmatrix}.
\]
Consequently
\[
\mathbb E[X]=\Re\bigl(\Tr(\rho \mathcal{J})\bigr),
\qquad
\mathbb E[Y]=\Im\bigl(\Tr(\rho \mathcal{J})\bigr),
\]
or equivalently \(\mathbb E[X + \ii Y ]= \Tr(\rho \mathcal{J})\).
One run of the Hadamard test quantum circuit produces an independent \(\pm1\) random variable with probabilities \((1 \pm \mathbb E[X])/2\) for an \(X\) measurement and similarly for a \(Y\) measurement.

A standard application of Hoeffding's inequality gives a quantification of the sample complexity of estimating \(\Tr(\rho \mathcal{J})\) using the Hadamard test and the obvious empirical mean estimator.

\begin{lemma}[Relative precision from Hadamard-test samples]
    \label[lemma]{lem:hadamard-test}
    Let \(\rho\) be a quantum state, \(\mathcal{J}\) be a unitary, and let \(\xi \coloneqq \Tr(\rho \mathcal{J})\), with \(\xi\ne0\). 
    If \(\varsigma, \delta \in (0,1)\), then \(2n_{\mathrm{sh}}\) independent Hadamard tests suffice to obtain an estimate \(\widehat \xi\) such that
\[\Pr\left\{
    |\widehat \xi-\xi|\le \varsigma |\xi|
    \right\}
    \ge
    1-\delta 
\]
    whenever 
    \begin{equation}
        n_{\mathrm{sh}}
        \ge
        \frac{4}{\varsigma^2|\xi|^2}
        \log\frac{4}{\delta}.
        \label{eq:hadamard-relative-shots}
    \end{equation}
\end{lemma}
The last ingredient is the inversion step, which turns estimates of the root moments into a measure that is close in Wasserstein distance to the thermal measure on the window.
We use the linear-programming estimator of Chen et al.~\cite{ChenEtAl2025}, which has two parts.
The Jackson kernel smooths a measure at a controlled Wasserstein cost, so that a finite number of Fourier coefficients suffice to determine it up to that cost.
The linear program then searches for a measure whose smoothed Fourier coefficients match the noisy data.

\begin{definition}[Jackson kernel]
    \label[definition]{def:jackson-kernel}
    For an integer \(n\ge1\), define the Jackson kernel on
\(\mathbb T=[0,1)\) by
\[
    J_n(x)
    =
    \frac{3}{n(2n^2+1)}
    \frac{\sin^4(\pi n x)}{\sin^4(\pi x)},
\]
    with its value at \(x=0\) defined by continuity.
    Its Fourier multiplier is the explicit even function
\[
    \widetilde{J}_{n}(\ell)
    =
    \frac{3}{n(2n^2+1)}
    \begin{cases}
        \displaystyle
        \binom{2n+1-|\ell|}{3}
        -4\binom{n+1-|\ell|}{3},
        & 0\le |\ell|\le n-2, \\[4pt]
        \displaystyle
        \binom{2n+1-|\ell|}{3},
        & n-1\le |\ell|\le 2n-2, \\[4pt]
        0, & |\ell|\ge 2n-1.
    \end{cases}
    \label{eq:jackson-multiplier}
\]
    Thus \(\widetilde{J}_{n}(\ell)=\int_{\mathbb T}\e^{2\pi\ii\ell x}J_n(x)\dd x\).
    The kernel is nonnegative and normalized, and it satisfies
\[
    |\widetilde{J}_{n}(\ell)|\le1,
    \qquad
    W_1^{\mathbb T}(\rho,\rho*J_n)\le\frac{1}{n}
\]
    for every probability measure \(\rho\) on \(\mathbb T\).
\end{definition}

The estimator below uses the multiplier \(\widetilde J_n\) to damp the noisy Fourier data to the scale of the smoothed target, and then looks for a piecewise-constant density whose Fourier coefficients agree with the damped data to within the tolerance \(\kappa\).

\begin{definition}[The LP estimator~\cite{ChenEtAl2025}]
    \label[definition]{def:chen-lp-estimator}
    Given normalized data \(\widehat u_\ell\) for \(|\ell|\le T\) and a
    tolerance \(\kappa>0\), choose an integer \(n\ge1\) and set
\[
    K=\left\lceil C\frac{n^2}{\kappa}\right\rceil,
\]
    where \(C\) is a sufficiently large universal constant.
    Let \(J_n\) and \(\widetilde{J}_{n}\) be the kernel and Fourier multiplier from
    \cref{def:jackson-kernel}.
    
    For \(I_j=[j/K,(j+1)/K)\), write
\[
    c_{j,\ell}
    =
    \int_{I_j}\e^{2\pi\ii\ell x}\dd x .
\]
    If a probability measure has density \(Ka_j\) on \(I_j\), then its
\(\ell\)th Fourier coefficient is exactly
\[
    K\sum_{j=0}^{K-1}c_{j,\ell}a_j.
\]
    The LP estimator is the probability measure \(\sigma\) with density
\(Ka_j\) on \(I_j\), where \(a_0,\ldots,a_{K-1}\) satisfies
    \begin{align}
        &a_j\ge0,
        \qquad
        \sum_{j=0}^{K-1}a_j=1,
        \label{eq:chen-lp-simplex}
        \\
        &\left|
        K\sum_{j=0}^{K-1}\Re(c_{j,\ell})a_j
        -
        \Re\bigl(\widetilde{J}_{n}(\ell)\widehat u_\ell\bigr)
        \right|
        \le
        \frac{\kappa}{4},
        \qquad
        |\ell|\le T,
        \label{eq:chen-lp-real}
        \\
        &\left|
        K\sum_{j=0}^{K-1}\Im(c_{j,\ell})a_j
        -
        \Im\bigl(\widetilde{J}_{n}(\ell)\widehat u_\ell\bigr)
        \right|
        \le
        \frac{\kappa}{4},
        \qquad
        |\ell|\le T.
        \label{eq:chen-lp-imag}
    \end{align}
\end{definition}

We now specify the data to which this estimator is applied, namely the Fourier coefficients of the thermal measure on the window, which are the normalized root moments.
Fix a positive integer \(k\).
Because the energy origin is arbitrary, it is useful to shift the window moments by its lower endpoint.
For the rest of this section, write
\[
\mathcal{M}_\Omega^{(k)}(s)
\coloneqq
\e^{ks \Omega_{\min}}\mathcal{M}_{A_\Omega}^{(k)}(s)
=
\sum_{E_j\in\Omega}\e^{-ks(E_j- \Omega_{\min})},
\qquad
\mnorm
=
\mathcal{M}_\Omega^{(k)}(\beta).
\]
This shift is a deterministic rescaling of the annular root moment and is invariant under a common translation of the energies and window endpoints.
Assume that \(\Omega\) contains at least one energy.
Since \(0<E_j-\Omega_{\min}<L\) for every such energy,
\begin{equation}
    \mnorm
    =
    \sum_{E_j\in\Omega}\e^{-\beta k(E_j-\Omega_{\min})}
    \ge
    \e^{-\beta kL}.
    \label{eq:shifted-mass-lower-bound}
\end{equation}
The associated thermally weighted probability measure is
\begin{align*}
    \nu(E)
    &=
    \frac{1}{\mnorm}
    \sum_{E_j\in\Omega}
    \e^{-\beta k(E_j-\Omega_{\min})}\delta(E-E_j). \\
    &= \frac{1}{\mathcal{M}_{A_{\Omega}}^{(k)}(\beta)}
    \sum_{E_j\in\Omega}
    \e^{-\beta k E_j}\delta(E-E_j).
\end{align*}
The common shift cancels between the numerator and denominator, so this is the same probability measure obtained from the unshifted Gibbs weights.
For \(k=1\) this is the Gibbs distribution restricted to the window, and is also the Fourier transform (with a suitable convention and normalized) of the root moment for the window.
Larger \(k\) simply changes the effective inverse temperature to \(k\beta\).

Through an affine transformation, we embed the energy band \(\Omega\) into half of the unit torus,
\[
\phi(E)=\frac{E-\Omega_{\min}}{2L}\in[0,1/2].
\]
The factor of \(2\) keeps the support away from the seam of the torus.
Let \(f_\#\rho\) denote the pushforward of a measure \(\rho\) by a map \(f\): it is the distribution of \(f(X)\) when \(X\) has distribution \(\rho\).
Explicitly, the new measure is simply
\[
\phi_\#\nu(x)
=
\frac1\mnorm\sum_{E_j\in\Omega}
\e^{-\beta k(E_j-\Omega_{\min})}
\delta\Bigl(x-\tfrac{E_j-\Omega_{\min}}{2L}\Bigr).
\]
For \(\ell\in\mathbb Z\), let \(u_\ell\) denote the \(\ell\)th Fourier coefficient of \(\phi_\#\nu\), with the convention of \cite{ChenEtAl2025}:
\[
u_\ell
=
\int \e^{2\pi\ii\ell x}\dd(\phi_\#\nu)(x).
\]
Let
\[
t_\ell=\frac{\pi\ell}{kL},
\qquad
s_\ell=\beta+\ii t_\ell,
\qquad
\Mell = \mathcal{M}_\Omega^{(k)}(s_\ell),
\qquad
0\le\ell\le T,
\]
where \(T\) is a positive integer cutoff fixed in \cref{prop:wasserstein-window-recovery}.
With this Fourier convention, the coefficients with \textit{negative} frequency are obtained from the root moments via
\begin{equation}
    u_{-\ell}
    =
    \frac{\Mell}{\mnorm}.
    \label{eq:window-fourier-from-root-moments}
\end{equation}
Because \(\phi_\#\nu\) is a real measure, \(u_\ell\) is the complex conjugate of \(u_{-\ell}\).
Thus measurements at \(t\ge0\), followed by complex conjugation in post-processing supply the full set of coefficients \(|\ell|\le T\).
Note also that
\[
u_0
=
\frac1\mnorm \mathcal{M}_0=
\frac1\mnorm \mathcal{M}_\Omega^{(k)}(\beta)
=
1
\]
as required for the Fourier transform of a normalized positive measure.
As in the previous sections, we denote the estimator for \(\mathcal{M}_{\ell}\) by \(\widehat{\mathcal{M}}_{\ell}\) and the estimator for \(u_{\ell}\) by \(\widehat{u}_{\ell}\).
In particular, set \(\mnormhat=\widehat{\mathcal{M}}_0\), and define
\[
\widehat u_{-\ell}
=
\frac{\widehat{\mathcal{M}}_\ell}{\mnormhat}
\qquad
(0\le\ell\le T),
\]
with \(\widehat u_\ell=\bigl(\widehat u_{-\ell}\bigr)^*\) for \(\ell>0\).

\begin{lemma}[Normalization of noisy Fourier data]
    \label[lemma]{lem:normalized-fourier-data}
    Let \(0<\kappa\le1\), and suppose that for each integer \(\ell\) in \([0,T]\) we have \[
    \bigl|\widehat{\mathcal{M}}_\ell-\Mell\bigr|\le \frac{\kappa}{32}\e^{-\beta kL}. \]
    Then
\[
    \left|
    \widehat u_\ell-u_\ell
    \right|
    \le
    \frac{\kappa}{8}
    \qquad
    (|\ell|\le T).
\]
\end{lemma}

The lemma says that root-moment errors of size \(\tfrac{\kappa}{32}\e^{-\beta kL}\) produce Fourier data accurate to \(\kappa/8\) after normalization; the factor \(\e^{-\beta kL}\) is the lower bound \cref{eq:shifted-mass-lower-bound} on the normalization.
The next proposition feeds such data to the LP estimator and converts the Fourier accuracy into a Wasserstein guarantee.

\begin{proposition}[Wasserstein recovery on an energy window]
    \label[proposition]{prop:wasserstein-window-recovery}
    Fix \(0<\epsilon<L/2\), and set
\[
    n=\left\lceil\frac{4L}{\epsilon}\right\rceil,
    \qquad
    T=\left\lceil\frac{24L}{\epsilon}\right\rceil,
    \qquad
    \kappa
    =
    c\,
    \frac{\epsilon}{L\log(4L/\epsilon)},
\]
    with \(c>0\) a sufficiently small universal constant.
    Assume that \(\Omega\) contains at least one energy, and compute \(\widehat u_\ell\) from moment estimates that satisfy \cref{lem:normalized-fourier-data}.
    If \(\sigma\) is the output of the LP estimator of \cref{app:lp-estimator} with these parameters, then
\(\widehat\nu=\psi_\#\sigma\), where
\[
    \psi(x)=\Omega_{\min}+2L\min\{x,1-x\},
    \qquad x\in[0,1),
\]
    satisfies
\[
    W_1\!\left(\nu,\widehat\nu\right)
    \le
    \epsilon.
\]
\end{proposition}

\begin{proof}[Proof of \cref{thm:spices-recovery}]
    \emph{Wasserstein recovery error.}
    Our objective is to prove \(W_1(\widehat\nu,\nu_\Omega)\leq\epsilon\).
    We will apply \cref{prop:wasserstein-window-recovery} with \(k=1\) and the accuracy \(\epsilon\) of the theorem. Its accuracy range and choices of \(n,T,\kappa\)
    are exactly those in the theorem statement, and the first hypothesis ensures that the
    window contains at least one protected energy. We choose the universal constant \(c\)
    small enough that \(0<\kappa\leq1\). The remaining requirement is to supply normalized
    Fourier data obtained from moment estimates satisfying \cref{lem:normalized-fourier-data}.
    The proposition then gives the desired estimator by taking the LP output \(\sigma\) and
    setting \(\widehat\nu=\psi_\#\sigma\), where
\(\psi(x)=\Omega_{\min}+2L\min\{x,1-x\}\).
    
    To meet that requirement, \cref{lem:normalized-fourier-data} reduces the task to producing
    shifted moment estimates with
\[
    |\widehat{\mathcal{M}}_\ell-\mathcal{M}_\Omega^{(1)}(s_\ell)|
    \leq \frac{\kappa}{32}\e^{-\beta L},
    \qquad 0\leq\ell\leq T.
\]
    We first check that these contour moments describe the measure \(\nu_\Omega\) to be
    recovered and have the positive normalization mass used by that lemma. By the first
    hypothesis and the power-sum formula \cref{prop:power-sum-formula},
\[
    \mathcal{M}_\Omega^{(1)}(s_\ell)=\e^{s_\ell \Omega_{\min}}\mathcal{M}_{A_\Omega}^{(1)}(s_\ell)
    =\sum_{a=1}^d\e^{-s_\ell(E_a-\Omega_{\min})}.
\]
    In particular, \(\mathcal{M}_\Omega^{(1)}(\beta)\geq\e^{-\beta L}\), since \(d\geq1\).
    At \(t_\ell=\pi\ell/L\), dividing these moments by \(\mathcal{M}_\Omega^{(1)}(\beta)\) gives the
    negative-frequency Fourier coefficients of \(\nu_\Omega\) under the rescaling
\(E\mapsto(E-\Omega_{\min})/(2L)\); the positive-frequency coefficients follow by complex
    conjugation. Thus the lemma applies to the normalized estimates
\(\widehat u_{-\ell}=\widehat{\mathcal{M}}_\ell/\widehat{\mathcal{M}}_0\) once the displayed moment-error bound
    is established, and gives Fourier error at most \(\kappa/8\).
    
    To obtain that moment accuracy, it suffices to estimate the disc moment on each boundary
    circle to error at most \(R\tau\). Indeed, subtracting the two circle estimates and
    multiplying by \(\e^{s_\ell \Omega_{\min}}\) then gives
\[
    |\widehat{\mathcal{M}}_\ell-\mathcal{M}_\Omega^{(1)}(s_\ell)|
    \leq (1+\e^{-\beta L})\tau
    \leq \frac{\kappa}{32}\e^{-\beta L},
\]
    where the last inequality is the reason for the stated choice of \(\tau\).
    It remains to justify the circle error \(R\tau\). The next paragraph proves it provided
    all endpoint signals are estimated to relative error at most \(\tau/4\); the final
    paragraph supplies that precision simultaneously with probability at least \(1-\delta\)
    and bounds the resources. On that event, \cref{lem:normalized-fourier-data} and then
    \cref{prop:wasserstein-window-recovery} yield the claimed Wasserstein error.
    
    \emph{Phase stability and quadrature.}
    At either boundary circle, the quantities \(K,B\) of \cref{prop:centered-product-phase-bound},
    applied to the actual roots, obey
\[
    K \leq M, \qquad B \leq \frac{Mq}{1-q}, \qquad K+B \leq \frac{M}{1-q}.
\]
    Hence the chosen \(S\) makes every lifted argument increment have modulus at most \(\pi/2\), as in
    the proof of \cref{prop:phase-stability-sampling}. In particular, \(\chi_S \geq \pi/2\). The
    single-root aliasing bound \cref{lem:single-root-aliasing} depends only on radial separation, so
    summing it over the roots gives
\[
    |\mathcal{M}_R^{(1)} - Q_{R,S}^{(1)}| \leq \frac{2R M}{S} \frac{q^{S-1}}{1-q^S} \leq \frac{R\tau}{2}.
\]
    For the last inequality, the third term in the definition of \(S\) gives \((S-1)\log(1/q)\ge\log(1+4M/\tau)\), that is, \(q^{S-1} \leq (1+4M/\tau)^{-1}\le\tau/(4M)\); and \(S\ge 4M/(1-q)\) gives \(S(1-q^S)\ge S(1-q)\ge1\). By
    \cref{lem:endpoint-noise-branch-stability,lem:shot-noise-quadrature-perturbation}, endpoint relative
    error \(\tau/4\) preserves the logarithm branches and contributes at most another \(R\tau/2\). Each
    circle therefore has total moment error at most \(R\tau\).
    
    \emph{Resources and operational accuracy.}
    The signal floor \(b\) and \cref{lem:hadamard-test} give \(O(b^{-2} \tau^{-2} \log(S(T+1)/\delta))\)
    shots per endpoint. A union bound over both circles and all \(T+1\) times proves
    \cref{eq:spices-certified-shot-count}. For approximate primitives, choose their operational errors
    small enough that the endpoint bias is at most \(b\tau/8\), and allocate the other \(b\tau/8\) to
    statistical error. This preserves relative error \(\tau/4\) and changes the shot bound only by a
    constant factor. The fifth hypothesis makes \(T\), \(S\), \(\tau^{-1}\), and \(Q\) polynomial. The
    largest sampled time is \(\pi T/L=O(\epsilon^{-1})\), so the fourth hypothesis
    supplies the required circuits with polynomial overhead. The LP reconstruction also has polynomial
    size at these parameters.
\end{proof}

\section{Sample Complexity for Quadrature and Phase Unwrapping Error 
}\label{sec:free-gapped}
The recovery theorem takes the signal floor \(b\) and the radial separation \(q\) as hypotheses. For a free
Hamiltonian both are automatic at a contour that is separated from the spectrum, \(\beta\Delta(\mu)\ge\log(3M)\):
then every root satisfies the radial separation with \(q=\e^{-\beta\Delta(\mu)}\le 1/(3M)\), and the normalized signal obeys \(|\Xi|\ge\e^{-1}\) (\cref{cor:log-inv-temp-suffices}).
This section makes the resulting resource count explicit for one contour, with all constants.
As in the previous section, proofs of the supporting statements are in \cref{app:proofs}.

Three conditions on the number of quadrature points \(S\) enter below.
The quadrature aliasing error of \cref{prop:quadrature-gap-bound} requires \(S \ge 2k\);
the phase margin of \cref{cor:phase-margin-guarantee} requires \(S \ge 4M+2\);
and a root-moment error of at most \(kR^k\tau\) on one circle, obtained from \cref{prop:fixed-root-moment-estimate} with endpoint precision \(\tau/4\), requires \(S-k \ge \log(1+4M/\tau)/\log(3M)\).
Here \(\tau\) is the relative accuracy of the root moment on one circle; it is the quantity denoted \(\tau\) in \cref{thm:spices-recovery}, whereas \(\varsigma\) denotes the relative error of a single endpoint estimate.
For convenience we define
\begin{equation}
    \label{eq:big-enough-S}
    S_0 \coloneqq S_0(M,k,\tau) = \max\biggl\{ 4M+2,\, 2k,\, k+\left\lceil\frac{\log(1+4M/\tau)}{\log(3M)}\right\rceil\biggr\}.
\end{equation}
In practice we are only interested in regimes where \(k = O(1)\) and
\(1/\tau = O(\operatorname{poly}(M))\), so the first constraint (\(S \ge 4M+2\)) is binding.
The definition of \(S_0\) assumes \(\beta\Delta(\mu)\ge \log(3M)\), so \(S \ge S_0\) is only sufficient in that case.
For \(q\le 1/(3M)\) one may also use \(S=S_0(M,1,\tau)\) in \cref{thm:spices-recovery}.

\begin{proposition}[Hadamard-test shots for one circle]
    \label[proposition]{prop:circle-shot-count}
    Let \(0<\delta<1\) and \(0<\tau<1/\sqrt2\).
    Assume \(\beta\Delta(\mu)\ge\log(3M)\) and choose \(S\ge S_0(M,k,\tau) \) via \cref{eq:big-enough-S}.
    If each endpoint \(\Xi_j\) with \(j = 0,\ldots,S-1\) is estimated using \(2n\) Hadamard tests (so \(2nS\) tests for the whole circle), then
\[
    n
    =
    O\biggl(\frac{1}{\tau^2}
    \log\frac{S}{\delta}\biggr)
\]
    suffices so that with probability at least \(1-\delta\)
\[
    \left|
    \mathcal{M}_R^{(k)}(s)-\widehat Q_{R,S}^{(k)}(s)
    \right|
    \le
    kR^k\tau .
\]
\end{proposition}

The proof combines the following supporting statements.
The first three bound the normalized signal from below, the next two control the quadrature error and the phase margin, and the last combines them with the shot-noise lemmas of \cref{sec:noise}.
We begin with an elementary inequality, which turns a product of single-mode factors into a single exponential.

\begin{lemma}
    \label[lemma]{lem:logarithmic-tanh-bound}
    For \(x\ge \frac12\log 2\), 
    \begin{equation}
        -\log\tanh x \le 3\e^{-2x}.
        \label{eq:log-tanh-simple}
    \end{equation}
\end{lemma}

The normalized signal is a product over modes, and each mode contributes a factor that is close to one when the mode is far from the contour.
The following proposition bounds the product from below; the second bound uses the inequality above.

\begin{proposition}[Gap lower bound for the normalized partition function]
    \label[proposition]{prop:gap-lower-bound}
    Fix \(\mu\) with \(\Delta(\mu)>0\), and recall 
\(s=\beta+\ii t\) and \(z=\e^{\beta\mu+\ii\vartheta}\).
    Then uniformly in \(t\) and \(\vartheta\) we have
    \begin{equation}
        |\Xi(s,z)|
        \ge
        \prod_{j=1}^{M}
        \tanh\left(\frac{\beta |E_j-\mu|}{2}\right)
        \ge
        \tanh\left(\frac{\beta\Delta(\mu)}{2}\right)^M .
        \label{eq:Xi0-gap-product-bound}
    \end{equation}
    Moreover, if \(\beta\Delta(\mu)\ge \log 2\), then
    \begin{equation}
        |\Xi(s,z)|
        \ge
        \exp\left(-3M\e^{-\beta\Delta(\mu)}\right).
        \label{eq:Xi0-gap-exponential-bound}
    \end{equation}
\end{proposition}

At a logarithmic inverse temperature the exponent \(3M\e^{-\beta\Delta(\mu)}\) is at most one, so the signal floor is a constant.

\begin{corollary}[Logarithmic inverse temperature suffices]
    \label[corollary]{cor:log-inv-temp-suffices}
    If \(\Delta(\mu)>0\) and
    \begin{equation}
        \beta
        \ge
        \frac{1}{\Delta(\mu)}\log(3M),
    \end{equation}
    then, uniformly in \(t\) and \(\vartheta\),
    \begin{equation}
        \lvert\Xi(\beta+\ii t,\e^{\beta\mu+\ii\vartheta})\rvert
        \ge
        \e^{-1}.
    \end{equation}
\end{corollary}

We now turn to the quadrature.
The same gap controls the aliasing error of the circular rule, which decays geometrically in the number of quadrature points \(S\) at rate \(q=\e^{-\beta\Delta(\mu)}\).
The statement is given for a single circle and for an annulus.

\begin{proposition}[Disc and annulus quadrature error]
    \label[proposition]{prop:quadrature-gap-bound}
    Let \(R=\e^{-\beta\mu}\), assume \(\Delta(\mu)>0\), and write
\[
    q=\e^{-\beta\Delta(\mu)}.
\]
    Then for all \(S\ge2k\) with \(k\) a positive integer, 
    \begin{equation}
        \left|
        \mathcal{M}_R^{(k)}(s)-Q_{R,S}^{(k)}(s)
        \right|
        \le
        \frac{2k}{S}R^k
        \frac{M q^{S-k}}{1-q^S}.
        \label{eq:disc-gap-bound}
    \end{equation}
    Similarly, for an annulus \(A\) with \(\Delta_A>0\), define
\(q_A=\e^{-\beta\Delta_A}\) and
\(Q_{A,S}^{(k)}(s) \coloneqq Q_{R_{\rm out},S}^{(k)}-Q_{R_{\rm in},S}^{(k)}\).
    Then
    \begin{align}
        \left|
        \mathcal{M}_{A(R_{\mathrm{in}},R_{\mathrm{out}})}^{(k)}(s)
        -
        Q_{A,S}^{(k)}(s)
        \right|
        \le
        \frac{2k}{S}
        \left(R_{\mathrm{out}}^k+R_{\mathrm{in}}^k\right)
        \frac{M q_A^{S-k}}{1-q_A^S},
        \label{eq:annulus-gap-bound}
    \end{align}
\end{proposition}

The gap also makes the phase margin of \cref{sec:noise} explicit: with a logarithmic gap, a number of quadrature points that is linear in \(M\) already gives a constant margin.

\begin{corollary}[Phase margin guarantee]
    \label[corollary]{cor:phase-margin-guarantee}
    If \(\beta\Delta(\mu)\ge\log(3M)\), then \(S\ge4M+2\) quadrature points suffice 
    to guarantee the phase margin satisfies \(\chi_S \ge \pi/2\).
\end{corollary}

The last statement adds the two sources of error.
The total error in a root moment is the sum of a quadrature term that decays in \(S\) and a statistical term, linear in the relative endpoint error \(\varsigma\), that comes from \cref{lem:shot-noise-quadrature-perturbation}.
Taking \(S\) large enough makes the first term no larger than the second.

\begin{proposition}[Fixed-parameter root-moment estimate]
    \label[proposition]{prop:fixed-root-moment-estimate}
    Assume \(\Delta(\mu)>0\), \(S\ge2k\), \(\chi_S(R,s)\ge\pi/2\), and
\(q=\e^{-\beta\Delta(\mu)}\).
    If all \(S\) endpoint estimates have relative error at most
\(\varsigma<1/\sqrt2\), then
\[
    \left|
    \mathcal{M}_R^{(k)}(s)-\widehat Q_{R,S}^{(k)}(s)
    \right|
    \le
    2kR^k \biggl(\frac{M}{S}
    \frac{q^{S-k}}{1-q^S}
    +
    \varsigma\biggr).
\]
    If in addition \(\beta\Delta(\mu)\ge \log(3M)\) and
\[
    S-k
    \ge
    \frac{\log(1+M/\varsigma)}{\log(3M)},
\]
    then
\[
    \left|
    \mathcal{M}_R^{(k)}(s)-\widehat Q_{R,S}^{(k)}(s)
    \right|
    \le
    4kR^k\varsigma.
\]
\end{proposition}

\begin{proof}[Proof of \cref{prop:circle-shot-count}]
    By \cref{prop:gap-lower-bound}, each endpoint satisfies
\[
    |\Xi_j|
    \ge
    \exp\left(-3M\e^{-\beta\Delta(\mu)}\right)
    \ge
    \e^{-1},
\]
    where the last inequality uses \(\beta\Delta(\mu)\ge\log(3M)\).
    Apply \cref{lem:hadamard-test} to each endpoint with failure
    probability \(\delta/S\) and relative precision \(\tau/4\).
    The lower bound \(|\Xi_j|\ge \e^{-1}\) shows that
\(n=O(\tau^{-2}\log(S/\delta))\) is sufficient.
    Taking a union bound over the \(S\) endpoints, all endpoint estimates have
    relative error at most \(\tau/4\) with probability at least \(1-\delta\).
    Since \(S\ge S_0\ge 4M+2\) and \(\beta\Delta(\mu)\ge\log(3M)\), \cref{cor:phase-margin-guarantee} gives \(\chi_S(R,s)\ge\pi/2\), which is the phase-margin hypothesis of \cref{prop:fixed-root-moment-estimate}.
    On this event, \cref{prop:fixed-root-moment-estimate} applies.
    The assumptions \(\beta\Delta(\mu)\ge\log(3M)\) and \(S\ge S_0\) are exactly
    the sufficient conditions collected above to reduce its bound to
\(4kR^k(\tau/4)=kR^k\tau\), taking the endpoint error \(\varsigma=\tau/4\).
\end{proof}

\section{Robustness to perturbation}\label{sec:perturbative-stability}
We now verify the contour conditions of \cref{thm:spices-recovery} for
Hamiltonians with interactions outside a preferred energy band, while
keeping the single-particle energies inside that band intact.
The perturbed Hamiltonian may contain arbitrary number-conserving interactions on modes outside the target window, provided those interactions do not couple the inside and outside mode algebras (as defined by the unperturbed Hamiltonian).
\begin{equation}
    H = H_{\mathrm{nonint}} + V,
    \qquad
    H_{\mathrm{nonint}} = \sum_{i \in I} E_i n_i + \sum_{j \in O} E_j n_j,
\end{equation}
where \(V\) is a number-conserving interaction supported on the modes \(O\) (made precise in \cref{subsec:spectator-sector-model}).
We call this a \emph{spectator-sector perturbation}, and we refer to objects relating to the preferred energy band as \emph{protected}.

Such a perturbation may seem quite restrictive, and indeed its special algebraic properties were chosen to facilitate a simple analysis.
The partition function factorizes into a protected in-band factor and an interacting spectator factor.
If the spectator factor has no zeros in the target annulus, then its contribution to the annular root moment vanishes exactly, regardless of the size of the perturbation.
The magnitude of the perturbation enters only when trying to ensure that the annulus remains free of zeros from the interacting sector and in the numerical conditioning of the contour estimator.
Proofs of the supporting statements are again in \cref{app:proofs}.

\subsection{Spectator-sector model}\label{subsec:spectator-sector-model}

We continue with the energy window \(\Omega=[\Omega_{\min},\Omega_{\max}]\) as before, and assume that no single-particle energy lies at either endpoint.
Let
\[
I=\{j:\Omega_{\min}<E_j<\Omega_{\max}\},
\qquad
O=\{1,\ldots,M\}\setminus I,
\]
and write
\[
M_I=|I|,
\qquad
M_O=|O|.
\]
The occupation basis gives a tensor-product decomposition of the full Hilbert space
\[
\mathcal F
=
\mathcal F_I\otimes\mathcal F_O.
\]
Under this decomposition,
\begin{align*}
    H_0
    &=
    H_I\otimes\I+\I\otimes H_O,
    &
    H_I
    &=
    \sum_{j\in I}E_jn_j,
    &
    H_O
    &=
    \sum_{j\in O}E_jn_j,\\
    N
    &=
    N_I\otimes\I+\I\otimes N_O,
    &
    N_I
    &=
    \sum_{j\in I}n_j,
    &
    N_O
    &=
    \sum_{j\in O}n_j.
\end{align*}

We consider perturbations \(V_O\) on \(\mathcal F_O\) that are number conserving and normalized, i.e.
\[
[V_O,N_O]=0,
\qquad
\|V_O\|=1.
\]
For real \(\lambda\), the interacting Hamiltonian \(H_\lambda\) can be written as
\begin{equation}
    H_\lambda
    =
    H_I\otimes\I
    +
    \I\otimes(H_O+\lambda V_O).
    \label{eq:spectator-H-lambda}
\end{equation}
Thus the perturbation
\[
V=\I\otimes V_O
\]
is supported outside the window as defined by the eigenstates of the unperturbed Hamiltonian $H_0$.
The operator \(V_O\) may contain arbitrary number-conserving interactions
among the outside modes, including density--density interactions.

Define the partition function of the interacting model as
\begin{equation}
    Z_\lambda(s,z)
    =
    \Tr\left(z^N\e^{-sH_\lambda}\right).
    \label{eq:spectator-partition-function}
\end{equation}
The partition function factorizes into an inside and an outside factor (\cref{prop:spectator-factorization}).
Here and below, \(\Tr_I\) and \(\Tr_O\) denote the traces over \(\mathcal F_I\) and \(\mathcal F_O\) of operators acting on those factors alone.
The inside and outside factors are given by
\begin{align}
    Z_I(s,z)
    &=
    \prod_{j\in I}
    \left(1+z\e^{-sE_j}\right),
    \label{eq:inside-Z-factor}\\
    Z_{O,\lambda}(s,z)
    &=
    \Tr_O\left(
    z^{N_O}\e^{-s(H_O+\lambda V_O)}
    \right).
    \label{eq:outside-Z-factor}
\end{align}

\subsection{Recovery under spectator-sector perturbations}

\begin{theorem}[SPICES recovery under spectator-sector perturbations]
    \label[theorem]{thm:spectator-spices-recovery}
    Consider the model \(H_\lambda\) of \cref{subsec:spectator-sector-model}, with \(I \neq \emptyset\)
    and window width \(L = \Omega_{\max} - \Omega_{\min} > 0\). Fix \(0 < \epsilon <
    L/2\) and an inverse-polynomial failure probability \(0 < \delta < 1\), and set
\[
    T=\left\lceil\frac{24L}{\epsilon}\right\rceil, \qquad t_{\max}=\frac{\pi T}{L}\,.
\]
    Suppose that
    \begin{equation}
        \beta\Delta_\Omega\ge\log(3M),
        \qquad
        |\lambda|\sqrt{\beta^2+t_{\max}^2}<\log\frac43.
        \label{eq:spectator-spices-conditions}
    \end{equation}
    For \(\mu\in\{\Omega_{\min},\Omega_{\max}\}\), assume that preparation of
\[
    \rho_{\beta,\mu}^{\lambda} = \frac{\e^{-\beta(H_\lambda - \mu N)}} {Z_\lambda(\beta, \e^{\beta
    \mu})}
\]
    and controlled evolutions \(\e^{-\ii t H_\lambda + \ii \vartheta N}\) at the required sampling times
\(|t| \leq t_{\max}\) and angles are available to any required inverse-polynomial operational
    accuracy with polynomial overhead. Also assume that \(L\), \(\epsilon^{-1}\), and
\(\e^{2\beta L}\) are polynomially bounded in \(M\).
    
    Then all five conditions of \cref{thm:spices-recovery} hold with
\[
    q=\frac 2 3\,, \qquad b = \frac 1 {4\e}\,.
\]
    Consequently, SPICES returns a finitely represented probability measure \(\widehat\nu\) such that,
    with probability at least \(1 - \delta\),
\[
    W_1(\widehat \nu, \nu_\Omega) \leq \epsilon, \qquad \nu_\Omega= \frac{\sum_{j \in
    I}\e^{-\beta (E_j - \Omega_{\min})}\delta_{E_j}} {\sum_{j \in I}\e^{-\beta(E_j -
    \Omega_{\min})}},
\]
    using polynomially many gates and samples and polynomial classical
    postprocessing, with the resource bound of \cref{thm:spices-recovery}.
\end{theorem}

\begin{remark}[Width of the window]
    The hypotheses of \cref{thm:spectator-spices-recovery} force the window to be narrow relative to the endpoint gap.
    Since \(I\neq\emptyset\), some protected energy lies within \(L/2\) of the nearer endpoint, so \(\Delta_\Omega\le L/2\).
    Conversely, if \(\e^{2\beta L}\le M^{c}\) then \(\beta L\le (c/2)\log M\), and together with \(\beta\Delta_\Omega\ge\log(3M)\) this gives \(L/\Delta_\Omega<c/2\).
    Hence \(2\le L/\Delta_\Omega<c/2\), so the window is wider than the gap by at most a constant factor.
    This comes from certifying radial separation and the signal floor through the logarithmic-temperature condition, not from \cref{thm:spices-recovery}, which takes \(q\) and \(b\) as inputs; we do not know whether the condition is necessary.
    The same constraint arises whenever \cref{cor:log-inv-temp-suffices} is combined with the polynomial bound on \(\e^{2\beta L}\) in \cref{thm:spices-recovery}.
\end{remark}

The proof of \cref{thm:spectator-spices-recovery} checks the five conditions of \cref{thm:spices-recovery} in turn, and the supporting statements fall into three groups.
The first shows that, as long as the spectator factor has no zeros in the annulus, it cannot change the window moments.
The second gives a Rouch\'e-type criterion, and a bound on the perturbation, that guarantee this for a finite time.
The third transfers the signal and quadrature bounds of \cref{sec:free-gapped} from the free polynomial to the perturbed one.

\begin{proposition}[Spectator-sector factorization]
    \label[proposition]{prop:spectator-factorization}
    For every \(s,z\in\C\),
    \begin{equation}
        Z_\lambda(s,z)
        =
        Z_I(s,z)Z_{O,\lambda}(s,z).
        \label{eq:spectator-Z-factorization}
    \end{equation}
\end{proposition}

We now define the inverse-fugacity polynomials
\begin{align*}
    P_I(s,w)
    &=
    w^{M_I}Z_I\!\left(s,-\frac1w\right),\\
    P_{O,\lambda}(s,w)
    &=
    w^{M_O}Z_{O,\lambda}\!\left(s,-\frac1w\right),\\
    P_\lambda(s,w)
    &=
    w^MZ_\lambda\!\left(s,-\frac1w\right).
\end{align*}
Proposition~\ref{prop:spectator-factorization} immediately implies
\begin{equation}
    P_\lambda(s,w)
    =
    P_I(s,w)P_{O,\lambda}(s,w).
    \label{eq:spectator-P-factorization}
\end{equation}
The outside factor is a degree-\(M_O\) polynomial because
\([V_O,N_O]=0\).
It need not factor into terms with zeros related to single-particle energies, but still has \(M_O\) complex roots.
For an explicit polynomial \(P\), we write \(\mathcal{M}_A^{(k)}[P](s)\) for the root moment of \(P(s,\cdot)\) over the compact region \(A\) in the complex plane.

For later conditioning estimates, also define the normalized signal
\begin{equation}
    \Xi_\lambda(s,z)
    =
    \frac{Z_\lambda(s,z)}{Z_\lambda(\beta,|z|)}.
    \label{eq:spectator-normalized-signal}
\end{equation}
At \(\lambda=0\), this agrees with the free normalized signal \(\Xi_0\).

\begin{proposition}[Exact protection of the window moments]
    \label[proposition]{prop:exact-window-protection}
    Fix \(s=\beta+\ii t\), and suppose that \(P_{O,\lambda}(s,\cdot)\) has no zeros in the closed annulus \(\overline{A_\Omega}\). 
    Then for every positive integer \(k\),
    \begin{equation}
        \mathcal{M}_{A_\Omega}^{(k)}[P_\lambda](s)
        =
        \sum_{j\in I}\e^{-ksE_j}.
        \label{eq:exact-protected-window-moment}
    \end{equation}
\end{proposition}

Exact protection requires the spectator factor to be zero-free in the annulus.
The next theorem gives a condition under which this persists after the perturbation is switched on.
It compares \(P_{O,\lambda}\) with the unperturbed spectator polynomial \(P_{O,0}\), which is zero-free in the annulus because the outside energies lie outside the window, and applies Rouch\'e's theorem on the two boundary circles.

\begin{theorem}[Finite-time spectator stability]
    \label[theorem]{thm:finite-time-spectator-stability}
    Fix \( t_{\max}>0\) and \(s=\beta+\ii t\).  
    Suppose that, for \(|t|\le t_{\max}\) and
\(|w|\in\{R_{\mathrm{in}},R_{\mathrm{out}}\}\),
    \begin{equation}
        \left|
        P_{O,\lambda}(s,w)-P_{O,0}(s,w)
        \right|
        <
        \left|P_{O,\lambda}(s,w)\right|
        +
        \left|P_{O,0}(s,w)\right|.
        \label{eq:symmetric-rouche-condition}
    \end{equation}
    Then for every positive integer \(k\) and uniformly for \(|t|\le t_{\max}\)
    \begin{equation}
        \mathcal{M}_{A_\Omega}^{(k)}[P_\lambda](s)
        =
        \sum_{j\in I}
        \e^{-k s E_j}.
        \label{eq:finite-time-protected-window-moment}
    \end{equation}
\end{theorem}

Some time constraint of this kind cannot be avoided for general spectator-sector perturbations when the gap is small: \cref{subsec:interacting-long-time-obstruction} gives an interacting example in which a spectator root enters the target annulus at a time of order \(1/|\lambda|\).
That example requires the spectator energies to be only weakly separated from the window.
It does not show that a finite horizon is necessary at larger gaps, where the finiteness of the horizon may be an artifact of the Rouch\'e argument used here.

To state the Rouch\'e condition \cref{eq:symmetric-rouche-condition} quantitatively, and to control the magnitude and phase of the perturbed signal below, we define the ratio of the perturbed to the unperturbed spectator polynomial,
\begin{equation}
    r_\lambda(s,w)
    \coloneqq
    \frac{P_{O,\lambda}(s,w)}
    {P_{O,0}(s,w)},
    \label{eq:ratio-r-lambda}
\end{equation}
(we write \(r_\lambda(w)\) when \(s\) is fixed), and the relative size of the perturbation on the two contour circles,
\begin{equation}
    \eta_\lambda(t_{\max})
    \coloneqq
    \sup_{\substack{|t|\le t_{\max}\\
    |w|\in\{R_{\mathrm{in}},R_{\mathrm{out}}\}}}
    \bigl|
    r_\lambda(s,w)
    -
    1
    \bigr|,
    \qquad s=\beta+\ii t.
    \label{eq:relative-outside-perturbation}
\end{equation}
The gap \(\Delta_\Omega>0\) ensures that the denominator is nonzero on both circles.
The Rouch\'e condition then follows from the simple requirement \(\eta_\lambda(t_{\max})<1\).

\begin{corollary}[Relative perturbation criterion]
    \label[corollary]{cor:relative-spectator-stability}
    If
    \begin{equation}
        \eta_\lambda( t_{\max})<1,
        \label{eq:rouche-smallness}
    \end{equation}
    then the conclusion of
    Theorem~\ref{thm:finite-time-spectator-stability} holds.
\end{corollary}

Whether \(\eta_\lambda<1\) depends on the spectator sector alone.
The next proposition bounds \(\eta_\lambda\) in terms of the perturbation strength \(\lambda\), the inverse temperature, the time horizon and the gap, using only the thermal weights and number conservation, and no further structure of the perturbation.

\begin{proposition}[Perturbation bound]
    \label[proposition]{prop:outside-duhamel-bound}
    For every real \(\lambda\),
    \begin{equation}
        \eta_\lambda( t_{\max})
        \le
        \sqrt{1+\frac{ t_{\max}^2}{\beta^2}}
        \left(\e^{\beta|\lambda|}-1\right)
        \tanh\!\left(
        \frac{\beta\Delta_{\Omega}}{2}
        \right)^{-M_O}.
        \label{eq:outside-duhamel-relative-bound}
    \end{equation}
\end{proposition}

In the logarithmic-gap regime the factor \(\tanh(\beta\Delta_\Omega/2)^{-M_O}\) is bounded by a constant, and the condition \(\eta_\lambda<1\) becomes a bound on \(|\lambda|\sqrt{\beta^2+t_{\max}^2}\).

\begin{corollary}[Stability in the logarithmic-gap regime]
    \label[corollary]{cor:log-gap-perturbative-stability}
    Suppose
    \begin{equation}
        \label{eq:log-gap-and-perturbation-conditions}
        \beta\Delta_\Omega\ge\log(3M),
        \qquad
        |\lambda|\sqrt{\beta^2+ t_{\max}^2}<\log \frac32.
    \end{equation}
    Then \(\eta_\lambda( t_{\max})<1\), and the conclusion of
    \cref{thm:finite-time-spectator-stability} holds that the annular root moments are preserved. 
\end{corollary}

Thus, for fixed \(\beta\) satisfying the logarithmic-gap condition, the certified stability horizon scales as \(\Theta_\beta(|\lambda|^{-1})\) as \(\lambda\to0\).
The example in \cref{subsec:interacting-long-time-obstruction} shows that order-one root motion can occur on the same time scale for an interacting spectator, although only under a weaker gap condition than the corollary assumes.
Thus, without a gap condition of this kind, a time horizon of order \(1/|\lambda|\) cannot be improved.
Whether the horizon can be lifted entirely under the logarithmic-gap condition is not settled by that example, and we leave it open.

\begin{proposition}[Normalized-signal stability]
    \label[proposition]{prop:conditioning-transfer}
    Suppose for all \(j\)
\[
    |r_\lambda(w_j)-1|
    \le
    \eta < 1.
\]
    Then at every quadrature point,
    \begin{equation}
        |\Xi_\lambda(s,-1/w_j)|
        \ge
        (1-\eta)\e^{-\beta|\lambda|}
        |\Xi_0(s,-1/w_j)|,
        \label{eq:normalized-signal-transfer}
    \end{equation}
\end{proposition}
The deterministic bounds in \cref{sec:free-gapped} depend only on
the radial separation between the roots and the contour.
In particular, they do not require the roots to have a single-particle interpretation.
It therefore suffices to certify a thickened zero-free region around each contour
circle.
The same Duhamel and Rouch\'e estimates used above then give a certificate directly in terms of the perturbation strength.

For \(\gamma>0\), define a thickened annular ``collar'' around a contour circle of radius \(R>0\) by
\begin{equation}
    \mathcal C_\gamma(R)
    \coloneqq
    \left\{
    w:
    R\e^{-\gamma}
    \le
    |w|
    \le
    R\e^\gamma
    \right\}.
    \label{eq:logarithmic-contour-collar}
\end{equation}
If the region
\(\mathcal C_\gamma(R_{\mathrm{in}}) \cup \mathcal C_\gamma(R_{\mathrm{out}})\)
is free of zeros, then it represents a stable region of alternative annuli that \(A_\Omega\) could deform to.
This situation is depicted in \cref{fig:spices-contour-certificates}.
Our strategy is to assume this region is zero-free and use that margin to prove stability using the same techniques as in \cref{sec:free-gapped}.

Fix \( t_{\max}>0\), and choose
\[
0<\gamma<\beta\Delta_{\Omega}
\]
so that (cf.~\cref{eq:outside-duhamel-relative-bound})
\begin{equation}
    \sqrt{1+\frac{ t_{\max}^2}{\beta^2}}
    \left(\e^{\beta|\lambda|}-1\right)
    \tanh\!\left(
    \frac{\beta\Delta_{\Omega}-\gamma}{2}
    \right)^{-M_O}
    <1.
    \label{eq:quadrature-collar-smallness}
\end{equation}
Set
\begin{equation}
    q_\lambda
    \coloneqq
    \exp\left[
    -\min\{\beta\Delta_\Omega,\gamma\}
    \right].
    \label{eq:perturbed-quadrature-gap}
\end{equation}
Let \(Q_{R,S}^{(k)}(s)\) denote the sinc-corrected circular quadrature from
\cref{eq:midpoint-quadrature-rule} of the free analysis, now applied to \(P_\lambda(s,\cdot)\).

\begin{proposition}[Quadrature gap stability]
    \label[proposition]{prop:quadrature-gap-transfer}
    Uniformly for \(|t|\le t_{\max}\), the spectator polynomial \(P_{O,\lambda}(s,\cdot)\) is zero-free in
\[
    \mathcal C_\gamma(R_{\mathrm{in}})
    \cup
    \mathcal C_\gamma(R_{\mathrm{out}}).
\]
    For either contour radius \(R\), every root \(\zeta\) of
\(P_\lambda(s,\cdot)\) therefore satisfies
    \begin{equation}
        \min\left\{
        \frac{|\zeta|}{R},
        \frac{R}{|\zeta|}
        \right\}
        \le
        q_\lambda.
        \label{eq:perturbed-root-radial-separation}
    \end{equation}
    Consequently, for \(S\ge2k\),
    \begin{equation}
        \left|
        \mathcal{M}_R^{(k)}[P_\lambda](s)
        -
        Q_{R,S}^{(k)}(s)
        \right|
        \le
        \frac{2k}{S}R^k
        \frac{
        Mq_\lambda^{S-k}
        }{
        1-q_\lambda^S
        }.
        \label{eq:perturbed-disc-quadrature-bound}
    \end{equation}
\end{proposition}

In the logarithmic-gap regime the constants in this proposition are explicit, and they give the radial separation \(q=2/3\) used in \cref{thm:spectator-spices-recovery}.

\begin{corollary}[Constant quadrature gap]
    \label[corollary]{cor:constant-perturbed-quadrature-gap}
    Suppose
    \begin{equation}
        \beta\Delta_\Omega\ge\log(3M),
        \qquad
        |\lambda|\sqrt{\beta^2+ t_{\max}^2}<\log\frac43.
        \label{eq:constant-quadrature-gap-conditions}
    \end{equation}
    Then the annular root moments are protected as in
    \cref{cor:log-gap-perturbative-stability}, and
    \cref{prop:quadrature-gap-transfer} holds with
\[
    \gamma=\log\frac32,
    \qquad
    q_\lambda=\frac23.
\]
    In particular, the perturbed quadrature error decays geometrically in
\(S-k\) with ratio \(2/3\).
\end{corollary}

We can now prove \cref{thm:spectator-spices-recovery}.
The corollary above supplies the first two conditions of \cref{thm:spices-recovery}, the conditioning bound gives the third, and the last two are assumptions.

\begin{proof}[Proof of \cref{thm:spectator-spices-recovery}]
    All sampled times \(t_\ell=\pi\ell/L\), \(0\le\ell\le T\), lie in \([0, t_{\max}]\). We verify the
    five conditions in order.
    \begin{enumerate}
        \item \emph{Protected spectrum.}
        Apply \cref{cor:constant-perturbed-quadrature-gap} with time horizon \(t_{\max}\). The spectator
        polynomial has no roots in the target annulus throughout this interval. By the factorization in
        \cref{prop:spectator-factorization}, the roots there are precisely \(\e^{-(\beta + \ii t)E_j}\),
        \(j \in I\), with multiplicity. This is a fixed, nonempty protected spectrum.
        
        \item \emph{Radial separation.}
        The same corollary gives \cref{prop:quadrature-gap-transfer} with
        \(q_\lambda = 2/3\) on both boundary circles, establishing condition~2.
        
        \item \emph{Normalized signal.}
        The gap hypothesis and \cref{prop:outside-duhamel-bound} give
        \[
        \eta_\lambda(t_{\max}) \leq 2\sqrt{1 + \frac{t_{\max}^2}{\beta^2}} (\e^{\beta |\lambda|} -
        1) \leq 2\left(\e^{|\lambda|\sqrt{\beta^2 + t_{\max}^2}}-1\right) <\frac 2 3.
        \]
        The factor~2 is the logarithmic-gap estimate used in \cref{cor:log-gap-perturbative-stability};
        the second inequality follows from convexity of the exponential. Moreover,
        \(\e^{-\beta|\lambda|} > 3/4\), and \cref{cor:log-inv-temp-suffices} gives \(|\Xi_0| \geq
        \e^{-1}\). Thus \cref{prop:conditioning-transfer} yields, uniformly at all quadrature nodes and
        sampled times,
        \[
        |\Xi_\lambda| \geq (1 - \eta_\lambda(t_{\max}))\e^{-\beta|\lambda|}|\Xi_0| \geq \frac 1
        {4\e}.
        \]
        Condition~3 therefore holds with the stated \(b\).
        
        \item \emph{Operational access.}
        Condition~4 is exactly the assumed access to preparations and controlled evolutions for the
        perturbed Hamiltonian \(H_\lambda\).
        
        \item \emph{Polynomial conditioning.}
        Here \(b^{-1} = 4\e\) and \((1-q)^{-1} = 3\) are constants, while \(L\),
        \(\epsilon^{-1}\), and \(\e^{2\beta L}\) are polynomially bounded by assumption.
        This proves condition~5.
    \end{enumerate}
    The conclusion now follows directly from \cref{thm:spices-recovery}.
\end{proof}

\subsection{Long-time obstruction for weakly gapped interacting spectators}
\label{subsec:interacting-long-time-obstruction}

For an unrestricted interacting \(V_O\) and a small gap, the finite-time qualification in
\cref{thm:finite-time-spectator-stability} cannot be removed.
The following two-mode example shows that, in this weak-gap regime, root motion of order one occurs on the time scale
\(1/|\lambda|\), so that a horizon of order \(1/|\lambda|\) cannot be improved.
The obstruction persists even when \(V_O\) commutes with \(H_O\).

Consider two degenerate outside modes with
\[
H_O
=
E(n_1+n_2),
\qquad
V_O
=
n_1n_2.
\]
Then
\[
\|V_O\|=1,
\qquad
[V_O,N_O]=0,
\]
and direct summation over the four occupation states gives
\begin{equation}
    Z_{O,\lambda}(s,z)
    =
    1
    +
    2z\e^{-sE}
    +
    z^2\e^{-s(2E+\lambda)}.
    \label{eq:interacting-two-mode-partition}
\end{equation}
Set \(a=\e^{-sE}\).
Substitution of \(z=-1/w\) gives
\[
P_{O,\lambda}(s,w)
=
w^2
-
2aw
+
a^2\e^{-s\lambda}.
\]
The quadratic formula yields the roots
\begin{equation}
    \zeta_\pm(s)
    =
    a\left(
    1
    \pm
    \sqrt{1-\e^{-s\lambda}}
    \right).
    \label{eq:interacting-two-mode-roots}
\end{equation}

Fix \(\beta>0\).
For \(\lambda>0\), set
\[
t_\lambda=\frac{\pi}{\lambda},
\qquad
s_\lambda=\beta+\ii t_\lambda.
\]
Then
\[
\e^{-s_\lambda\lambda}
=
-\e^{-\beta\lambda}.
\]
At this time, the term inside the square root in
\eqref{eq:interacting-two-mode-roots} is the positive real number
\(1+\e^{-\beta\lambda}\).
The two root moduli are therefore
\[
\e^{-\beta E}
\left(\sqrt{1+\e^{-\beta\lambda}}\pm1\right).
\]
As \(\lambda\to0\) these converge to
\[
\e^{-\beta E}(\sqrt2\pm1).
\]

For example, suppose that \(E>\Omega_{\max}\) and
\[
\beta(E-\Omega_{\max})
<
\log(1+\sqrt2).
\]
For all sufficiently small \(\lambda>0\), both roots begin below the inner
contour at \(t=0\), while the larger root lies above it at
\(t=t_\lambda\).
By continuity, an outside root enters the target annulus at some time no
larger than \(\pi/\lambda\), contaminating the window moment.
Thus an arbitrarily small interaction can produce order-one radial motion
on the time scale \(1/|\lambda|\).
In particular, when the spectator is only weakly separated from the window, the finite-time qualification cannot be removed for general
interacting spectators.
Thus, despite the
commutation relation \([V_O,H_O]=0\), the interaction induces
a change in the root moment.

The example is a weak-gap phenomenon.
For every real \(t\) we have \(|1-\e^{-s\lambda}|\le 1+\e^{-\beta\lambda}<2\), so by \eqref{eq:interacting-two-mode-roots} both roots satisfy
\(|\zeta_\pm(s)|<(1+\sqrt2)\,\e^{-\beta E}\) at all times.
If \(\beta(E-\Omega_{\max})\ge\log(1+\sqrt2)\), neither root can ever reach the inner contour, and the mechanism above disappears.
We therefore make no claim about the necessity of a finite horizon in the logarithmic-gap regime of \cref{cor:log-gap-perturbative-stability}.

\section{DQC1 hardness with a protected single-particle energy}\label{sec:ff-hardness}

This section complements the recovery results above with a matching hardness construction. 
To construct an output instance matching the requirements of the algorithmic result, the reduction produces,
\begin{itemize}
    \item A sparse list of Pauli operators encoding the Hamiltonian,
    \item A uniform circuit compiler for number measurement, controlled joint evolution, and grand-canonical Gibbs preparation.
\end{itemize}
We use the occupation convention of \cref{sec:free-model},
\[
n_j=\frac{\I+Z_j}{2}\,.
\]

\begin{definition}[Gapped real trace sign]
    \label[definition]{def:ff-trace-sign}
    The input is a circuit \(\mathcal{U}\) encoded via a fixed gate set on \(m\) qubits and a rational \(g\ge1/\operatorname{poly}(m)\). Let \(\ntr(\mathcal{U}) = 2^{-m} \Tr \mathcal{U}\). We are promised that either,
    \begin{equation}
        \label{eq:ff-trace-promise}
        \Re\ntr(\mathcal{U})\ge g \quad\text{or}\quad \Re\ntr(\mathcal{U})\le-g
    \end{equation}
    holds. 
    The task is to decide which holds. 
    We use the standard fact that this promise problem is DQC1-hard \cite{cade2018quantum, ji2026dqc1}.
\end{definition}

We use the access model of \cref{def:ff-operational-access} (\cref{sec:spse}), including its parameters \(\ell_{\mathrm{enc}}\), \(\ell_{\mathrm{qry}}\) and \(p_{\mathrm{enc}}\).

\begin{theorem}[Pauli-sparse DQC1 hardness with public number access]
    \label[theorem]{thm:ff-main}
    There is a deterministic polynomial-time reduction from \cref{def:ff-trace-sign} to an operationally described triple \((M, H,N)\) on
\[
    M=2m+4
\]
    qubits with the following properties.
    \begin{enumerate}
        \item The Hamiltonian is the three-term Pauli sum
        \begin{equation}
            H = 16n_R+17n_S-\frac{17}{2}\I = 8\I+8Z_R+\frac{17}{2}Z_S. \label{eq:ff-physical-H}
        \end{equation}
        
        \item There exists a unitary \(W\) such that,
        \[
        N=WN_{\mathrm{std}}W^\dagger.
        \]
        In particular, \(N\) has eigenvalues \(x \in \{0,\ldots,M\}\) with multiplicities \(\binom Mx\), and \([H,N]=0\).
        
        \item In the \(W\) basis (that is, taking \(U=W^\dagger\) in \cref{eq:ffansatz}), \(H\) has the spectator-sector form of
        \cref{subsec:spectator-sector-model}.  
        If \(\Re\ntr(\mathcal{U})\ge g\), then
        \begin{equation}
            \label{eq:ff-positive-tensor}
            W^\dagger H W = 16n_R\otimes\I + \I_R\otimes\frac{17}{2}(V_+)_{SC}, \qquad \|V_+\|=1,
        \end{equation}
        while if \(\Re\ntr(\mathcal{U})\le-g\), then
        \begin{equation}
            \label{eq:ff-negative-tensor}
            W^\dagger H W = 17n_S\otimes\I + \I_S\otimes\frac{17}{2}(V_-)_{RC}, \qquad \|V_-\|=1.
        \end{equation}
        Here \(V_\pm\) are explicit unit-norm operators, constructed in \cref{sec:number_operator}, which conserve the spectator occupation.
        
        \item For the window
        \[
        \Omega=[15,18],
        \]
        the unique protected energy is
        \[
        E_\star=
        \begin{cases}
            16,&\Re\ntr(\mathcal{U})\ge g,\\
            17,&\Re\ntr(\mathcal{U})\le-g.
        \end{cases}
        \]
        Its endpoint gap is \(\Delta_\Omega=1\), and the perturbation scale is \(17/2<16\).
        This is far outside the perturbative regime of \cref{thm:spectator-spices-recovery}, so the contour conditions needed by SPICES are verified directly in \cref{prop:ff-certificates}.
        
        \item The operational compiler satisfies the resource bounds of
        \cref{def:ff-operational-access}, including polynomial dependence on
        \(\e^\beta\) for Gibbs preparation. At \(\beta=\log(3M)\), all three
        primitives have polynomial cost for polynomial-bit queries and
        polynomially bounded \(r\). In particular, accuracy \(2^{-M^2}\) at
        this temperature has a polynomial-size circuit description.
    \end{enumerate}
\end{theorem}

The proof is given in \Cref{sec:ff-main-proof}, and the required pieces are developed in the
following sections. As a brief outline,
\begin{enumerate}
    \item \Cref{sec:twirlprojector} uses a Hadamard test to construct an operator which encodes
    $\Re \ntr(\mathcal{U})$ into the spectrum of an operator $\mathcal G$. This is used to build
    projectors $\Pi_{\leq}, \Pi_{\geq}$ which will be used to distinguish between the
    two cases of \Cref{eq:ff-trace-promise}.
    \item \Cref{sec:number_operator} establishes the relationship between the number operator
    and the change of basis $W$.
    \item \Cref{sec:ff-polynomial} constructs polynomial representing the partition function.
    This yields a factorization into a protected factor containing the signal SPICES
    uses to discern the two cases of \Cref{eq:ff-trace-promise}, and a spectator factor.
    The spectator factor is shown to be zero-free within the SPICES contour in
    \Cref{sec:ff-zero-free}, proving robustness.
    \item \Cref{sec:ff-compiler} proves that we can efficiently construct access to the number
    operator, joint evolution, and the Gibbs state.
\end{enumerate}

After proving \Cref{thm:ff-main}, we prove in \Cref{sec:ff-spices} that SPICES can retrieve the
signal allowing it to solve \Cref{def:ff-trace-sign}.

\subsection{DQC1 hardness of protected-energy recovery}
\label{sec:ff-recovery-hardness}

We now show that the task of \cref{def:ff-recovery}, which is the single-atom special case of \cref{def:spse-formal}, is DQC1-hard.

\begin{definition}[Operational protected-energy recovery]
    \label[definition]{def:ff-recovery}
    The input is an operational description in the sense of \cref{def:ff-operational-access}, together with a target interval \(\Omega\). 
    It is promised that the protected measure in \(\Omega\) is a single atom \(\delta_{E_\star}\). 
    A valid output is either:
    \begin{enumerate}
        \item a polynomial-bit rational \(\widehat E\) satisfying \(|\widehat E-E_\star|<1/2\); or
        \item a finitely represented probability measure \(\widehat\nu\) whose first moment is approximable in deterministic polynomial time to additive error \(1/12\), and which satisfies
        \[
        W_1(\widehat\nu,\delta_{E_\star})\le\frac 1 3.
        \]
    \end{enumerate}
    The solver must succeed with probability at least \(2/3\), with total
    classical running time and elementary-gate count polynomial in the input
    length on every execution. These costs include compilation and execution
    of all invoked primitives; a Gibbs call is charged its full
    temperature-dependent implementation cost.
\end{definition}

\begin{theorem}[Operational DQC1 hardness]
    \label[theorem]{thm:ff-hardness}
    The task in \Cref{def:ff-recovery} is DQC1-hard even on the family of
    \cref{thm:ff-main}.
\end{theorem}

\begin{proof}
    Given an instance \(\mathcal{U},g\) of \cref{def:ff-trace-sign}, apply the
    deterministic polynomial-time reduction of \cref{thm:ff-main} and set
\(\Omega=[15,18]\). The resulting operational description is a valid
    instance of \cref{def:ff-recovery}: by item~4 of \cref{thm:ff-main},
    its protected measure is
    \begin{equation}
        \label{eq:op_hardness_cases}
        \nu_{\mathcal{U}}=
        \begin{cases}
            \delta_{16},&\Re\ntr(\mathcal{U})\ge g,\\
            \delta_{17},&\Re\ntr(\mathcal{U})\le-g.
        \end{cases}
    \end{equation}
    Run any algorithm for \cref{def:ff-recovery} on this instance.
    
    For the first output format, threshold \(\widehat E\) at \(33/2\):
    an error strictly smaller than \(1/2\) determines the correct trace sign.
    For the second format, compute \(\widetilde E\) within \(1/12\) of
\(\int E\,\widehat\nu(\dd E)\). Since the only coupling with
\(\delta_{E_\star}\) pairs every \(E\) with \(E_\star\),
\[
    |\widetilde E-E_\star|
    \le \frac1{12}+\int|E-E_\star|\,\widehat\nu(\dd E)
    = \frac1{12}+W_1(\widehat\nu,\delta_{E_\star})
    \le \frac5{12}<\frac12.
\]
    Thresholding \(\widetilde E\) at \(33/2\) therefore also determines
    the case of \Cref{eq:op_hardness_cases}.
    
    The reduction and postprocessing take polynomial time and the resulting decision procedure succeeds
    with probability at least \(2/3\). Since \cref{def:ff-trace-sign} is DQC1-hard, so is
    \Cref{def:ff-recovery}.
\end{proof}

\subsection{Pauli twirl and signal subspace projector}
\label{sec:twirlprojector}

Let \(\reg L=\reg L_1 \reg L_2\) be a \(2m\)-qubit register. 
A polynomial-size circuit \(R_\Phi\) prepares
\[
R_\Phi\ket{0^{2m}} = \ket{\Phi} = 2^{-m/2}\sum_{x\in\{0,1\}^m}\ket{ x,x}.
\]
Define
\begin{equation}
    \label{eq:ff-CU}
    {\widetilde{\mathcal{U}}} = R_\Phi^\dagger(\mathcal{U}\otimes\I)R_\Phi.
\end{equation}
The maximally entangled-state identity gives
\begin{equation}
    \label{eq:ff-CU-amplitude}
    \bra{0^{2m}}{\widetilde{\mathcal{U}}}  \ket{0^{2m}} = \bra{\Phi}(\mathcal{U}\otimes\I)\ket\Phi = 2^{-m}\Tr \mathcal{U} = \ntr(\mathcal{U}).
\end{equation}
We implement the Hadamard test for ${\widetilde{\mathcal{U}}}$ by introducing another register $a$; the standard analysis would give a success probability equal to $\frac 1 2  (1 + \Re \ntr(\mathcal{U}))$. However, we want to rescale this so that it is strictly between $0$ and $1$. To achieve this, introduce another single qubit register $\reg b$ on which we control the entire Hadamard test. Then, the full ``computation'' registers are $\reg C = (\reg a, \reg b, \reg L)$, which comprise $M_{\reg C} = 2 + 2m$ qubits. We call this the computational registers to distinguish them from the registers $\reg R$, $\reg S$ which will hold the signal-carrying modes of the Hamiltonian. Define the rank-$1$ projector
\begin{equation}
    \label{eq:ff-Pi0}
    \Pi_0 = \ketbra{0^{\,2m+2}}{0^{\,2m+2}}_{\reg C}\,.
\end{equation}

We next define a circuit \(V_{\mathcal{U}}\) on \(a,b,L\).  
Starting from \(\ket{0}_a\ket{0}_b\ket{0^{2m}}_L\), it:
\begin{enumerate}
    \item applies Hadamards \(H_b\) and \(H_a\);
    \item applies \({\widetilde{\mathcal{U}}}\) \eqref{eq:ff-CU} to \(L\) controlled on \(a=b=1\); 
    \item applies Hadamard \(H_a\), controlled on \(b=1\).
\end{enumerate}
The circuit is shown in \cref{fig:ff-V-circuit}.

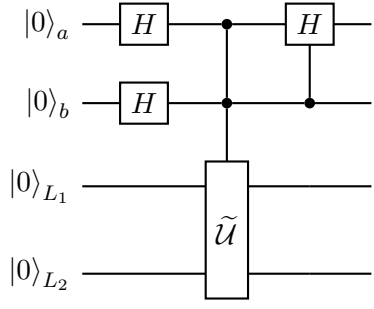
\begin{figure}
    \centering
    \begin{quantikz}
        \lstick{$\ket{0}_a$}     & \gate{H} & \ctrl{1}      & \gate{H}  & \qw \\
        \lstick{$\ket{0}_b$}     & \gate{H} & \ctrl{1}      & \ctrl{-1} & \qw \\
        \lstick{$\ket{0}_{L_1}$} & \qw      & \gate[2]{{\widetilde{\mathcal{U}}}} & \qw       & \qw \\
        \lstick{$\ket{0}_{L_2}$} & \qw      &               & \qw       & \qw
    \end{quantikz}
    \caption{The circuit \(V_{\mathcal{U}}\) described in this section, where \({\widetilde{\mathcal{U}}}\) is as defined in \cref{eq:ff-CU}.}
    \label{fig:ff-V-circuit}
\end{figure}

Set
\begin{equation}
    \label{eq:ff-Pi1}
    \Pi_1 = V_{\mathcal{U}}^\dagger \bigl( \ketbra{0}{0}_a \otimes\I_{\reg b,\reg L} \bigr) V_{\mathcal{U}}.
\end{equation}

\begin{lemma}[Scalar principal angle]
    \label[lemma]{lem:ff-principal-angle}
    The two projectors obey
    \begin{equation}
        \label{eq:ff-alpha}
        \Pi_0\Pi_1\Pi_0 = \alpha_{\mathcal{U}}\Pi_0, \qquad \alpha_{\mathcal{U}} = \frac 1 2+\frac 1 4\Re\ntr(\mathcal{U})\,.
    \end{equation}
    In particular, \(1/4\le\alpha_{\mathcal{U}}\le3/4\).
\end{lemma}

\begin{proof}
    Expanding both projectors, we write 
    \begin{align*}
        \Pi_0 \Pi_1 \Pi_0 &= \ketbra{0^{\,2m+2}}{0^{\,2m+2}}_{a,b,L} \cdot \left(V^\dagger_{\mathcal{U}} (\ketbra{0}{0}_a \otimes \I_{\reg b, \reg L}) V_{\mathcal{U}}\right)\ketbra{0^{\,2m+2}}{0^{\,2m+2}}_{a,b,L} \\
        &= \bra{0^{2m + 2}} \left(V^\dagger_{\mathcal{U}} (\ketbra{0}{0}_a \otimes \I_{\reg b, \reg L}) V_{\mathcal{U}}\right) \ket{0^{2m + 2}} \cdot \ketbra{0^{\,2m+2}}{0^{\,2m+2}}_{a,b,L}\\
        &= |\bigl(\bra{0}_a \otimes \I_{b, L} \bigr) V_{\mathcal{U}} \ket{0^{2m + 2}}|^2 \Pi_0
    \end{align*}
    which proves the first equation in \cref{eq:ff-alpha}, where
    \begin{align*}
        \alpha_{\mathcal{U}} = |\bigl(\bra{0}_a \otimes \I_{b, L} \bigr) V_{\mathcal{U}} \ket{0^{2m + 2}}|^2
    \end{align*}
    and 
    \begin{align*}
        V_{\mathcal{U}} \ket{0^{2m + 2}} &= \frac{1}{2} \bigl[\bigl(\ket{00} + \ket{+1} + \ket{10} \bigr)_{a, b} \ket{0^{2m}}_{L} + \ket{-1}_{a, b} R_{\Phi}^{\dagger}\bigl(\mathcal{U}_{L_1} \otimes \I_{L_2}\bigr)\ket{\Phi}_L \bigr] \\
        \bigl(\bra{0}_a \otimes \I_{b, L} \bigr) V_{\mathcal{U}} \ket{0^{2m + 2}} &= \frac{1}{2} \bigl[\bigl(\ket{0} + \frac{1}{\sqrt{2}}\ket{1}\bigr)_{b} \ket{0^{2m}}_{L} + \frac{1}{\sqrt{2}}\ket{1}_b R_{\Phi}^{\dagger}\bigl(\mathcal{U}_{L_1} \otimes \I_{L_2}\bigr)\ket{\Phi}_L \bigr] \\
        &= \frac{1}{2} \bigl\{\ket{0^{2m + 1}}_{b, L} + \frac{1}{\sqrt{2}}\ket{1}_b\bigl[\ket{0^{2m}} + R_{\Phi}^{\dagger}\bigl(\mathcal{U}_{L_1} \otimes \I_{L_2}\bigr)\ket{\Phi}\bigr]_L\bigr\} \\
        |\bigl(\bra{0}_a \otimes \I_{b, L} \bigr) V_{\mathcal{U}} \ket{0^{2m + 2}}|^2 &= \frac{1}{4} \bigl\{1 + \frac{1}{2}\bigl[2 + 2 \Re{\bra{0^{2m}} R_{\Phi}^{\dagger} \bigl(\mathcal{U}_{L_1} \otimes \I_{L_2} \bigr)\ket{\Phi} }\bigr] \bigr\} \\
        &= \frac{1}{2} + \frac{1}{4} \Re{\ntr(\mathcal{U})} 
    \end{align*}
    proving the second equation in \cref{eq:ff-alpha}.
\end{proof}

We now define the unitary rotation encoding \(\alpha_{\mathcal{U}}\) as a phase.

\begin{lemma}
    \label[lemma]{lem:ff-two-projector}
    Let \(P,Q\) be projectors satisfying \(PQP=\alpha P\), with \(0<\alpha<1\).  
    Set
    \begin{equation}
        \label{eq:principle_angle_calc}
        \mathcal G=(2Q-\I)(2P-\I), \qquad \cos\omega=2\alpha-1, \qquad 0<\omega<\pi\,.
    \end{equation}
    The operator $\mathcal G$ yields an orthogonal decomposition of the Hilbert space as,
\[
    \mathcal H' \oplus \mathcal H_{0,1} \oplus \mathcal H_{0, 0}\,,
\]
    where
    \begin{align*}
        &\mathcal H' = \supp(P) \oplus \supp((\I - P)QP)\,,\\
        &\mathcal H_{0,1} = \supp(\I - P) \cap \supp(Q)\,,\\
        &\mathcal H_{0,0} = \supp(\I - P) \cap \supp(\I - Q)\,.
    \end{align*}
    These spaces have eigenvalue and multiplicity,
    \begin{equation}
        \label{eq:eig_table}
        \begin{array}{c|c|c|c|c}
            \text{Space}&\multicolumn{2}{|c|}{\mathcal H'}&\mathcal H_{0,1}&\mathcal H_{0,0}\\ \hline
            \text{Eigenvalue}&\e^{\ii\omega}&\e^{-\ii\omega}&-1&+1\\ \hline
            \text{Dimension}& \rank(P)&\rank(P)&\rank(Q)-\rank(P)&d-\rank(P)-\rank(Q).
        \end{array}
    \end{equation}
    Moreover, if \(\Pi_+\) is the projector onto the \(\e^{\ii\omega}\)-eigenspace, then
    \begin{equation}
        \label{eq:ff-half-overlap}
        \dim \Pi_+ = \rank(P)\, \qquad P\Pi_+P=\frac 1 2 P\,.
    \end{equation}
\end{lemma}
\begin{proof}
    For a unit vector \(\ket p  \in\supp(P)\), define
\[
    \ket{q_\perp} = \frac{Q \ket p-\alpha \ket p}{\sqrt{\alpha(1-\alpha)}}\,.
\]
    The identity \(PQP=\alpha P\) shows that \(\ket {q_\perp}\perp\supp P\) and
\(\|\ket{q_\perp}\|=1\); this implies that $\ket{q_\perp} \in \supp((\I -
    P)QP)$. Additionally, we see that for any orthogonal pair \(\ket p\), \(\ket
    {p'} \in \supp P\), the corresponding vectors \(\ket{q_\perp}\),
\(\ket{q'_\perp}\) are also orthogonal. Consider a subspace spanned by the
    vectors \((\ket p,\ket {q_\perp})\). In this basis, we have
\[
    P= \begin{pmatrix}1&0\\0&0\end{pmatrix},
    \qquad Q= \begin{pmatrix}
        \alpha&\sqrt{\alpha(1-\alpha)}\\
        \sqrt{\alpha(1-\alpha)}&1-\alpha
    \end{pmatrix}\,.
\]
    Consequently $\mathcal G$ acts invariantly in this two-dimensional subspace and
\[
    \mathcal G = \begin{pmatrix}
        2\alpha-1&-2\sqrt{\alpha(1-\alpha)}\\
        2\sqrt{\alpha(1-\alpha)}&2\alpha-1
    \end{pmatrix} = (2 \alpha - 1) \I - 2\ii \sqrt{\alpha (1 - \alpha)} \sigma^y = \exp\bigl(-\ii \omega \sigma^y \bigr)\,.
\]
    where we recall that \(\cos\omega=2\alpha-1\) in the last equality.
    In particular, this is a two-dimensional rotation with eigenvalues \(\e^{\pm\ii\omega}\).  
    Its normalized \(\e^{\ii\omega}\)-eigenvector has first coordinate of modulus \(1/\sqrt2\), which proves \eqref{eq:ff-half-overlap}. 
    
    Next, consider $\mathcal H_{0,0}$ and $\mathcal H_{0,1}$. First, since $\rank(PQP) = \rank(P)$, this establishes $\rank(\mathcal H_{0,1}) = \rank(Q) - \rank(P)$. Additionally, $\supp(P) \cap \supp(\I - Q) = \emptyset$ since $PQP = \alpha P$. Thus, the remaining space $\supp(\I - P) \cap \supp(\I - Q)$ has dimension $d - \rank(P) - \rank(Q)$, as desired. The eigenvalues follow by definition of $\mathcal G$.
\end{proof}

We apply \Cref{lem:ff-two-projector} to \(P=\Pi_0\), \(Q=\Pi_1\), and write
\[
\mathcal G_{\mathcal{U}}=(2\Pi_1-\I)(2\Pi_0-\I).
\]
By construction, the ambient dimension is \(d=2^{M_{\reg C}}\),
\(\rank(P)=1\), and \(\rank(Q)=2^{M_{\reg C}-1}\). Thus \cref{eq:eig_table}
gives eigenvalues \(\e^{\pm\ii\omega_{\mathcal{U}}}\) of multiplicity one each, and eigenvalues
\(+1,-1\), each of multiplicity \(2^{M_{\reg C}-1}-1\).
The eigenvalues $e^{\pm \ii \omega_{\mathcal{U}}}$ satisfy
\begin{equation}
    \label{eq:ff-phase-trace}
    \cos\omega_{\mathcal{U}} = 2\alpha_{\mathcal{U}}-1 = \frac 1 2\Re\ntr(\mathcal{U})\,.
\end{equation}
By construction, \(-1 \leq \Re{\ntr(\mathcal{U})} \leq 1\) and so \(\omega_{\mathcal{U}}\in[\pi/3,2\pi/3]\). We can write the projector $\Pi_+$ from \Cref{lem:ff-two-projector} as
\begin{equation}
    \Pi_+ = \mathbf1_{\{\pi/4<\arg z<3\pi/4\}}(\mathcal G_{\mathcal{U}}), \label{eq:ff-Piplus}\,.
\end{equation}
We will see that depending on the value of $\Re \ntr(\mathcal{U})$, all eigenphases will be smaller or larger than $\pi/2$. To handle both cases generically, define the projectors
\begin{align}
    \Pi_\geq &= \mathbf1_{\{\pi/4<\arg z<\pi/2\}}(\mathcal G_{\mathcal{U}}), \label{eq:ff-PiL}\\
    \Pi_\leq &= \mathbf1_{\{\pi/2<\arg z<3\pi/4\}}(\mathcal G_{\mathcal{U}}). \label{eq:ff-PiR}\,.
\end{align}
Their relationship to $\Re \ntr(\mathcal{U})$ is formalized in the following lemma.
\begin{corollary}[$\Pi_+$ distinguishes between small and large trace]
    \label[corollary]{cor:ff-trace-switch}
    The projector \(\Pi_+\) satisfies
    \begin{equation}
        \label{eq:ff-principal-half}
        \Pi_0\Pi_+\Pi_0 = \frac 1 2 \Pi_0\,,
    \end{equation}
    and has dimension,
    \begin{equation}
        \label{eq:piplus_rank}
        \rank(\Pi_+) = \rank(\Pi_0) = 1\,.
    \end{equation}
    Moreover,
\[
    (\Pi_\geq,\Pi_\leq) = \begin{cases}
        (\Pi_+,0),&\Re\ntr(\mathcal{U})\ge g,\\
        (0,\Pi_+),&\Re\ntr(\mathcal{U})\le-g.
    \end{cases}
\]
    The active phase is at angular distance at least \(g/2\) from \(\pi/2\), and at constant distance from the other arc boundaries.
\end{corollary}

\begin{proof}
    The rank and overlap facts follows directly from \Cref{eq:ff-half-overlap} of
    \Cref{lem:ff-two-projector}. \Cref{eq:ff-phase-trace} places \(\omega_{\mathcal{U}}\) to the
    left of \(\pi/2\) in the positive case and to its right in the negative case. 
    If \(\delta=|\omega_{\mathcal{U}}-\pi/2|\), then
\[
    \sin\delta = |\cos\omega_{\mathcal{U}}| = \frac 1 2|\Re\ntr(\mathcal{U})| \ge\frac g 2\,.
\]
    Since \(\arcsin x\ge x\) on \([0,1]\), one has \(\delta\ge g/2\).  The inclusion \(\omega_{\mathcal{U}}\in[\pi/3,2\pi/3]\) supplies the constant outer gap.
\end{proof}

\subsection{Number operator construction}
\label{sec:number_operator}

We build the number operator by defining a unitary operator $W$ which implements a change-of basis
from the physical basis to the desired basis of the Hamiltonian. If $N_\mathrm{std}$ is the physical
basis number operator, $W N_\mathrm{std} W^\dagger$ will give the number operator for $H$. When we
refer to the ``number-basis'', this refers to the basis in which $H$ is diagonal. In particular, the
number-basis form of the Hamiltonian is $W^\dagger H W$. Recall from \Cref{cor:ff-trace-switch} that
the projector $\Pi_+$ differentiates between the YES and NO cases of the distinguishing problem. We
first show that we can rotate from $\Pi_0$, which is diagonal in the physical basis, to $\Pi_+$.
Intuitively, this is possible since $\rank(\Pi_0) = \rank(\Pi_+)$. We have defined, 
\[
\Pi_0 = \ketbra{0^{\,2m+2}}{0^{\,2m+2}}_{\reg C}\,.
\] 
Let
\begin{equation}
    \label{eq:ff-K}
    R_+=2\Pi_+-\I, \qquad R_0=2\Pi_0-\I, \qquad K=R_+R_0\,.
\end{equation}
We can characterize the spectrum of $K$ via \Cref{lem:ff-two-projector}. In the notation of that
lemma, we have $d = 2^{M_{\reg C}}$ and $\rank(P) = \rank(Q) = 1$. This implies that there is no
$-1$ eigenspace. Additionally, by \Cref{eq:ff-principal-half}, the principal angle between the
supports of \(\Pi_0\) and \(\Pi_+\) is \(\pi/4\). 
Therefore, \(K\) has spectrum contained in \(\{1,\ii,-\ii\}\). Define
\begin{equation}
    \label{eq:ff-Urot-definition}
    f(1)=1,\qquad f(\ii)=\e^{\ii\pi/4},\qquad f(-\ii)=\e^{-\ii\pi/4}, \qquad
    U_{\mathrm{rot}}=f(K)\,.
\end{equation}

\begin{lemma}[Primary rotation operator]
    The unitary $U_\mathrm{rot} = f(K)$ satisfies
    \begin{equation}
        \label{eq:ff-Urot-maps}
        U_{\mathrm{rot}}\Pi_0U_{\mathrm{rot}}^\dagger=\Pi_+\,.
    \end{equation}
\end{lemma}

\begin{proof}
    Consider each two-dimensional invariant subspace from \(\Pi_0 \Pi_1 \Pi_0\), spanned by vectors \(\ket u, \ket v\). 
    In this basis, the rotation \(\mathcal G_{\mathcal{U}}\) has eigenvectors
\[
    \ket{c_+} = \frac 1 {\sqrt 2} (\ket u - \ii \ket v) \quad \text{and} \quad \ket{c_-} = \frac
    1 {\sqrt 2} (\ket u  + \ii \ket v)\,.
\]
    Since \(\Pi_+\) is the projector onto \(\ket{c_+}\), within this subspace,
\[
    \Pi_0= \begin{pmatrix}1&0\\0&0\end{pmatrix}, \qquad
    \Pi_+=\frac 1 2 \begin{pmatrix}1&\ii\\-\ii&1\end{pmatrix}.
\]
    Then
\[
    K=-\ii \begin{pmatrix}0&1\\1&0\end{pmatrix}, \qquad f(K) = \exp\left[ -\frac{\pi\ii}{4}
    \begin{pmatrix}0&1\\1&0\end{pmatrix}
    \right]\,.
\]
\(f(K)\) maps \(\ket u \rightarrow \ket{c_+}\), and thus \(U_{\text{rot}} \ket u\) for \(\ket u \in \supp \Pi_0\) spans \(\supp \Pi_+\). 
    On the orthogonal complement, the only remaining eigenspace has eigenvalue $+1$ so \(K=\I\) and \(f(K)=\I\).
    This proves the operator identity.
\end{proof}

Next, given that we have rotated into the basis of $\Pi_+$, we can work with the projectors
$\Pi_{\leq}$ and $\Pi_\geq$. Define the controlled flips $F_R$ and $F_S$ which perturb either the $n_{\reg
R}$ or $n_{\reg S}$ operator, depending on whether $\Re \ntr(\mathcal{U}) \geq g$ or $\Re \ntr(\mathcal{U}) \leq -g$.
\begin{align}
    F_R &= (\I-\Pi_{\leq})\otimes\I_R + \Pi_{\leq}\otimes X_R\,, \label{eq:ff-FR}\\
    F_S &= (\I-\Pi_{\geq})\otimes\I_S + \Pi_{\geq}\otimes X_S\,. \label{eq:ff-FS}
\end{align}
$F_R, F_S$ commute since \(\Pi_{\leq}\Pi_{\geq}=0\). The final change-of-basis for the number
operator is given by the unitary \(W\), defined as
\begin{equation}
    \label{eq:ff-W}
    W=F_R F_S U_{\mathrm{rot}}\,.
\end{equation}
The total system contains \(M_{\reg C}+2=2m+4\) qubits, as claimed.
The circuit is shown in \cref{fig:ff-W-circuit}.

\begin{figure}[tbp]
    \centering
    \input{W_circuit.tex}
    \caption{Circuit for \(W=F_RF_SU_{\mathrm{rot}}\) from \eqref{eq:ff-W}, read from left to right.
    Registers \(R,S\) are single qubits, while \(C=(a,b,L_1,L_2)\) contains \(2m+2\) qubits.
    The rotation is \(U_{\mathrm{rot}}=f(K)\), where \(K=(2\Pi_+-\I)(2\Pi_0-\I)\),
    \(f(1)=1\), and \(f(\pm\ii)=\e^{\pm\ii\pi/4}\), as defined in \eqref{eq:ff-K} and
    \eqref{eq:ff-Urot-definition}. The projectors are \(\Pi_{\geq}=\mathbf1_{\{\pi/4<\arg
    z<\pi/2\}}(\mathcal G_{\mathcal{U}})\) and \(\Pi_{\leq}=\mathbf1_{\{\pi/2<\arg z<3\pi/4\}}(\mathcal
    G_{\mathcal{U}})\), defined in \eqref{eq:ff-PiL} and \eqref{eq:ff-PiR}. A dot labeled \(\Pi\)
    denotes the coherent subspace-controlled flip \((\I-\Pi)\otimes\I+\Pi\otimes X\), giving \(F_S\)
    and \(F_R\) from \eqref{eq:ff-FS} and \eqref{eq:ff-FR}.}
    \label{fig:ff-W-circuit}
\end{figure}
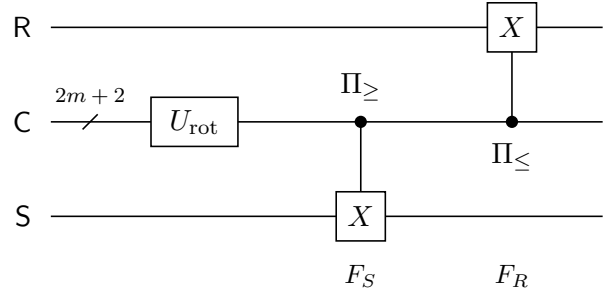

\begin{lemma}[Number-basis form of the Hamiltonian]
    \label[lemma]{lem:ff-standard-H}
    Recall from \Cref{eq:ff-physical-H} that $H = 16n_R+17n_S-\frac{17}{2}\I$. If \(\Re\ntr(\mathcal{U})\ge g\), then
    \begin{equation}
        \label{eq:ff-positive-standard}
        W^\dagger H W = 16n_R+17\Lambda_{\reg S}-\frac{17}{2}\I, \qquad \Lambda_{\reg S} = (\I-n_S) \otimes (\Pi_0)_{\reg C} +n_S \otimes (\I-\Pi_0)_{\reg C}\,.
    \end{equation}
    If \(\Re\ntr(\mathcal{U})\le-g\), then
    \begin{equation}
        \label{eq:ff-negative-standard}
        W^\dagger H W = 17n_S+16\Lambda_{\reg R}-\frac{17}{2}\I, \qquad \Lambda_{\reg R} = (\I-n_R) \otimes (\Pi_0)_{\reg C}+n_R \otimes (\I-\Pi_0)_{\reg C}\,.
    \end{equation}
    Moreover, there exists an operators $V_+, V_-$ satisfying $\|V_\pm\| = 1$ and,
    \begin{equation}
        \label{eq:ff-H-final-form}
        W^\dagger H W = 
        \begin{cases}
            16 n_{\reg R} \otimes \I + \I_{\reg{R}} \otimes \frac {17} 2 (V_+)_{\reg{SC}} & \Re
            \ntr(\mathcal{U}) \geq g\,,\\
            17 n_{\reg S} \otimes \I + \I_{\reg{S}} \otimes \frac {17} 2 (V_-)_{\reg{RC}} & \Re
            \ntr(\mathcal{U}) \leq - g\,.
        \end{cases}
    \end{equation}
    Define the number operator $N = W N_\mathrm{std} W^\dagger$. The operators \(V_\pm\) commute with the standard occupation on their respective spectator registers, and hence \([H,N]=0\).
\end{lemma}

\begin{proof}
    For a projector \(\Pi\) on a register disjoint from a single mode, set
\(F=\I\otimes(\I-\Pi)+X\otimes\Pi\). Then
    \begin{equation}
        \label{eq:conj_expr}
        F^\dagger(n\otimes\I)F
        =n\otimes(\I-\Pi)+(\I-n)\otimes\Pi.
    \end{equation}
    In the positive case, \cref{cor:ff-trace-switch} gives \(\Pi_{\geq}=\Pi_+\) and \(\Pi_{\leq}=0\).
    Thus \(F_R=\I\) and applying \Cref{eq:conj_expr} yields
\[
    F_S^\dagger n_S F_S = (\I-n_S) \otimes \Pi_+ +n_S\otimes (\I-\Pi_+)\,.
\]
    Conjugating this expression by \(U_{\mathrm{rot}}\) and using \eqref{eq:ff-Urot-maps} replaces \(\Pi_+\) by \(\Pi_0\), proving \eqref{eq:ff-positive-standard}. 
    The negative case interchanges the roles: \(\Pi_{\leq}=\Pi_+\), \(\Pi_{\geq}=0\), and the same calculation acts on \(R\).
    
    Now we show \Cref{eq:ff-H-final-form}. The computational basis number operator is defined as
$N_\mathrm{std} = N_{\mathrm{std}, \reg C} + n_{\reg S} + n_{\reg R}$. In the positive case,
    \eqref{eq:ff-positive-standard} can be written as
\[
    W^\dagger H W = 16n_{\reg R} + \frac{17}{2} \left(2\Lambda_{\reg S}-\I\right)\,.
\]
    Set \(V_+=2\Lambda_{\reg S}-\I\). Since \(\Lambda_{\reg S}\) is a projector, the eigenvalues of \(V_+\) are \(\pm1\), so
\(\|V_+\|=1\), and \(V_+\) acts only on \(\reg C,S\). To show the desired commutation relation,
    recall $\Lambda_{\reg S} = \Pi_0(\I - n_{\reg S}) + (\I - \Pi_0)n_{\reg S}$ from
    \Cref{eq:ff-positive-standard} and note that $\Pi_0 = \ketbra{0^{M_{\reg C}}}{0^{M_{\reg C}}}_{\reg C}$ is diagonal in the
    computational basis. Thus, $\Lambda_{\reg S}$ commutes with $N_\mathrm{std}$.
    
    In the negative case,
\[
    W^\dagger H W = 17n_S + \frac{17}{2} \left(\frac{32}{17}\Lambda_{\reg R}-\I\right)\,.
\]
    Set \(V_-=(32/17)\Lambda_{\reg R}-\I\).  
    Its two eigenvalues are \(-1\) and \(15/17\), so \(\|V_-\|=1\), and it acts only on \(\reg C,\reg R\). 
    The same argument gives \([V_-,N_\mathrm{std}]=0\).
    
    Finally, in the number basis both Hamiltonians are the sum of the number operator $n_{\reg R}$ or $n_{\reg S}$ and $V_\pm$. Therefore, we conclude that
\[
    [W^\dagger H W,N_{\mathrm{std}}]=0\,.
\]
    Conjugation by \(W\), together with \(N=WN_{\mathrm{std}}W^\dagger\), proves \([H,N]=0\).
\end{proof}

\subsection{Partition polynomial}
\label{sec:ff-polynomial}

Continue to use the complex inverse temperature \(s=\beta+\ii t\) and the physical fugacity \(z\) from \cref{sec:free-model}.  
For brevity put
\[
\zeta_E=\e^{-sE}, \qquad Z(s,z)=\Tr\bigl(z^N\e^{-sH}\bigr)\,.
\]
Since \(N=WN_{\mathrm{std}}W^\dagger\), cyclicity of trace gives
\begin{equation}
    \label{eq:ff-standard-partition}
    Z(s,z) = \Tr\left( z^{N_{\mathrm{std}}}\e^{-sW^\dagger H W} \right)\,.
\end{equation}

We define the generating polynomials of the register \(\reg C\). 
Recall that \(\reg C\) is the auxiliary register on \(M_{\reg C}=2m+2\) qubits.
\begin{equation}
    \label{eq:ff-source-polynomials}
    f(z)=\Tr_{\reg C}\left( z^{(N_{\mathrm{std}})_{\reg C}}\right) = (1+z)^{M_{\reg C}}, \qquad
    f_{\mathrm{signal}}(z) = \Tr_{\reg C}\left( \Pi_0 z^{(N_{\mathrm{std}})_{\reg C}} \right) = z^{M_{\reg C}}\,.
\end{equation}
The latter equation follows since \(\Pi_0 = \ketbra{0^{M_{\reg C}}}{0^{M_{\reg C}}}_{\reg C}\).
\begin{lemma}[Exact factorization in physical fugacity]
    \label[lemma]{lem:ff-Z-factorization}
    Let
\[
    (E_\star,\bar E) = \begin{cases}
        (16,17),&\Re\ntr(\mathcal{U})\ge g,\\
        (17,16),&\Re\ntr(\mathcal{U})\le-g.
    \end{cases}
\]
    Then
    \begin{equation}
        \label{eq:ff-Z-factorization}
        Z(s,z) = \e^{17s/2}(1+z\zeta_{E_\star})
        \left[
        f_\mathrm{signal}(z)(\zeta_{\bar E}+z) + \bigl(f(z)-f_\mathrm{signal}(z)\bigr)(1+z\zeta_{\bar E})
        \right].
    \end{equation}
\end{lemma}

\begin{proof}
    Consider first the positive case.
    By \Cref{eq:ff-H-final-form}, the \(R\)-mode contributes
\[
    \Tr_{\reg R}\left(z^{n_{\reg R}}\e^{-16s n_\reg{R}}\right) = 1+z\zeta_{16}\,.
\]
    The projector \(\Lambda_{\reg S}\) is block diagonal across \(\Pi_0\), \(\I - \Pi_0\) so,
    \begin{equation}
        \label{eq:ff-block-exponential}
        \e^{-17s \Lambda_{\reg S}} = \Pi_0\e^{-17s(\I-n_{\reg S})} + (\I-\Pi_0)\e^{-17sn_{\reg S}}\,.
    \end{equation}
    Thus,
    \begin{align*}
        &\Tr_{\reg C \reg S}\left[ z^{N_{\mathrm{std},\reg C}+n_{\reg S}}\e^{-17s\Lambda_{\reg S}} \right]\\
        &\quad= \Tr_\reg{C}\left( \Pi_0z^{N_{\mathrm{std},\reg{C}}} \right) \Tr_\reg{S}\left( z^{n_\reg{S}}\e^{-17s(\I-n_\reg{S})}
        \right)\\
        &\qquad+ \Tr_\reg{C}\left( (\I-\Pi_0)z^{N_{\mathrm{std},\reg{C}}} \right) \Tr_\reg{S}\left( z^{n_\reg{S}}\e^{-17sn_\reg{S}}
        \right)\\
        &\quad= f_{\mathrm{signal}}(z)(\zeta_{17}+z) + \bigl(f(z)-f_{\mathrm{signal}}(z)\bigr)(1+z\zeta_{17})\,.
    \end{align*}
    The scalar \(-17\I/2\) contributes \(\e^{17s/2}\). 
    Multiplying the three factors proves the positive
    case. 
    The negative case is identical after interchanging \(\reg S\) with \(\reg R\) and \(17\) with \(16\).
\end{proof}

We use the inverse-fugacity convention \(w=-1/z\) already introduced in
\cref{eq:inverse-fugacity-coordinate}, and define
\[
P_H(s,w)=w^MZ(s,-1/w).
\]
This is a degree-\(M\) polynomial in \(w\). 
To show that SPICES succeeds on \(P_H\), we need to show that the only root in the target annulus is \(E_\star\). 
Equivalently, we show that we can decompose \(P_H\) into a protected and spectator factor, with \(E_\star\) a root of the protected factor, and the spectator remaining zero free in the annulus.

\begin{proposition}[Protected factor and spectator factor]
    \label[proposition]{prop:ff-inverse-factorization}
\(P_H(s,w)\) has a decomposition as
    \begin{equation}
        \label{eq:ff-inverse-factorization}
        P_H(s,w) = \e^{17s/2} (w-\zeta_{E_\star}) (w-1)^{M_{\reg C}} \mathcal R_{\bar E}(s,w)\,.
    \end{equation}
    where
    \begin{equation}
        \label{eq:ff-Re}
        \mathcal R_{\bar E}(s,w) = w-\zeta_{\bar E} + (\zeta_{\bar E}-1)\frac{w+1}{(w-1)^{M_{\reg C}}}\,.
    \end{equation}
    Although \(\mathcal R_{\bar E}\) is written as a rational function, \((w-1)^{M_\reg{C}}\mathcal R_{\bar E}\) is a polynomial after cancellation of the denominator.
\end{proposition}

\begin{proof}
    First, write
\[
    \frac{f_{\mathrm{signal}}(-1/w)}{f(-1/w)} =
    \frac{ (-1/w)^{M_{\reg C}}}{ (1-1/w)^{M_{\reg C}} } = \frac1{(w-1)^{M_{\reg C}}}\,.
\]
    Moreover,
\[
    f(-1/w) = w^{-M_{\reg C}}(w-1)^{M_{\reg C}}.
\]
    Factoring \(f(z)\) from the square bracket in \eqref{eq:ff-Z-factorization} (which corresponds to the spectator contribution from \(\Lambda_{\reg R}, \Lambda_{\reg S}\)), substituting \(z=-1/w\), and collecting the remaining factor \(1/w\) gives
\[
    f(-1/w)\frac 1 w \left[ w-\zeta_{\bar E} +(\zeta_{\bar E}-1)\frac{w+1}{(w-1)^{M_{\reg C}}} \right]\,.
\]
    and thus,
\[
    Z(s, -1/w) = e^{17s/2} \left(1 - \frac{\zeta_{E_\star}}{w}\right) \cdot f(-1/w)\frac 1 w
    \left[ w-\zeta_{\bar E} +(\zeta_{\bar E}-1)\frac{w+1}{(w-1)^{M_{\reg C}}} \right]\,.
\]
    Writing \((1-\zeta_{E_\star}/w)=(w-\zeta_{E_\star})/w\) and multiplying by \(w^M = w^{M_{\reg C} +
    2}\) cancels all powers of \(w\) in the denominator and proves the formula.
\end{proof}

\subsection{Zero-freeness of the spectator polynomial}
\label{sec:ff-zero-free}

\newcommand{\varq}{{M_{\reg C}}}

The factor
\[
P_{\mathrm{spec}}(s,w) = (w-1)^{M_{\reg C}}\mathcal R_{\bar E}(s,w)
\]
from \Cref{eq:ff-inverse-factorization} is the spectator inverse-fugacity polynomial, up to a
nonzero scalar. We first show directly that it has no zero near the two SPICES contours. This
construction-specific argument replaces the generic small-perturbation condition in
\cref{cor:log-gap-perturbative-stability}.
The argument has two parts.
An elementary estimate shows that the normalized spectator polynomial is close to \(-1\) near the origin (\cref{lem:ff-Re-bound}); this gives a zero-free collar (\cref{prop:ff-zero-free}) and hence exactly one root in the target annulus (\cref{cor:ff-one-root}).
The same estimate then gives the signal floor (\cref{prop:ff-signal}) and the radial separation of all roots (\cref{cor:ff-root-separation}), which are the remaining hypotheses of \cref{thm:spices-recovery}.

\begin{lemma}[Elementary bound near the origin]
    \label[lemma]{lem:ff-Re-bound}
    Let \(M\ge4\), let \(M_{\reg C}\) be even with \(2\le M_{\reg C}\le M\), \(\beta=\log(3M)\), \(\bar E\in\{16,17\}\), and
\[
    |w|\le r_0=\e^{-14\beta}=(3M)^{-14}.
\]
    Then, uniformly for every real \(t\),
    \begin{equation}
        \label{eq:ff-Re-close}
        |\mathcal R_{\bar E}(\beta+\ii t,w)+1|<\frac 1 4\,.
    \end{equation}
\end{lemma}

\begin{proof}
    As before, write \(x=\e^{-(\beta+\ii t)\bar E} = \zeta_{\bar E}\).  Since \(\varq\) is even,
\((w-1)^\varq=(1-w)^\varq\), and therefore the above \Cref{prop:ff-inverse-factorization}
    yields
\[
    \mathcal R_{\bar E}(s,w)+1 = w +(1-x) \left[ 1-\frac{1+w}{(1-w)^\varq } \right]\,.
\]
    For \(|w|\le r_0<1/2\),
\[
    \left| \frac{1+w}{(1-w)^\varq }-1 \right| \le \frac{1+r_0}{(1-r_0)^\varq }-1\,.
\]
    The inequality \(-\log(1-r_0)\le2r_0\) gives $(1-r_0)^{-\varq } \leq \e^{2 \varq r_0}$. Because
\(2\varq r_0 \leq 1\), \(\e^{2\varq r_0}\le1+4\varq r_0\). Additionally, using \(\varq \le M\) and
\(r_0 = (3M)^{-14} \leq 1/4\) gives,
\[
    \frac{1+r_0}{(1-r_0)^\varq }-1 \le (1+r_0)(1+4\varq r_0)-1 \le6Mr_0\,.
\]
    Finally, since \(|1-x| \leq 1 + |x| \le 2\), we conclude that
\[
    |\mathcal R_{\bar E}(s,w)+1| \le r_0+12Mr_0 \le13M(3M)^{-14} <\frac 1 4\,.
\]
\end{proof}

Since the normalized spectator polynomial stays within \(1/4\) of \(-1\), it cannot vanish in this disk, and neither can the remaining factor \((w-1)^{M_{\reg C}}\).

\begin{proposition}[Spectator zero-free collar]
    \label[proposition]{prop:ff-zero-free}
    For every real \(t\), the spectator polynomial has no zero in
\[
    \e^{-19\beta} \le |w| \le \e^{-14\beta}\,.
\]
    In fact, it has no zero in the entire disk \(|w|\le\e^{-14\beta}\).
\end{proposition}

\begin{proof}
    The factor \((w-1)^{M_{\reg C}}\) is nonzero in the stated disk.
    By \cref{lem:ff-Re-bound}, \(|\mathcal R_{\bar E}+1|<1/4\), so \(|\mathcal R_{\bar E}|>3/4\).
    Thus the second factor is also nonzero.
\end{proof}

The collar contains no spectator root, so the only root in the window annulus is the protected one.

\begin{corollary}[One simple target root]
    \label[corollary]{cor:ff-one-root}
    The annulus \(A_\Omega\) from \cref{subsec:spectator-sector-model}, with \(\Omega=[15,18]\), contains exactly the simple root
\[
    w_\star(t)=\e^{-(\beta+\ii t)E_\star}\,.
\]
\end{corollary}

\begin{proof}
    The factorization \eqref{eq:ff-inverse-factorization} supplies the displayed root.  
    Its modulus lies strictly between \(\e^{-18\beta}\) and \(\e^{-15\beta}\), because \(E_\star\in\{16,17\}\).
    Every root of the spectator factor is excluded by \cref{prop:ff-zero-free}. 
    The spectator factor is nonzero at \(w_\star\), so the protected root has multiplicity one.
\end{proof}

The same zero-free disk controls the normalized signal on the two contour circles.
The next proposition bounds it from below uniformly in \(t\) and in the angle \(\vartheta\), which is the signal-floor hypothesis of \cref{thm:spices-recovery} with \(b=9/40\).

\begin{proposition}[Uniform normalized-signal bound]
    \label[proposition]{prop:ff-signal}
    For \(\mu\in\{15,18\}\), \(z=\e^{\beta\mu+\ii\vartheta}\), and every real \(t,\vartheta\),
\[
    \left| \frac{Z(\beta+\ii t,z)} {Z(\beta,|z|)} \right| \ge\frac 9 {40}\,.
\]
\end{proposition}

\begin{proof}
    Recall that $Z(s, -1/w) = w^{-M} P_H(s, w)$. By \Cref{prop:ff-inverse-factorization},
    \[
    Z(s,z)=w^{-M}\e^{17s/2} (w-\zeta_{E_\star}(s))(w-1)^{M_{\reg C}} \mathcal R_{\bar E}(s,w)\,.
    \]
    Apply this identity for \( R_\mu=\e^{-\beta\mu}\) and at $w = -1/z$ and $-1/|z| = - R_\mu$.
    Use \(|w|=|w_0|= R_\mu\) and that the powers of \(w\) and the scalar \(\e^{17s/2}\) in the
    numerator have unit modulus after division by the \(t=0\) denominator.
    
    For the protected factor, the reverse triangle inequality gives
    \begin{equation}
        \frac{|w-\zeta_{E_\star}(\beta+\ii t)|} {|w_0-\zeta_{E_\star}(\beta)|} \ge
        \frac{ \bigl| R_\mu-\e^{-\beta E_\star}\bigr| }{  R_\mu+\e^{-\beta E_\star} } = 
        \frac{ \bigl|1-\e^{-\beta \lvert E_\star - \mu \rvert}\bigr| }{ 1+\e^{-\beta \lvert E_\star - \mu \rvert}} \ge
        \frac{1-\e^{-\beta}}{1+\e^{-\beta}} \ge\frac 1 2.
        \label{eq:ff-protected-signal}
    \end{equation}
    For the middle equality, divide the numerator and denominator by \(R_\mu\), which gives \(|1-x|/(1+x)\) with \(x=\e^{-\beta(E_\star-\mu)}\). For \(\mu=15\) we have \(E_\star>\mu\), so \(x=\e^{-\beta|E_\star-\mu|}\). For \(\mu=18\) we have \(E_\star<\mu\), so \(x>1\), and multiplying the numerator and denominator by \(x^{-1}=\e^{-\beta|E_\star-\mu|}\) gives the same expression.
    The following inequality uses the fact that \(\lvert E_\star - \mu \rvert \geq 1\) and the right side is increasing in this quantity.
    The last inequality uses \(\e^{-\beta}=1/(3M)\le1/3\).
    
    For the \((w-1)^{M_\reg{C}}\) factor,
\[
    \left| \frac{w-1}{w_0-1} \right|^{M_\reg{C}} \ge \left(\frac{1- R_\mu}{1+ R_\mu}\right)^{M_\reg{C}}\,.
\]
    Here \( R_\mu\le(3M)^{-15}\). 
    Since \(\log[(1- R_\mu)/(1+ R_\mu)]\ge-4 R_\mu\) for \(0\le R_\mu\le1/2\), the last expression is at least
    \begin{equation}
        \label{eq:ff-B-signal}
        \e^{-4M_\reg{C} R_\mu}>\frac 3 4\,.
    \end{equation}
    
    Finally, both \(w\) and \(w_0\) lie in the disk of \cref{lem:ff-Re-bound}.  
    Hence
    \begin{equation}
        \label{eq:ff-R-signal}
        \frac{|\mathcal R_{\bar E}(\beta+\ii t,w)|} {|\mathcal R_{\bar E}(\beta,w_0)|} \ge
        \frac{1-1/4}{1+1/4} = \frac 3 5.
    \end{equation}
    Multiplying \eqref{eq:ff-protected-signal}, \eqref{eq:ff-B-signal}, and \eqref{eq:ff-R-signal} gives \((1/2)(3/4)(3/5)=9/40\).
\end{proof}

Finally, every root of \(P_H\), protected or spectator, is radially separated from both contour circles by a factor of at least \(3M\), which is the radial-separation hypothesis with \(q=1/(3M)\).

\begin{corollary}[All-root radial separation]
    \label[corollary]{cor:ff-root-separation}
    For either contour radius
\[
    R_\mu=\e^{-\beta\mu},
    \qquad
    \mu\in\{15,18\},
\]
    and every root \(\zeta\) of \(P_H(\beta+\ii t,\cdot)\),
    \begin{equation}
        \label{eq:ff-root-separation}
        \min\left\{ \frac{|\zeta|}{R_\mu}, \frac{R_\mu}{|\zeta|} \right\} \le \e^{-\beta} =
        \frac1{3M}\,.
    \end{equation}
\end{corollary}

\begin{proof}
    \Cref{prop:ff-inverse-factorization} gives the factorization,
    \[
    P_H(s,w)=e^{17s/2}(w-e^{-sE_\star})P_{\mathrm{spec}}(s,w)\,,
    \]
    yielding the ``protected'' root $e^{-sE_\star}$ and spectator roots from
    $P_{\mathrm{spec}}(s,w)$. For the protected root, the energy distance from either endpoint
    is at least one.  For a spectator root, \cref{prop:ff-zero-free} gives
    \(|\zeta|>\e^{-14\beta}\).  Therefore
    \[
    \frac{R_{15}}{|\zeta|}<\e^{-\beta}, \qquad \frac{R_{18}}{|\zeta|}<\e^{-4\beta}\,.
    \]
    These bounds prove the claim.
\end{proof}

\subsection{Establishing the access model}
\label{sec:ff-compiler}
We establish the three primitives in \cref{def:ff-operational-access}.
First we compile \(W\), then use it for number measurement and joint evolution, and finally
prepare the Gibbs state by standard filtering from an easily prepared reference state.

\begin{lemma}[Coherent gapped phase classification]
    \label[lemma]{lem:ff-phase-filter}
    Let \(\mathcal{U}\) be a unitary with uniform circuits for controlled \(\mathcal{U}\) and controlled \(\mathcal{U}^\dagger\).  
    Let \(\mathcal A\) be a fixed finite union of intervals on the unit circle, with efficiently computable endpoints.
    Suppose every eigenphase of \(\mathcal{U}\) is at angular distance at least \(\Delta>0\) from any boundary of \(\mathcal A\). 
    For every \(0<\eta<1/10\), there is a uniform circuit \(\mathcal F_{\mathcal A,\eta}\) such that, on every eigenvector \(\mathcal{U}\ket{\psi_\theta}=\e^{\ii\theta} \ket{\psi_\theta}\),
\[
    \left\| \mathcal F_{\mathcal A,\eta} \ket{\psi_\theta}\ket{0^a} - \ket{\psi_\theta}
    \ket{\mathbf1_{\{\theta\in\mathcal A\}}} \ket{0^{a-1}} \right\| \le\eta\,.
\]
    It uses
\[
    O\left( \Delta^{-1}\log\frac1\eta \right)
\]
    controlled calls to \(\mathcal{U}\) or \(\mathcal{U}^\dagger\) and up to polynomial overhead in the work-register size.
    The same guarantees hold with an external control.
\end{lemma}

\begin{proof}
    Standard phase estimation (see, e.g., \cite[Sec.~5.2]{NielsenChuang2010}) estimates an eigenphase to
    accuracy \(\Delta/4\) with probability at least \(3/4\), using \(O(\Delta^{-1})\) controlled
    queries. The boundary-gap assumption therefore gives the correct membership bit with the same
    probability. We can then repeat phase estimation \(O(\log(1/\eta))\) times, and then compute the
    majority bit to reduce the error probability to \(p \leq \eta^2/16\).
    
    Finally, we copy the output bit into a fresh qubit and uncompute all work registers. The standard
    compute-copy-uncompute bound (see, e.g., \cite[Theorem 4.14]{bennett1997strengths}) gives norm error
    at most \(2\sqrt p \leq \eta/2\). Allowing total gate-synthesis error \(\eta/2\) gives the
    claimed error and query complexity.
\end{proof}

This filter is the only quantum subroutine we need to implement the projectors that define \(W\).
The next lemma uses it to compile \(W\) itself.

\begin{lemma}[Approximation for the mode change]
    \label[lemma]{lem:ff-W-compiler}
    For every integer \(r\ge1\), a deterministic classical algorithm outputs a finite circuit \(\widehat W_r\), acting on the \(M\) system qubits and polynomially many clean work qubits, such that
    \begin{equation}
        \left\|
        \widehat W_r
        (\ket{\psi}\ket{0^a})
        -
        (W\ket{\psi})\ket{0^a}
        \right\|
        \le2^{-r}
        \label{eq:ff-W-compiler}
    \end{equation}
    for every system state \(\ket{\psi}\). 
    The same holds for the adjoint and controlled versions.  
    Its gate count and run time are polynomial in \(|\mathcal{U}|,m,g^{-1},r\).
\end{lemma}

\begin{proof}
    The reflections about \(\Pi_0\) and \(\Pi_1\), and hence controlled
\(\mathcal G_{\mathcal{U}}\), have polynomial-size circuits by their definitions.
    By \cref{cor:ff-trace-switch}, the relevant phase boundaries have gap at least
\(\min\{\pi/12,g/2\}\). Thus \cref{lem:ff-phase-filter}, followed by a phase
    or a controlled \(X\) and uncomputation, implements \(2\Pi_+-\I\), \(F_R\),
    and \(F_S\) to clean-workspace norm error \(\eta\) with cost polynomial in
\(|\mathcal{U}|,m,g^{-1},\log(1/\eta)\). These operations preserve each eigenspace of
\(\mathcal G_{\mathcal{U}}\), so the error bound also holds for superpositions and states
    entangled with an external reference.
    
    For \(U_{\mathrm{rot}}=f(K)\), use \(K^4=\I\). Two-bit phase estimation
    identifies each of \(1,\ii,-1,-\ii\) exactly with a constant number of
    controlled calls to \(K\); apply the phase \(f\) and reverse phase estimation
    \cite[Sec.~5.2]{NielsenChuang2010}. Assign \(f(-1)=1\), since this eigenspace
    is absent. Replace each call to \(K=(2\Pi_+-\I)(2\Pi_0-\I)\) by the compiled
    reflections above, using fresh clean work registers. A hybrid argument
    bounds the total error by a constant multiple of \(\eta\).
    
    Choose \(\eta=2^{-r}/C\) for a sufficiently large absolute constant \(C\)
    and compose \(F_RF_SU_{\mathrm{rot}}\). Synthesis over a fixed universal
    gate set adds only polynomial overhead in \(r\)
    \cite{DawsonNielsen2006}. Reversing or externally controlling the same
    construction gives the adjoint and controlled versions, with their relative
    phases retained. This proves the claimed norm error and resource bounds.
\end{proof}

With \(W\) available, number measurement and joint evolution reduce to single-qubit operations conjugated by \(W\).

\begin{corollary}[Number measurement and joint evolution]
    \label[corollary]{cor:ff-number-primitives}
    For every \(r\ge1\), polynomial-size finite circuits implement, to operational error at most \(2^{-r}\):
    \begin{enumerate}
        \item projective measurement of the exact semantic \(N\), with its outcome stored in a binary register;
        \item the channel and controlled channel associated with \(\e^{-\ii tH+\ii\vartheta N}\), for polynomial-bit \(t,\vartheta\).
    \end{enumerate}
\end{corollary}

\begin{proof}
    Compile \(W,W^\dagger\) to clean-workspace norm error \(2^{-r-4}\) using
    \cref{lem:ff-W-compiler}. For number measurement, apply \(W^\dagger\),
    coherently compute \(N_{\mathrm{std}}\) into a binary register, measure that
    register, and apply \(W\). Since \(n_j=(\I+Z_j)/2\), this counts zero bits,
    not one bits. Measuring only the sum preserves coherence within each
    number eigenspace.
    
    Since \([H,N]=0\),
    \begin{equation}
        \label{eq:ff-joint-factorization}
        \e^{-\ii tH+\ii\vartheta N}
        =\e^{-\ii tH}W\e^{\ii\vartheta N_{\mathrm{std}}}W^\dagger.
    \end{equation}
    Here \(\e^{-\ii tH}\) is a scalar phase times two one-qubit rotations, and
\(\e^{\ii\vartheta N_{\mathrm{std}}}=\prod_{j=1}^{M}\e^{\ii\vartheta n_j}\).
    Compile these rotations to total operator norm error at most \(2^{-r-3}\);
    in the controlled version implement the scalar phase on the control as well.
    Gate synthesis has cost polynomial in \(r,M\) and the query bit length
    \cite{DawsonNielsen2006}.
    
    For isometries, diamond distance between the induced channels is at most
    twice their operator norm distance, also after measurement or discarding
    work registers. A hybrid argument therefore bounds the operational error
    of each construction by \(2^{-r}\). All circuit sizes and compilation times
    are polynomial in \(|\mathcal{U}|,m,g^{-1},r\) and the query bit length.
\end{proof}

The remaining primitive is Gibbs-state preparation.
It is the only one whose cost depends on \(\beta\), and the dependence is exponential, which is why the hardness result is stated at \(\beta=\log(3M)\).

\begin{lemma}[Grand-canonical Gibbs preparation]
    \label[lemma]{lem:ff-gibbs-preparation}
    For the constructed pair \((H,N)\), a deterministic classical compiler, given
\(\mathcal{U},g,r\) and rational \(\beta\ge0,\mu\) of total bit length \(\ell_{\mathrm{qry}}\), produces
    a circuit whose output \(\widehat\rho\) satisfies
\[
    \frac12\|\widehat\rho-\rho_{\beta,\mu}\|_1\le2^{-r},
    \qquad
    \rho_{\beta,\mu}=\frac{\e^{-\beta(H-\mu N)}}{Z(\beta,\e^{\beta\mu})}.
\]
    The compilation time, ancilla count, and worst-case gate count are bounded by
\[
    \e^{33\beta}\operatorname{poly}(|\mathcal{U}|,m,g^{-1},r,\ell_{\mathrm{qry}},1+\beta),
\]
    and hence are polynomial in \(|\mathcal{U}|,m,g^{-1},r,\allowbreak\ell_{\mathrm{qry}},\e^\beta\).
    In particular, they are polynomial in \(|\mathcal{U}|,m,g^{-1},r,\allowbreak\ell_{\mathrm{qry}}\) at
\(\beta=\log(3M)\), with \(\ell_{\mathrm{qry}}\) then denoting the bit length of \(\mu\).
    The circuit uses no preliminary trace-sign decision and measures no branch register.
\end{lemma}

\begin{proof}
    Apply the standard Gibbs-filtering step of
    \cite[Eqs.~(4)--(5)]{poulin2009sampling}, starting from the reference state
    \begin{equation}
        \label{eq:starting_state}
        \sigma_{\beta,\mu}
        =W_{\mathcal{U}}\left(\frac{\e^{\beta\mu n}}{1+\e^{\beta\mu}}\right)^{\otimes M}W_{\mathcal{U}}^\dagger
        =\frac{\e^{\beta\mu N}}{(1+\e^{\beta\mu})^M}.
    \end{equation}
    Here \(n=(\I+Z)/2\) acts on one qubit. Using \Cref{lem:ff-W-compiler} with inputs \(\mathcal{U},g\), compile
    the corresponding mode-change unitary \(W_{\mathcal{U}}\) to the precision specified below. Prepare the product
    state and apply \(W_{\mathcal{U}}\) to obtain \(\sigma_{\beta,\mu}^{(\mathcal{U})}\). First add
    a shift to $H$ so that it becomes positive definite; $H_0 = H + (17/2) \I$ satisfies \(0 \preceq H_0
    \preceq 33\I\). Since \(H_0\) is diagonal on the two physical qubits \(R,S\), four controlled ancilla
    rotations implement the Kraus operator \(F=\e^{-\beta H_0/2}\), corresponding to measuring the POVM
$\{F^\dagger F, \I - F^\dagger F\}$. Conditioning on the first outcome, the result is,
\[
    \frac{F\sigma_{\beta,\mu}F^\dagger}
    {\Tr(F\sigma_{\beta,\mu}F^\dagger)}=\rho_{\beta,\mu}\,,
\]
    with success probability $\Tr(\e^{-\beta H_0}\sigma_{\beta,\mu}) \ge\e^{-33\beta}$.
    
    For target accuracy \(\eta=2^{-r}\), cap the number of independent attempts at
\(T_{\mathrm{att}} = \lceil \e^{33\beta} \log(4/\eta) \rceil\), rounding this cap upward
    within a constant factor if necessary. Return the first successful state,
    or any fixed state if all attempts fail. The latter event has probability
    at most \((1-p_{\mathrm{succ}})^T_{\mathrm{att}}\le\eta/4\).
    Approximate each complete trial, including reference preparation and the
    success instrument, to channel error at most \(\eta/(4T_{\mathrm{att}})\), so a union bound implies the total
    error is $T_{\mathrm{att}} \cdot \eta /(4T_{\mathrm{att}}) = \eta / 4$.
    
    To get the desired accuracy $\eta / 4T_{\mathrm{att}}$, we need to implement each single qubit state in
    \Cref{eq:starting_state} to precision $\eta/(4T_{\mathrm{att}}M)$, requiring bit accuracy $O(\log
    ((T_{\mathrm{att}}M)/\eta)) = O(r + \log T_{\mathrm{att}} + \log M) = O(r + \beta + \log M)$. For gate synthesis use
    \cref{lem:ff-W-compiler} and \cite{DawsonNielsen2006} to approximate via a fixed gate set.
    This proves the claimed error and worst-case resource bounds.
    
    For \(\beta=\log(3M)\), approximate \(\beta\) to
\(O(r+\ell_{\mathrm{qry}}+\log M)\) binary places. Since
\(\|H-\mu N\|\le49/2+M|\mu|\), a change \(d\) in \(\beta\) changes
    the Gibbs state in trace distance by at most
\((\e^{2|d|\|H-\mu N\|}-1)/2\). Choosing this at most \(\eta/4\)
    keeps the total error below \(\eta\), while
\(\e^{33\beta}=(3M)^{33}\) gives polynomial cost.
\end{proof}

\subsection{Proof of the main hardness theorem}
\label{sec:ff-main-proof}

\begin{proof}[Proof of \cref{thm:ff-main}]
    Given \(\mathcal{U},g\), output \(M=2m+4\), the Pauli list
    \eqref{eq:ff-physical-H}, the symbolic operator
\(N=WN_{\mathrm{std}}W^\dagger\) with \(W\) from \eqref{eq:ff-W}, and
    the compiler established in \cref{sec:ff-compiler}. All these descriptions
    are constructed deterministically in polynomial time without deciding the
    trace-sign branch. We verify the five claims.
    \begin{enumerate}
        \item Expanding \(n_j=(\I+Z_j)/2\) for \(j=R,S\) gives exactly
        \(H=8\I+8Z_R+(17/2)Z_S\).
        
        \item The factors of \(W\) are unitary, so conjugation preserves the
        eigenvalues and binomial multiplicities of \(N_{\mathrm{std}}\).
        The commutation relation \([H,N]=0\) follows from
        \cref{lem:ff-standard-H}.
        
        \item The same lemma gives the stated tensor decompositions, with
        \(V_+=2\Lambda_{\reg S}-\I\) and \(V_-=(32/17)\Lambda_{\reg R}-\I\), each of norm one and
        conserving spectator occupation. This is precisely the model of
        \cref{subsec:spectator-sector-model}, taking all free spectator
        energies to be zero and \(\lambda=17/2\).
        
        \item The tensor decompositions identify \(E_\star=16\) in the positive
        branch and \(E_\star=17\) in the negative branch. All free spectator
        energies are zero, so this is the unique protected energy in \([15,18]\),
        the endpoint gap is \(\Delta_\Omega=1\), and \(|\lambda|=17/2<16\).
        
        \item \Cref{cor:ff-number-primitives,lem:ff-gibbs-preparation}
        supply all three primitives to error \(2^{-r}\), with resources
        polynomial in \(|\mathcal{U}|,m,g^{-1},r,\ell_{\mathrm{qry}}\) for number measurement and joint
        evolution, and additionally in \(\e^\beta\) for Gibbs preparation.
        Since \(g^{-1}\) is polynomially bounded in \(m\), these meet
        \cref{def:ff-operational-access}. At
        \(\beta=\log(3M)\), we have \(\e^\beta=3M\), so preparation also has
        polynomial cost. In particular, setting \(r=M^2\) preserves polynomial
        size at this temperature.
    \end{enumerate}
    This proves all claimed properties of the reduction.
\end{proof}

\subsection{Compatibility with the SPICES recovery algorithm}
\label{sec:ff-spices}
We apply \cref{thm:spices-recovery} by verifying its five hypotheses for our hard family. Rather
than appealing to the generic perturbative condition in \cref{cor:log-gap-perturbative-stability},
we verify the conditions explicitly.
The following proposition collects the certificates for the hard family from the estimates of \cref{sec:ff-zero-free}.

\begin{proposition}[Contour certificates for the hard family]
    \label[proposition]{prop:ff-certificates}
    For the family of \cref{thm:ff-main}, let \(\beta=\log(3M)\), \(\Omega=[15,18]\), and \(P_H(s,w)=w^MZ(s,-1/w)\). Then for every real time \(t\):
    \begin{enumerate}[label=(\roman*)]
        \item the annulus \(A_\Omega\) contains exactly one root of \(P_H(\beta+\ii t,\cdot)\), the simple protected root \(\e^{-(\beta+\ii t)E_\star}\), and the spectator polynomial is zero-free throughout the enlarged energy annulus \([14,19]\);
        \item on both contour circles of \(\Omega\), with \(R_\mu=\e^{-\beta\mu}\) for \(\mu\in\{15,18\}\), every root \(\zeta\) satisfies
        \[
            \min\left\{\frac{|\zeta|}{R_\mu},\frac{R_\mu}{|\zeta|}\right\}\le\frac1{3M},
        \]
        so the radial separation \(q=1/(3M)\) holds in the sense of \cref{eq:spices-radial-separation};
        \item the normalized signal satisfies
        \begin{equation}
            \inf_{t,\vartheta}
            \left|
            \frac{
            Z(\beta+\ii t,\e^{\beta\mu+\ii\vartheta})
            }{
            Z(\beta,\e^{\beta\mu})
            }
            \right|
            \ge\frac9{40},
            \qquad
            \mu\in\{15,18\}.
            \label{eq:ff-signal-bound}
        \end{equation}
    \end{enumerate}
\end{proposition}

\begin{proof}
    Item (i) is \cref{cor:ff-one-root} together with \cref{prop:ff-zero-free}, which excludes spectator roots from the enlarged annulus by means of the factorization of \cref{prop:ff-inverse-factorization}. Item (ii) is \cref{cor:ff-root-separation}, and item (iii) is \cref{prop:ff-signal}.
\end{proof}

\begin{corollary}[SPICES recovers the hard protected energy]
    \label[corollary]{cor:ff-spices-hard-family}
    For the family in \cref{thm:ff-main}, let \(0 < \epsilon, \delta < 1\) be
    inverse-polynomial accuracy and failure parameters. In the operational input model of
    \cref{def:ff-operational-access}, the SPICES algorithm of \cref{thm:spices-recovery} at
\(\beta=\log(3M)\) uses polynomially many gates and samples, including
    the cost of Gibbs preparation, and returns, with probability at least
\(1-\delta\), a finitely represented probability
    measure \(\widehat\nu\) satisfying
\[
    W_1(\widehat\nu,\delta_{E_\star}) \le \epsilon.
\]
    In particular, its first moment determines \(E_\star\) to inverse-polynomial accuracy.
\end{corollary}

\begin{proof}
    Use
\[
    M=2m+4, \qquad \Omega=[15,18], \qquad L=3, \qquad \beta=\log(3M).
\]
    For \(s=\beta+\ii t\), consider the degree-\(M\) inverse-fugacity polynomial
\[
    P_H(s,w)=w^M Z(s,-1/w).
\]
    We verify the five hypotheses of \cref{thm:spices-recovery} in order.
    
    \begin{enumerate}
        \item \emph{The target annulus contains precisely the protected roots.} The annulus associated
        with \(\Omega\) is
        \[
        A_\Omega = \left\{ w\in\C: \e^{-18\beta} < |w| < \e^{-15\beta} \right\}.
        \]
        By \cref{prop:ff-certificates}(i), for every real \(t\) this annulus contains exactly one root
        of \(P_H(\beta+\ii t,\cdot)\). It is the simple protected root
        \[
        w_\star(t) = \e^{-(\beta+\ii t)E_\star}, \qquad E_\star \in \{16,17\}.
        \]
        In particular, there are no spectator roots in the annulus.
        
        \item \emph{Every root is radially separated from both boundary circles.} For either boundary
        energy \(\mu\in\{15,18\}\), write \(R_\mu=\e^{-\beta\mu}\). By
        \cref{prop:ff-certificates}(ii), every root \(\zeta\) of \(P_H(\beta+\ii t,\cdot)\)
        satisfies
        \[
        \min\left\{ \frac{|\zeta|}{R_\mu}, \frac{R_\mu}{|\zeta|} \right\}
        \leq \e^{-\beta} =\frac{1}{3M}.
        \]
        Thus \eqref{eq:spices-radial-separation} holds with \(q = 1/(3M)\), uniformly for all
        real \(t\) and on both circles.
        
        \item \emph{The normalized endpoint signals are bounded away from zero.} At a quadrature node
        \(w_j=R_\mu\e^{\ii\theta_j}\), the corresponding physical fugacity is
        \[
        -1/w_j = \e^{\beta\mu + \ii(\pi-\theta_j)}.
        \]
        Applying \cref{prop:ff-certificates}(iii) with
        \(\vartheta=\pi-\theta_j\) gives
        \[
        \left| \Xi(\beta+\ii t,-1/w_j) \right| =
        \left| \frac{ Z(\beta+\ii t, \e^{\beta\mu+\ii(\pi-\theta_j)}) }{ Z(\beta,\e^{\beta\mu}) }
        \right| \geq \frac{9}{40}.
        \]
        Hence this item holds with \(b=9/40\).
        
        \item \emph{The required operational primitives are
        available with polynomial overhead.} For each \(\mu\in\{15,18\}\), preparation of the
        grand-canonical state
        \[
        \rho_{\beta,\mu} = \frac{\e^{-\beta(H-\mu N)}} {Z(\beta,\e^{\beta\mu})}
        \]
        to operational error at most \(2^{-r}\) is supplied by
        \cref{lem:ff-gibbs-preparation}. Its temperature-dependent cost
        is polynomial here because \(\e^{33\beta}=(3M)^{33}\).
        Thus the preparation hypothesis of \cref{thm:spices-recovery}
        holds at the parameters used by this family. Moreover,
        \cref{cor:ff-number-primitives} supplies the controlled joint evolutions
        \[
        \e^{-\ii tH+\ii\vartheta N}
        \]
        to error at most \(2^{-r}\), with polynomial overhead for the query parameters used by SPICES.
        Choosing \(r\) sufficiently large therefore provides both operations to any required
        inverse-polynomial accuracy.
        
        \item \emph{The conditioning parameters are polynomially bounded.} Here \(L=3\), \(b^{-1} =
        40/9\), and \((1-q)^{-1} \leq 3/2\); the requested \(\epsilon^{-1}\) is
        polynomial by assumption. Since \(L = 3\) and \(\beta = \log(3M)\),
        \[
        \e^{2\beta L} = \e^{6\beta} = (3M)^6,
        \]
        which is polynomial in \(M\).
    \end{enumerate}
    
    All five hypotheses of \cref{thm:spices-recovery} therefore
    hold.  To identify the measure recovered by the theorem, its first
    hypothesis gives the protected first root moment
\[
    \mathcal{M}_{A_\Omega}^{(1)}(s)=\e^{-sE_\star}.
\]
    After the deterministic shift by the lower window endpoint \(\Omega_{\min}=15\), the shifted moment is
\[
    \mathcal{M}_\Omega^{(1)}(s) = \e^{-s(E_\star-15)}.
\]
    Normalizing its thermal weight yields
\[
    \nu_{\mathcal{U}} = \frac{ \e^{-\beta(E_\star - 15)}\delta_{E_\star} }{ \e^{-\beta(E_\star - 15)} } =
    \delta_{E_\star}.
\]
    Thus \cref{thm:spices-recovery} gives the claimed Wasserstein guarantee with polynomial resources.
    Finally,
\[
    \left| \int E\,\widehat\nu(\dd E)-E_\star \right| \leq W_1(\widehat \nu, \delta_{E_\star}) \leq
    \epsilon.
\]
    The reconstruction returns a finitely represented measure, so its first moment can be evaluated to
    inverse-polynomial accuracy in polynomial time.  This gives the claimed estimate of \(E_\star\).
\end{proof}

\appendix
\section{Proofs of supporting statements}\label{app:proofs}

\subsection{Proofs for \texorpdfstring{\cref{sec:prelim}}{Section~\ref*{sec:prelim}}}
\ProofOf{prop:free-factorization}
\begin{proof}
    Because the \(n_j\)'s are commuting projections, the Fock space has a joint occupation basis indexed by \(m=(m_1,\ldots,m_M)\in\{0,1\}^M\), with
\[
    n_j \lvert m\rangle = m_j \lvert m\rangle .
\]
    In this basis,
\[
    N\lvert m\rangle = \left(\sum_{j=1}^{M}m_j\right)\lvert m\rangle, \qquad H_0\lvert m\rangle = \left(\sum_{j=1}^{M}E_jm_j\right)\lvert m\rangle .
\]
    Therefore the analytically continued partition function is the finite sum
    \begin{align}
        Z(s,z) &= \Tr\!\left\{z^N\e^{-sH_0}\right\} \notag\\
        &= \sum_{m\in\{0,1\}^M}
        z^{\sum_j m_j}
        \exp\left(-s\sum_{j=1}^{M}E_jm_j\right) \notag\\
        &= \sum_{m\in\{0,1\}^M}
        \prod_{j=1}^{M}
        \left(z\e^{-sE_j}\right)^{m_j}. \label{eq:occupation-sum}
    \end{align}
    The last sum factorizes since the occupation variables are independent binary indices:
\[
    \sum_{m\in\{0,1\}^M}
    \prod_{j=1}^{M}
    \left(z\e^{-sE_j}\right)^{m_j}
    = \prod_{j=1}^{M}
    \sum_{m_j=0}^{1}
    \left(z\e^{-sE_j}\right)^{m_j}
    =
    \prod_{j=1}^{M}\left(1+z\e^{-sE_j}\right).
\]
    Since \(\e^{-sE_j}\neq0\), this is a polynomial of degree \(M\) in \(z\).
    Using \(\zeta_j=\e^{-sE_j}\) gives
\[
    1+z\e^{-sE_j}
    =
    1+z\zeta_j,
\]
    which proves the factorization in the physical fugacity.
    Substituting \(z=-1/w\) gives
\[
    w^M Z\!\left(s,-\frac1w\right)
    =
    \prod_{j=1}^{M}w\left(1-\frac{\zeta_j}{w}\right)
    =
    \prod_{j=1}^{M}(w-\zeta_j).
\]
    Repeated values of \(\zeta_j\) give the corresponding root multiplicity.
\end{proof}

\ProofOf{prop:power-sum-formula}
\begin{proof}
    Since \(P\) has no zeros on \(\partial A\), the logarithmic derivative \(P'(w)/P(w)\) is meromorphic in a neighborhood of \(A\) and has poles only at zeros of \(P\).
    If \(\alpha\) is a distinct zero with multiplicity \(m_\alpha\), then locally
\[
    \frac{P'(w)}{P(w)}
    =
    \frac{m_\alpha}{w-\alpha}
    +
    \text{holomorphic terms}.
\]
    By the residue theorem and \cref{eq:alt-root-moment},
\[
    \mathcal{M}_A^{(k)}
    =
    \sum_{\alpha\in A}
    m_\alpha\alpha^k.
\]
    Counting the zeros with multiplicity gives \cref{eq:power-sum}.
\end{proof} 

\ProofOf{lem:single-root-aliasing}
\begin{proof}
    Let \(h=2\pi/S\).
    After dividing the numerator and denominator by \(-\zeta\), the single-root contribution is
    \begin{equation}
        q_{R,S}^{(k)}(\zeta)
        =
        \frac{\sinc(k/S)^{-1}}{2\pi\ii}
        \sum_{j=0}^{S-1}
        w_{j+1/2}^{\,k}
        \Log
        \left(
        \frac{1-w_{j+1}/\zeta}{1-w_j/\zeta}
        \right),
        \label{eq:single-root-raw-log-rule}
    \end{equation}
    where the logarithmic increments are taken by continuous continuation along
    the circle.
    We compute \(q_{R,S}^{(k)}\) using the Taylor expansion for the logarithm.
    
    Suppose first that \(|\alpha|<1\).
    On the circle \(w=R\e^{\ii\theta}\),
\[
    1-\frac{w}{\zeta}
    =
    -\frac{w}{\zeta}\left(1-\frac{\zeta}{w}\right)
    =
    -\frac{R\e^{\ii\theta}}{\zeta}
    \left(1-\alpha\e^{-\ii\theta}\right).
\]
    Thus the continuous logarithmic increment over the \(j\)th arc is
\[
    \ii h+L(\theta_{j+1})-L(\theta_j),
    \qquad
    L(\theta)=\log(1-\alpha\e^{-\ii\theta}).
\]
    The term \(\ii h\) contributes zero to
    \cref{eq:single-root-raw-log-rule}, since \(k>0\) and
\[
    \sum_{j=0}^{S-1}\e^{\ii k\theta_{j+1/2}}=0.
\]
    For \(|\alpha|<1\), the Taylor expansion is uniformly convergent and gives
\[
    L(\theta)
    =
    -\sum_{n=1}^{\infty}
    \frac{\alpha^n}{n}\e^{-\ii n\theta}.
\]
    Therefore
\[
    L(\theta_{j+1})-L(\theta_j)
    =
    \sum_{n=1}^{\infty}
    \frac{2\ii}{n}
    \alpha^n
    \sin\left(\frac{nh}{2}\right)
    \e^{-\ii n\theta_{j+1/2}}.
\]
    Substituting into \cref{eq:single-root-raw-log-rule} gives
\[
    q_{R,S}^{(k)}(\zeta)
    =
    \frac{\sinc(k/S)^{-1} R^k}{\pi}
    \sum_{n=1}^{\infty}
    \frac{\alpha^n}{n}
    \sin\left(\frac{nh}{2}\right)
    \sum_{j=0}^{S-1}
    \e^{\ii(k-n)\theta_{j+1/2}}.
\]
    The last sum vanishes unless \(n=k+\ell S\) for some \(\ell\ge0\).
    For these indices,
\[
    \sum_{j=0}^{S-1}
    \e^{\ii(k-n)\theta_{j+1/2}}
    =
    (-1)^\ell S,
    \qquad
    \sin\left(\frac{(k+\ell S)h}{2}\right)
    =
    (-1)^\ell\sin\left(\frac{\pi k}{S}\right).
\]
    The signs cancel, and the sinc correction cancels up to a factor of \(k\). 
    Isolating the \(\ell = 0\) term, we obtain
    \begin{equation}
        q_{R,S}^{(k)}(\zeta)
        =
        \zeta^k
        \left[
        1+
        \sum_{\ell=1}^{\infty}
        \frac{k}{k+\ell S}
        \alpha^{\ell S}
        \right].
        \label{eq:inside-aliasing}
    \end{equation}
    
    Now suppose that \(|\alpha|>1\), and set \(\eta=\alpha^{-1}=R/\zeta\).
    Then \(|\eta|<1\), and on the circle
\[
    1-\frac{w}{\zeta}
    =
    1-\eta\e^{\ii\theta}.
\]
    The continuous logarithm may be chosen as
\[
    L(\theta)
    =
    \log(1-\eta\e^{\ii\theta})
    =
    -\sum_{n=1}^{\infty}
    \frac{\eta^n}{n}\e^{\ii n\theta}.
\]
    Hence
\[
    L(\theta_{j+1})-L(\theta_j)
    =
    -\sum_{n=1}^{\infty}
    \frac{2\ii}{n}
    \eta^n
    \sin\left(\frac{nh}{2}\right)
    \e^{\ii n\theta_{j+1/2}}.
\]
    Substitution into \cref{eq:single-root-raw-log-rule} gives
\[
    q_{R,S}^{(k)}(\zeta)
    =
    -\frac{\sinc(k/S)^{-1}R^k}{\pi}
    \sum_{n=1}^{\infty}
    \frac{\eta^n}{n}
    \sin\left(\frac{nh}{2}\right)
    \sum_{j=0}^{S-1}
    \e^{\ii(k+n)\theta_{j+1/2}}.
\]
    The last sum vanishes unless \(k+n=\ell S\), that is
\(n=\ell S-k\) with \(\ell\ge1\).
    For those indices,
\[
    \sum_{j=0}^{S-1}
    \e^{\ii\ell S\theta_{j+1/2}}
    =
    (-1)^\ell S,
    \qquad
    \sin\left(\frac{(\ell S-k)h}{2}\right)
    =
    (-1)^{\ell+1}\sin\left(\frac{\pi k}{S}\right).
\]
    The sign from the Taylor expansion cancels the remaining signs, and again the sinc correction cancels up to a factor of \(k\). 
    Putting it together, we obtain
    \begin{equation}
        q_{R,S}^{(k)}(\zeta)
        =
        \zeta^k
        \sum_{\ell=1}^{\infty}
        \frac{k}{\ell S-k}
        \eta^{\ell S}.
        \label{eq:outside-aliasing}
    \end{equation}
    
    The exact single-root moment on the disc is
\(\zeta^k\mathbf{1}_{\{|\alpha|<1\}}\).
    If \(|\alpha|<1\), then \(a=|\alpha|\), and
    \cref{eq:inside-aliasing} gives
\[
    \left|
    \zeta^k-q_{R,S}^{(k)}(\zeta)
    \right|
    \le
    |\zeta|^k
    \sum_{\ell=1}^{\infty}
    \frac{k}{k+\ell S}a^{\ell S}
    \le
    \frac{k}{S}|\zeta|^k
    \frac{a^S}{1-a^S}.
\]
    If \(|\alpha|>1\), then \(a=|\eta|=|\alpha|^{-1}\), and the exact
    single-root moment is zero.
    Since \(S\ge2k\),
\(\ell S-k\ge S/2\) for every \(\ell\ge1\), so
    \cref{eq:outside-aliasing} gives
\[
    \left|
    q_{R,S}^{(k)}(\zeta)
    \right|
    \le
    |\zeta|^k
    \sum_{\ell=1}^{\infty}
    \frac{k}{\ell S-k}a^{\ell S}
    \le
    \frac{2k}{S}|\zeta|^k
    \frac{a^S}{1-a^S}.
\]
    This proves the stated bound.
\end{proof}

\subsection{Proofs for \texorpdfstring{\cref{sec:noise}}{Section~\ref*{sec:noise}}}
\ProofOf{lem:one-factor-phase-bound}
\begin{proof}
    The path is nonzero since
\[
    |u(\theta)|
    =
    |\e^{\ii\theta}-r\e^{\ii\phi}|
    \ge
    |1-r|.
\]
    The endpoint ratio \(u(h)/u(0)\) is the factor displayed in \cref{eq:ratio_factor}, so it has argument \(\Delta\arg u\) modulo \(2\pi\).
    If \(0\le r<1\), then
\[
    \e^{\ii\theta}-r\e^{\ii\phi}
    =
    \e^{\ii\theta} \cdot
    \left(1-r\e^{\ii(\phi-\theta)}\right).
\]
    Under the logarithmic derivative, the first factor contributes exactly \(h\) to the lifted argument.
    The second factor has logarithmic derivative upper bounded in modulus by
\[
    \frac{r}{1-r}.
\]
    Thus, its integral is at most \(hr/(1-r)\), and the residual argument increment magnitude (the left side of \cref{eq:centered-single-factor-bound}) is also most \(hr/(1-r)\).
    If \(r>1\), then
\[
    \e^{\ii\theta}-r\e^{\ii\phi}
    =
    -\e^{\ii\phi}
    \cdot r \left(1-r^{-1}\e^{\ii(\theta-\phi)}\right).
\]
    Here, the first factor contributes no argument increment, as it is a constant with respect to \(\theta\), and the remaining factor has logarithmic derivative upper bounded in modulus by
\[
    \frac{r^{-1}}{1-r^{-1}}.
\]
    This proves \cref{eq:centered-single-factor-bound}.
\end{proof}

\ProofOf{prop:centered-product-phase-bound}
\begin{proof}
    For each factor labeled by \(j\), let \(u_j(\theta)=\e^{\ii\theta}-r_j\e^{\ii\phi_j}\) for \(0\le\theta\le h\), and write
\[
    \Delta\arg u_j=h\mathbf 1_{\{r_j<1\}}+\eta_j,
    \qquad
    |\eta_j|\le h\,\frac{\rho_j}{1-\rho_j},
\]
    using \cref{lem:one-factor-phase-bound}.
    The product \(u=\prod_ju_j\) has lifted argument increment
\[
    \Delta\arg u=\sum_{j=1}^{m}\Delta\arg u_j
    =
    Kh+\sum_{j=1}^{m}\eta_j,
\]
    with
\[
    |\Delta\arg u-Kh|\le hB.
\]
    The ratio in \cref{eq:centered-product-phase-bound} has argument \(\Delta\arg u\) modulo \(2\pi\).
    For any real \(x\),
\[
    \pi-|\operatorname{Arg}\e^{\ii x}|
    =
    \operatorname{dist}\!\left(x,(2\mathbb Z+1)\pi\right).
\]
    The distance to a fixed closed set is \(1\)-Lipschitz, so \cref{eq:centered-product-phase-bound} follows.
\end{proof}

\ProofOf{prop:phase-stability-sampling}
\begin{proof}
    Here \(\zeta_\ell=\e^{-sE_\ell}\) are the roots of \(P\) (\cref{prop:free-factorization}).
    For a fixed \(j\), write
\[
    r_\ell=\frac{|\zeta_\ell|}{R}
    =
    \e^{-\beta(E_\ell-\mu)},
    \qquad
    \phi_{\ell,j}=\arg(\zeta_\ell)-\theta_j .
\]
    Then \(r_\ell<1\) exactly when \(E_\ell>\mu\), so the integer \(K\) in
    \cref{prop:centered-product-phase-bound} is the \(K\) defined
    above.
    Moreover,
\[
    \min\{r_\ell,r_\ell^{-1}\}
    =
    \e^{-\beta|E_\ell-\mu|}.
\]
    Applying \cref{prop:centered-product-phase-bound} to
\(\Gamma_j(R,s)\), and then minimizing over \(j\), gives
\[
    \chi_S(R,s)
    \ge
    \operatorname{dist}\!\left(
    \frac{2\pi K}{S},(2\mathbb Z+1)\pi
    \right)
    -
    \frac{2\pi B}{S}.
\]
    
    If \(S\ge4(K+B)\), then
\[
    \frac{2\pi K}{S}\le\frac{\pi}{2}.
\]
    Thus
\[
    \operatorname{dist}\!\left(
    \frac{2\pi K}{S},(2\mathbb Z+1)\pi
    \right)
    =
    \pi-\frac{2\pi K}{S}.
\]
    Substitution in the preceding margin bound gives
\[
    \chi_S
    \ge
    \pi-\frac{2\pi}{S}\left(K+B\right)
    \ge
    \frac{\pi}{2}.
\]
\end{proof}

\ProofOf{lem:endpoint-noise-branch-stability}
\begin{proof}
    The known monomial factor relating \(\Xi\) to \(P\) does not affect the
    noise perturbation.
    Thus the noisy endpoint ratio differs from the exact endpoint ratio by
\[
    \frac{\widehat \Xi_{j + 1}}{\widehat \Xi_j} = \frac{1+\varsigma_{j+1}}{1+\varsigma_j} \cdot \frac{\Xi_{j+ 1}}{\Xi_j}.
\]
    Notice that \(0<\chi_S\le\pi\).
    Consider the disc of radius \(\varsigma\) centered at 1, \(\{u:|u-1|\le\varsigma\}\). 
    The absolute argument \(\lvert \operatorname{Arg}u \rvert\) of any point in this disc is at most the angle of the tangent ray from the origin to this disc that lies in the first quadrant. 
    For \(0<\eta \le \pi/2\), if \(\varsigma < \sin\eta\), then \(|\operatorname{Arg}u| < \eta\).
    From our assumption that \(|\varsigma_j| \le \varsigma < \sin\bigl(\tfrac{\chi_S}{2}\bigr) \), we have
\[
    \lvert\operatorname{Arg}(1+\varsigma_j)\rvert<\frac{\chi_S}{2}
    \qquad
    \text{for every } j.
\]
    It follows that the phase of
\((1+\varsigma_{j+1})/(1+\varsigma_j)\) has magnitude less than \(\chi_S\).
    Since every exact endpoint ratio is at distance at least \(\chi_S\) from the
    branch cut at argument \(\pm\pi\), the noisy endpoint ratio cannot cross that cut.
    Thus the principal logarithm selects the same branch correction as in the exact logarithmic increment.
\end{proof}

\ProofOf{lem:shot-noise-quadrature-perturbation}
\begin{proof}
    The phase margin condition and \cref{lem:endpoint-noise-branch-stability}
    show that the noisy endpoint ratios preserve the exact branch assignment.
    Let
\[
    b_j=\Log(1+\varsigma_j),
    \qquad
    b_S=b_0.
\]
    Then the noisy logarithmic increment differs from the exact one by
\(b_{j+1}-b_j\).
    Since \(|\varsigma_j|\le\varsigma<1/\sqrt{2}\),
\[
    |b_j|
    \le
    -\log(1-\varsigma).
\]
    Hence
    \begin{align*}
        \widehat Q_{R,S}^{(k)}(s)-Q_{R,S}^{(k)}(s)
        &=
        \frac{1}{\sinc(k/S)}
        \frac{1}{2\pi\ii}
        \sum_{j=0}^{S-1}
        w_{j+1/2}^{\,k}(b_{j+1}-b_j).
    \end{align*}
    By relabeling the summation over cyclic indices, we find
\[
    \sum_{j=0}^{S-1}
    w_{j+1/2}^{\,k}(b_{j+1}-b_j)
    =
    \sum_{j=0}^{S-1}
    \left(w_{j-1/2}^{\,k}-w_{j+1/2}^{\,k}\right)b_j .
\]
    Since
\[
    \left|w_{j-1/2}^{\,k}-w_{j+1/2}^{\,k}\right|
    =
    2R^k\sin\frac{\pi k}{S},
\]
    the sinc factor cancels at the cost of a factor of \(k\) and yields
\[
    \left|
    \widehat Q_{R,S}^{(k)}(s)-Q_{R,S}^{(k)}(s)
    \right|
    \le
    kR^k[-\log(1-\varsigma)].
\]
    Finally, the loose bound \(-\log(1-\varsigma)<2\varsigma\) holds for
\(0\le\varsigma\le1/\sqrt2\).
\end{proof}

\ProofOf{lem:hadamard-test}
\begin{proof}
    Let \(\overline X_{n_{\mathrm{sh}}}\) and \(\overline Y_{n_{\mathrm{sh}}}\) be the empirical means of the outcomes of the \(n_{\mathrm{sh}}\) independent \(X\) measurements and \(n_{\mathrm{sh}}\) independent \(Y\) measurements. 
    Define the empirical estimator
\[
    \widehat \xi = \overline X_{n_{\mathrm{sh}}}+\ii \overline Y_{n_{\mathrm{sh}}} .
\]
    Each measurement outcome is an independent random variable in \([-1,1]\) with mean \(\Re(\xi)\) for \(X\) measurements and \(\Im(\xi)\) for \(Y\) measurements. 
    By Hoeffding's inequality,
\[
    \mathbb P\left\{
    |\overline X_{n_{\mathrm{sh}}}-\Re(\xi)|>\frac{\varsigma |\xi|}{\sqrt2}
    \right\}
    \le
    2\exp\left(-\frac{n_{\mathrm{sh}}\varsigma^2|\xi|^2}{4}\right),
\]
    and the same bound holds for
\(|\overline Y_{n_{\mathrm{sh}}}-\Im(\xi)|\).
    The choice \cref{eq:hadamard-relative-shots} makes the union bound over the two quadratures at most \(\delta\).
    On the complementary event,
\[
    |\widehat \xi-\xi|
    \le
    \left(
    \frac{\varsigma^2|\xi|^2}{2}
    +
    \frac{\varsigma^2|\xi|^2}{2}
    \right)^{1/2}
    =
    \varsigma |\xi|
\]
    and the result follows.
\end{proof} 

\ProofOf{lem:normalized-fourier-data}
\begin{proof}
    Let \(\eta = \frac{\kappa}{32}\e^{-\beta kL}\).
    We have \(\bigl|\Mell\bigr|\le \mnorm\).
    Moreover, \cref{eq:shifted-mass-lower-bound} and the assumption on \(\kappa\) give \(\eta\le \mnorm/2\), so \(|\mnormhat-\mnorm|\le\eta\) implies \(|\mnormhat|\ge \mnorm/2\).
    For \(0\le\ell\le T\),
\[
    \left|
    \frac{\widehat{\mathcal{M}}_\ell}{\mnormhat}
    -
    \frac{\Mell}{\mnorm}
    \right|
    \le
    \frac{\eta}{|\mnormhat|}
    +
    \bigl|\Mell\bigr|
    \frac{|\mnormhat-\mnorm|}{\mnorm|\mnormhat|}
    \le
    \frac{4\eta}{\mnorm}
    \le
    \frac{\kappa}{8}.
\]
    This proves the claim for the negative-frequency coefficients, and the positive-frequency claim follows by complex conjugation.
\end{proof}

\ProofOf{prop:wasserstein-window-recovery}
\begin{proof}
    We use the properties of the Jackson kernel and LP estimator discussed in \cref{def:jackson-kernel} in \cref{app:lp-estimator}. 
    This proof is a direct translation of the proof by Chen et al.~\cite{ChenEtAl2025} adapted to our setting.
    
    Write \(\rho=\phi_\#\nu\), and let \(\bar\rho=\rho*J_n\) be its Jackson
    smoothing.
    We compare the LP output first with \(\bar\rho\), and then account for the
    smoothing error between \(\bar\rho\) and \(\rho\).
    \cref{def:jackson-kernel} gives
\[
    W_1^{\mathbb T}(\rho,\bar\rho)
    \le
    \frac{1}{n}
    \le
    \frac{\epsilon}{4L}.
\]
    
    The \(\ell\)th Fourier coefficient of \(\bar\rho\) is
\(\widetilde{J}_{n}(\ell)u_\ell\).
    The LP constraints bound the real and imaginary parts of the difference
    between the corresponding coefficient of \(\sigma\) and
\(\widetilde{J}_{n}(\ell)\widehat u_\ell\) by \(\kappa/4\).
    Combining the real and imaginary bounds with
    \cref{lem:normalized-fourier-data} and
\(|\widetilde{J}_{n}(\ell)|\le1\), we obtain, for \(|\ell|\le T\),
\[
    \left|
    \int \e^{2\pi\ii\ell x}\dd\sigma(x)
    -
    \int \e^{2\pi\ii\ell x}\dd\bar\rho(x)
    \right|
    \le
    \frac{\sqrt2\kappa}{4}+\frac{\kappa}{8}
    <
    \kappa.
\]
    The cutoff \(T=\lceil 24L/\epsilon\rceil\) and the tolerance \(\kappa\) (with \(c\le1/4000\)) are exactly those for which
    \cref{lem:chen-stability} applies with target torus accuracy \(\varepsilon'=\epsilon/(4L)\), as computed in \cref{app:lp-estimator}.
    Applying it to the pair \((\sigma,\bar\rho)\) and using the Fourier bound above for all \(|\ell|\le T\),
    we obtain
\[
    W_1^{\mathbb T}(\bar\rho,\sigma)
    \le
    \frac{\epsilon}{4L}.
\]
    The triangle inequality now gives
\[
    W_1^{\mathbb T}(\rho,\sigma)
    \le
    \frac{\epsilon}{2L}.
\]
    Finally, \(\psi\circ\phi\) is the identity on \(\Omega\), and \(\psi\) is
\(2L\)-Lipschitz with respect to the torus metric.
    The contraction property of Wasserstein distance under a Lipschitz map
    therefore yields
\[
    W_1\!\left(\nu,\widehat\nu\right)
    =
    W_1\!\left(\psi_\#\rho,\psi_\#\sigma\right)
    \le
    2L\,W_1^{\mathbb T}(\rho,\sigma)
    \le
    \epsilon.
\]
\end{proof}

\subsection{Proofs for \texorpdfstring{\cref{sec:free-gapped}}{Section~\ref*{sec:free-gapped}}}
\ProofOf{lem:logarithmic-tanh-bound}
\begin{proof}
    Set \(y=\e^{-2x}\in(0,1/2]\).
    Then
\[
    -\log\tanh x
    =
    \log\frac{1+y}{1-y}
    =
    \log(1+y)-\log(1-y).
\]
    Using \(\log(1+y)\le y\) and
\(-\log(1-y)\le y/(1-y)\le 2y\) gives
\[
    -\log\tanh x
    \le
    3y,
\]
    which proves the claim.
\end{proof}

\ProofOf{prop:gap-lower-bound}
\begin{proof}
    By \cref{prop:free-factorization} and \cref{eq:normalized-signal},
\[
    \Xi(s,z)
    =
    \prod_{j=1}^{M}
    \frac{
    1+\e^{-\beta(E_j-\mu)}\e^{\ii(\vartheta-tE_j)}
    }{
    1+\e^{-\beta(E_j-\mu)}
    }.
\]
    For \(a>0\) and any phase \(\phi\),
\[
    |1+a\e^{\ii\phi}|
    \ge
    |1-a|.
\]
    Taking \(a=\e^{-\beta(E_j-\mu)}\) gives
\[
    \frac{|1+a\e^{\ii\phi}|}{1+a}
    \ge
    \frac{|1-a|}{1+a}
    =
    \tanh\left(\frac{\beta |E_j-\mu|}{2}\right).
\]
    Multiplying over \(j\) proves the product bound.
    Since \(\Delta(\mu)\le |E_j-\mu|\) and \(\tanh\) is increasing on
\((0,\infty)\), the second inequality in
    \cref{eq:Xi0-gap-product-bound} follows.
    Taking logarithms and applying \cref{lem:logarithmic-tanh-bound} with
\(x=\beta\Delta(\mu)/2\) gives \cref{eq:Xi0-gap-exponential-bound}.
\end{proof}

\ProofOf{cor:log-inv-temp-suffices}
\begin{proof}
    The condition implies
\[
    3M\e^{-\beta\Delta(\mu)}\le 1.
\]
    The claim follows from \cref{eq:Xi0-gap-exponential-bound}.
\end{proof}

\ProofOf{prop:quadrature-gap-bound}
\begin{proof}
    By \cref{prop:free-factorization}, \(P(w)=\prod_{j=1}^{M}(w-\zeta_j)\), so the logarithmic derivative of \(P\) is the sum of the logarithmic
    derivatives of the factors \(1-w/\zeta_j\) (which differ from \(w-\zeta_j\) by a constant), and the quadrature error is the
    sum of the single-root errors.
    Applying \cref{lem:single-root-aliasing} to each root gives the sharper
    intermediate estimate
\[
    \left|
    \mathcal{M}_R^{(k)}(s)-Q_{R,S}^{(k)}(s)
    \right|
    \le
    \frac{2k}{S}
    \sum_{j=1}^{M}
    \e^{-\beta k E_j}
    \frac{\e^{-S\beta |E_j-\mu|}}
    {1-\e^{-S\beta |E_j-\mu|}},
\]
    since
\[
    |\zeta_j|^k=\e^{-\beta k E_j},
    \qquad
    a_j=\e^{-\beta |E_j-\mu|}.
\]
    The coarser estimate \cref{eq:disc-gap-bound} follows by splitting the sum
    according to whether \(E_j\) is above or below \(\mu\).
    If \(E_j>\mu\), then
\[
    \e^{-\beta kE_j}
    \frac{\e^{-S\beta(E_j-\mu)}}{1-\e^{-S\beta(E_j-\mu)}}
    \le
    R^k\frac{q^{S+k}}{1-q^S}.
\]
    If \(E_j<\mu\), then \(S\ge2k\) gives
\[
    \e^{-\beta kE_j}
    \frac{\e^{-S\beta(\mu-E_j)}}{1-\e^{-S\beta(\mu-E_j)}}
    \le
    R^k\frac{q^{S-k}}{1-q^S}.
\]
    Both terms are bounded by
\(R^k q^{S-k}/(1-q^S)\), and there are \(M\) roots.
    
    The quadrature for an annulus is the difference of two disc quadratures. 
    Applying \cref{eq:disc-gap-bound} at the two boundary chemical potentials
\(\mu_-\) and \(\mu_+\), using the triangle inequality, and then using
\(\e^{-\beta\Delta(\mu_\pm)}\le q_A\), gives
    \cref{eq:annulus-gap-bound}.
\end{proof} 

\ProofOf{cor:phase-margin-guarantee}
\begin{proof}
    By \cref{prop:phase-stability-sampling} it suffices that \(S\ge4(K+B)\).
    If \(\beta\Delta(\mu)\ge\log(3M)\), then \(q=\e^{-\beta\Delta(\mu)}\le 1/(3M)\), so \(B\le Mq/(1-q)\le M/(3M-1)\le\tfrac12\), and \(S\ge4K+2\) suffices.
    Since \(K\le M\), \(S\ge4M+2\) suffices, and this also satisfies \(S>2\).
\end{proof}

\ProofOf{prop:fixed-root-moment-estimate}
\begin{proof}
    The first statement is the triangle inequality, using the deterministic aliasing bound
    \cref{eq:disc-gap-bound} and
    \cref{lem:shot-noise-quadrature-perturbation}. 
    The second statement uses \(q \le 1/3M\) to bound the quadrature aliasing error. 
\end{proof}

\subsection{Proofs for \texorpdfstring{\cref{sec:perturbative-stability}}{Section~\ref*{sec:perturbative-stability}}}

\ProofOf{prop:spectator-factorization}
\begin{proof}
    The Hamiltonian \(H_\lambda=H_I\otimes\I+\I\otimes(H_O+\lambda V_O)\) and \(N=N_I\otimes\I+\I\otimes N_O\) have the form of \cref{lem:spse-factorization} with \(U=\I\) and spectator Hamiltonian \(H_O+\lambda V_O\), which is supported on \(O\) and commutes with \(N_O\) because \([H_O,N_O]=[V_O,N_O]=0\).
    The lemma therefore gives \eqref{eq:spectator-Z-factorization}.
\end{proof}

\ProofOf{prop:exact-window-protection}
\begin{proof}
    On the boundary of the annulus,
    \eqref{eq:spectator-P-factorization} gives
\[
    \dd\log P_\lambda
    =
    \dd\log P_I
    +
    \dd\log P_{O,\lambda}.
\]
    Integrating this identity shows that root moments are additive under the
    product:
\[
    \mathcal{M}_{A_\Omega}^{(k)}[P_\lambda]
    =
    \mathcal{M}_{A_\Omega}^{(k)}[P_I]
    +
    \mathcal{M}_{A_\Omega}^{(k)}[P_{O,\lambda}].
\]
    The outside factor has no zeros in the annulus, so the power-sum formula
    gives
\[
    \mathcal{M}_{A_\Omega}^{(k)}[P_{O,\lambda}]
    =
    0.
\]
    The claim follows from the above additivity and from \cref{prop:power-sum-formula}, the power sum formula for the free root moments.
\end{proof}

\ProofOf{thm:finite-time-spectator-stability}
\begin{proof}
    The proof is a direct application of Rouch\'{e}'s theorem. 
    Fix \(t\) with \(|t|\le t_{\max}\), and on either boundary circle \(C\), set
\[
    f(w)=P_{O,0}(s,w),
    \qquad
    g(w)=P_{O,\lambda}(s,w).
\]
    Condition~\eqref{eq:symmetric-rouche-condition} implies that neither \(f\) nor \(g\) vanishes on \(C\).
    Consider the straight-line homotopy
\[
    h_u(w)=(1-u)f(w)+ug(w),
    \qquad
    0\le u\le1.
\]
    Suppose \(h_u(w) = 0\) for some \(w \in C\) and \(0 < u < 1\). Then
\(g(w) = -\tfrac{1-u}{u}f(w)\), so \(f(w)\) and \(g(w)\) are antiparallel
    (since \(\tfrac{1-u}{u} > 0\) for \(u \in (0, 1)\)).
    Hence
\[
    |g(w)-f(w)|
    =
    |g(w)|+|f(w)|,
\]
    contrary to \eqref{eq:symmetric-rouche-condition}.
    Thus the homotopy is nonzero on \(C\).
    Parameterize \(C\) by \(\gamma\colon[a,b]\to\mathbb{C}\) with \(\gamma(a)=\gamma(b)\).
    By Cauchy's argument principle applied to the polynomials \(f\) and \(g\) (which have no poles), the difference between their zero counts inside \(C\) equals
\[
    \frac{1}{2\pi i}\oint_C \!\left\{\frac{f'(w)}{f(w)} - \frac{g'(w)}{g(w)}\right\}dw
    = \frac{1}{2\pi i}\int_a^b \frac{d}{dt}
    \log\frac{f[\gamma(t)]}{g[\gamma(t)]}\,dt.
\]
    Since \(h_u\), \(f\), and \(g\) are all nonzero on \(C\), the ratio \(f(w)/g(w)\) is nonzero and avoids the negative real axis on \(C\). 
    The logarithm on the right-hand side is therefore single-valued, and the integral vanishes. 
    This gives that \(f\) and \(g\) have the same number of zeros inside \(C\).
    
    Apply this conclusion first at \(R_{\mathrm{in}}\) and then at \(R_{\mathrm{out}}\).
    The unperturbed outside polynomial has no zeros between the two circles, so its two zero counts are equal.
    The two perturbed zero counts are therefore also equal.
    Thus \(P_{O,\lambda}(s,\cdot)\) has no zeros in \(A_\Omega\), and, as noted above, none on the two circles; so it has no zeros in the closed annulus \(\overline{A_\Omega}\).
    
    Proposition~\ref{prop:exact-window-protection} now gives \eqref{eq:finite-time-protected-window-moment}.
    Since the argument applies to every \(|t|\le t_{\max}\), the conclusion is uniform over the time interval.
\end{proof}

\ProofOf{cor:relative-spectator-stability}
\begin{proof}
    The definition of \(\eta_\lambda( t_{\max})\) gives
\[
    |P_{O,\lambda}-P_{O,0}|
    <
    |P_{O,0}|
    \le
    |P_{O,\lambda}|+|P_{O,0}|
\]
    on both boundary circles throughout the time interval.
    Thus \eqref{eq:symmetric-rouche-condition} holds, and the theorem applies.
\end{proof}

\ProofOf{prop:outside-duhamel-bound}
\begin{proof}
    Fix \(|t|\le t_{\max}\) and set \(s=\beta+\ii t\).
    On either contour circle, write 
\[
    K_\lambda=H_O+\lambda V_O.
\]
    Duhamel's formula gives
    \begin{equation}
        \e^{-sK_\lambda}
        -
        \e^{-sH_O}
        =
        -\lambda s
        \int_0^1
        \e^{-(1-u)sK_\lambda}
        V_O
        \e^{-usH_O}
        \,\dd u.
        \label{eq:outside-duhamel-formula}
    \end{equation}
    We first bound the trace of the Duhamel integrand without separating the number weight from the thermal factors.
    For \(0<u<1\), define the following operators in terms of the complex fugacity \(z\)
\[
    X_u
    =
    z^{(1-u)N_O}
    \e^{-(1-u)sK_\lambda}, \qquad
    Y_u
    =
    z^{uN_O}
    \e^{-usH_O}.
\]
    Number conservation implies
\[
    X_uV_OY_u
    =
    z^{N_O}
    \e^{-(1-u)sK_\lambda}
    V_O
    \e^{-usH_O}.
\]
    Now introduce the H\"{o}lder pair of exponents
\[
    p=\frac{1}{1-u},
    \qquad
    p'=\frac{1}{u}.
\]
    Since \(N_O\) commutes with both Hamiltonians, in Schatten norms we have
\[
    \|X_u\|_p
    =
    Z_{O,\lambda}(\beta,|z|)^{1-u},
    \qquad
    \|Y_u\|_{p'}
    =
    Z_{O,0}(\beta,|z|)^u.
\]
    Using the matrix H\"{o}lder inequality \cite[Theorem 2]{Baumgartner2011} and \(\|V_O\|=1\) gives
    \begin{equation}
        \left|
        \Tr\!\left(
        z^{N_O}
        \e^{-(1-u)sK_\lambda}
        V_O
        \e^{-usH_O}
        \right)
        \right|
        \le
        Z_{O,\lambda}(\beta,|z|)^{1-u}
        Z_{O,0}(\beta,|z|)^u.
        \label{eq:thermal-duhamel-integrand-bound}
    \end{equation}
    The endpoint cases \(u=0,1\) follow by continuity.
    
    Because \([V_O,N_O]=0\), the two outside Hamiltonians decompose into the same fixed-number sectors.
    Weyl's inequality moves every eigenvalue in each sector by at most
\(|\lambda|\), and hence
    \begin{equation}
        Z_{O,\lambda}(\beta,|z|) \le \e^{\beta|\lambda|} Z_{O,0}(\beta,|z|).
        \label{eq:Weyl-upper-bound}
    \end{equation}
    Substituting this estimate into \eqref{eq:thermal-duhamel-integrand-bound}, integrating \eqref{eq:outside-duhamel-formula}, and using
\[
    \int_0^1
    \e^{(1-u)\beta|\lambda|}
    \,\dd u
    =
    \frac{\e^{\beta|\lambda|}-1}{\beta|\lambda|}
\]
    gives
    \begin{equation}
        |Z_{O,\lambda}(s,z)-Z_{O,0}(s,z)|
        \le
        \frac{|s|}{\beta}
        \left(\e^{\beta|\lambda|}-1\right)
        Z_{O,0}(\beta,|z|).
        \label{eq:thermal-partition-difference}
    \end{equation}
    For \(\lambda=0\), the same inequality holds by continuity.
    
    Define the normalized free outside signal in accordance with, e.g. \eqref{eq:normalized-signal}
\[
    \Xi_{O,0}(s,z)
    =
    \frac{Z_{O,0}(s,z)}{Z_{O,0}(\beta,|z|)}.
\]
    The monomial \(w^{M_O}\) cancels from the relative difference of the two
    outside inverse-fugacity polynomials.
    Dividing \eqref{eq:thermal-partition-difference} by
\(|Z_{O,0}(s,z)|\) therefore gives
    \begin{equation}
        \left|
        \frac{P_{O,\lambda}(s,w)}{P_{O,0}(s,w)}-1
        \right|
        \le
        \frac{|s|}{\beta}
        \frac{\e^{\beta|\lambda|}-1}{|\Xi_{O,0}(s,z)|}.
        \label{eq:thermal-relative-perturbation-bound}
    \end{equation}
    
    It remains to lower-bound the free outside signal. 
    We follow a similar argument to \cref{prop:gap-lower-bound}. 
    On the circle associated with \(\mu\in\{\Omega_{\min},\Omega_{\max}\}\), we have \(|z|=\e^{\beta\mu}\).
    By \cref{prop:free-factorization} (applied to the modes in \(O\)) and the reverse triangle inequality, we have
\[
    |\Xi_{O,0}(s,z)|
    \ge
    \prod_{j\in O}
    \frac{\left|1-\e^{-\beta(E_j-\mu)}\right|}
    {1+\e^{-\beta(E_j-\mu)}}
    =
    \prod_{j\in O}
    \tanh\!\left(
    \frac{\beta|E_j-\mu|}{2}
    \right)
    \ge
    \tanh\!\left(
    \frac{\beta\Delta_{\Omega}}{2}
    \right)^{M_O}.
\]
    Inserting this bound into
    \eqref{eq:thermal-relative-perturbation-bound}, using
\(|s|\le\sqrt{\beta^2+ t_{\max}^2}\), and taking the supremum over both circles
    and all \(|t|\le t_{\max}\) proves \eqref{eq:outside-duhamel-relative-bound}.
\end{proof}

\ProofOf{cor:log-gap-perturbative-stability}
\begin{proof}
    The gap hypothesis \(\beta\Delta_\Omega\ge\log(3M)\) and \(M_O\le M\) give
\[
    \tanh\!\left(\frac{\beta\Delta_\Omega}{2}\right)^{-M_O}
    \le
    \left(\frac{3M+1}{3M-1}\right)^M
    \le
    2.
\]
    For the right hand side of \cref{eq:outside-duhamel-relative-bound} to be less than one, solving for \(|\lambda|\) shows it suffices that 
\[
    |\lambda|<\frac1\beta\log\Bigl(1+\tfrac{\beta}{2 \sqrt{\beta ^2+ t_{\max} ^2}}\Bigr).
\]
    Now multiply both sides by \(\sqrt{\beta^2+ t_{\max}^2}\) and let \(x = \sqrt{1+ t_{\max}^2/\beta^2} \ge 1\). 
    Since \(x \log\bigl(1+1/(2x)\bigr)\) is increasing as a function of \(x\) and the minimum value of \(x\) is 1, a simpler sufficient condition is \(|\lambda|\sqrt{\beta^2+ t_{\max}^2}<\log \frac32\), from which \cref{cor:relative-spectator-stability} follows. 
\end{proof}

\ProofOf{prop:conditioning-transfer}
\begin{proof}
    We compare the positive normalization factors.
    Because \([V_O,N_O]=0\), both outside Hamiltonians decompose into fixed-number sectors.
    Within each sector, Weyl's inequality moves every eigenvalue of \(H_O\) by at most
\[
    |\lambda|\|V_O\|
    =
    |\lambda|.
\]
    Exponentiating, summing the eigenvalues in each sector, and then weighting the sectors by the positive factors \(|z|^{N_O}\) gives
    \begin{equation}
        \e^{-\beta|\lambda|}
        Z_{O,0}(\beta,|z|)
        \le
        Z_{O,\lambda}(\beta,|z|)
        \le
        \e^{\beta|\lambda|}
        Z_{O,0}(\beta,|z|)
        \label{eq:outside-positive-partition-transfer}
    \end{equation}
    by a corresponding argument to that used to derive \eqref{eq:Weyl-upper-bound}.
    
    Next use \cref{prop:spectator-factorization}, at \(\lambda\) and at \(\lambda=0\).
    The inside factors \(Z_I\) cancel, and the monomials in the inverse-fugacity
    polynomials cancel, so
    \begin{equation}
        \frac{\Xi_\lambda(s,z)}{\Xi_0(s,z)}
        =
        r_\lambda(w)
        \frac{
        Z_{O,0}(\beta,|z|)
        }{
        Z_{O,\lambda}(\beta,|z|)
        },
        \qquad
        z=-1/w.
        \label{eq:normalized-signal-ratio-factorization}
    \end{equation}
    By hypothesis, we have at every quadrature point
\[
    |r_\lambda(w_j)|
    \ge
    1-\eta,
\]
    while the upper bound in
    \eqref{eq:outside-positive-partition-transfer} gives
\[
    \frac{
    Z_{O,0}(\beta,|z|)
    }{
    Z_{O,\lambda}(\beta,|z|)
    }
    \ge
    \e^{-\beta|\lambda|}.
\]
    Substitution in
    \eqref{eq:normalized-signal-ratio-factorization} proves
    \eqref{eq:normalized-signal-transfer}.
\end{proof}

\ProofOf{prop:quadrature-gap-transfer}
\begin{proof}
    Fix \(|t|\le t_{\max}\) and one of the two contour radii
\(R=\e^{-\beta\mu}\).
    The boundary circles of \(\mathcal C_\gamma(R)\) correspond to the shifted
    chemical potentials \(\mu\mp\gamma/\beta\).
    We then have
\[
    \min_{j\in O}
    \left|E_j-\left(\mu\mp\frac{\gamma}{\beta}\right)\right|
    \ge
    \Delta_{\Omega}-\frac{\gamma}{\beta}.
\]
    Repeating the estimate in \cref{prop:outside-duhamel-bound} on either
    shifted circle gives
\[
    \left|
    \frac{P_{O,\lambda}(s,w)}{P_{O,0}(s,w)}-1
    \right|
    \le
    \sqrt{1+\frac{ t_{\max}^2}{\beta^2}}
    \left(\e^{\beta|\lambda|}-1\right)
    \tanh\!\left(
    \frac{\beta\Delta_{\Omega}-\gamma}{2}
    \right)^{-M_O}
    <1.
\]
    The relative Rouch\'e criterion from
    \cref{cor:relative-spectator-stability} therefore preserves the zero count
    inside each boundary circle.
    At \(\lambda=0\), the two counts agree because
\(\gamma<\beta\Delta_{\Omega}\).
    The perturbed counts also agree, so \(P_{O,\lambda}(s,\cdot)\) has no zeros
    in \(\mathcal C_\gamma(R)\).
    Applying the argument at both contour radii proves that the thickened region is zero-free.
    
    By \cref{eq:spectator-P-factorization}, the roots of \(P_\lambda(s,\cdot)\) are those of the protected factor \(P_I\) together with those of the spectator factor \(P_{O,\lambda}\).
    Every root of the protected factor of the partition polynomial has the form
\[
    \zeta_j=\e^{-sE_j},
    \qquad
    j\in I.
\]
    Its logarithmic radial distance from either contour is at least
\(\beta\Delta_\Omega\).
    Every root of the spectator factor has logarithmic radial distance at least
\(\gamma\) from either contour by the collar hypothesis.
    Thus every root of the product has logarithmic radial distance at least
\[
    \min\{\beta\Delta_\Omega,\gamma\}.
\]
    Exponentiating this statement gives
    \eqref{eq:perturbed-root-radial-separation}.
    
    The single-root aliasing estimate of \cref{lem:single-root-aliasing} depends only on
    the root location relative to the contour, and never requires a single-particle interpretation.
    Apply that estimate to all \(M\) roots and use the common bound
\(q_\lambda\).
    The same inside--outside split as in the proof of \cref{prop:quadrature-gap-bound} gives
    \eqref{eq:perturbed-disc-quadrature-bound}.
\end{proof}
\ProofOf{cor:constant-perturbed-quadrature-gap}
\begin{proof}
    The time condition in
    \eqref{eq:constant-quadrature-gap-conditions} is stronger than the one in
    \cref{cor:log-gap-perturbative-stability}, so the root moments are
    protected.
    Moreover, the hypothesis on \(\beta \Delta_\Omega\) gives
\[
    \beta\Delta_{\Omega}-\log\frac32
    \ge
    \log(2M).
\]
    Consequently,
\[
    \tanh\!\left(
    \frac{\beta\Delta_{\Omega}-\log(3/2)}{2}
    \right)^{-M_O}
    \le
    \left(\frac{2M+1}{2M-1}\right)^M
    \le
    3.
\]
    Convexity of the exponential and
\(|s|/\beta\ge1\) imply
\[
    \frac{|s|}{\beta}
    \left(\e^{\beta|\lambda|}-1\right)
    \le
    \e^{|s||\lambda|}-1.
\]
    Uniformly for \(|t|\le t_{\max}\), the second condition in
    \eqref{eq:constant-quadrature-gap-conditions} therefore gives
\[
    \sqrt{1+\frac{ t_{\max}^2}{\beta^2}}
    \left(\e^{\beta|\lambda|}-1\right)
    <
    \frac13.
\]
    Thus \eqref{eq:quadrature-collar-smallness} holds with
\(\gamma=\log(3/2)\), and
    \cref{prop:quadrature-gap-transfer} applies.
    Finally,
\[
    \beta\Delta_\Omega
    \ge
    \log(3M)
    >
    \log\frac32,
\]
    so \eqref{eq:perturbed-quadrature-gap} gives \(q_\lambda=2/3\).
\end{proof}

\section{The LP estimator}\label{app:lp-estimator}

Let \(\rho\) be an unknown probability measure on the unit torus, with Fourier
coefficients \(u_\ell\), and suppose that only noisy coefficients
\(\widehat u_\ell\) with \(|\ell|\le T\) are available.
Our argument rests on the one-dimensional stability theorem of
Chen--De--Huang--Nadimpalli--Servedio--Yang~\cite{ChenEtAl2025}, which we use without proof.
Here \(\mathbb T=[0,1)\) is the unit torus with its toroidal metric, \(W_1^{\mathbb T}\) is the Wasserstein-1 distance for that metric, and the Fourier coefficients of a probability measure \(\rho\) on \(\mathbb T\) are \(u_\ell=\int_{\mathbb T}\e^{2\pi\ii\ell x}\dd\rho(x)\), \(\ell\in\mathbb Z\).

\begin{lemma}[Stability of low-frequency Fourier data {\cite[Theorem~7, case \(d=1\)]{ChenEtAl2025}}]
    \label[lemma]{lem:chen-stability}
    Let \(0<\varepsilon'<\tfrac12\), and set
\[
    T=\left\lceil\frac{6}{\varepsilon'}\right\rceil,
    \qquad
    \kappa_0=\frac{10^{-3}\,\varepsilon'}{\log(1/\varepsilon')},
\]
    where \(\log\) denotes the natural logarithm, as in \cite{ChenEtAl2025}.
    Let \(\rho_1,\rho_2\) be probability measures on \(\mathbb T\) with Fourier coefficients \(u^{(1)}_\ell,u^{(2)}_\ell\).
    If \(|u^{(1)}_\ell-u^{(2)}_\ell|\le\kappa\) for all integers \(|\ell|\le T\), for some \(\kappa\le\kappa_0\), then
    \(W_1^{\mathbb T}(\rho_1,\rho_2)\le\varepsilon'\).
\end{lemma}

We apply \cref{lem:chen-stability} in the proof of \cref{prop:wasserstein-window-recovery} with \(\varepsilon'=\epsilon/(4L)\in(0,\tfrac18)\).
Then \(\lceil 6/\varepsilon'\rceil=\lceil 24L/\epsilon\rceil=T\), and \(\kappa_0=10^{-3}\epsilon/(4L\log(4L/\epsilon))\), so the requirement \(\kappa\le\kappa_0\) holds for \(\kappa=c\,\epsilon/(L\log(4L/\epsilon))\) as soon as \(c\le1/4000\).
This is the sense in which \(c\) is a sufficiently small universal constant in that proposition.
It is therefore enough to find \emph{any} probability measure consistent with
the observed coefficients to the required accuracy.

There are two obstacles to expressing this search as a finite-dimensional
linear program.
First, \(\rho\) may be atomic, so a direct piecewise-constant approximation
need not be accurate.
Convolution with a Jackson kernel moves \(\rho\) by only \(O(1/n)\) in
Wasserstein distance while producing a smooth density.
Second, the space of smooth densities is still infinite dimensional.
Partitioning the torus into \(K\) cells and using the mass in each cell as a
variable gives a finite probability simplex.
The scale \(K=O(n^2/\kappa)\) is chosen because the Jackson-smoothed density
has Lipschitz constant \(O(n^2)\), so replacing it by its cell averages changes
its Fourier coefficients by \(O(\kappa)\).
\Cref{def:jackson-kernel,def:chen-lp-estimator} make this construction explicit.

The logic behind the constraints is as follows.
Let \(\bar\rho=\rho*J_n\), whose \(\ell\)th Fourier coefficient is
\(\widetilde{J}_{n}(\ell)u_\ell\), and set \(a_j=\bar\rho(I_j)\).
The resulting piecewise-constant probability measure has the same mass as
\(\bar\rho\) in each cell.
For sufficiently large \(C\), the Lipschitz bound on the
Jackson smoothing makes its Fourier coefficients differ from those of
\(\bar\rho\) by at most \(\kappa/8\).
If the input data also obey
\(|\widehat u_\ell-u_\ell|\le\kappa/8\), multiplication by the Jackson
multiplier does not increase this error, since
\(|\widetilde{J}_{n}(\ell)|\le1\).
The two errors sum to at most \(\kappa/4\), proving that the LP is feasible.
This is why the \(\kappa/8\) data tolerance provides useful slack before the
recovery estimate is applied.

Conversely, the real and imaginary constraints imply that any feasible
output \(\sigma\) satisfies
\[
\left|
\int_{\mathbb T}\e^{2\pi\ii\ell x}\dd\sigma(x)
-\widetilde{J}_{n}(\ell)u_\ell
\right|
\le
\frac{\sqrt2\kappa}{4}+\frac{\kappa}{8}
<
\kappa
\]
for \(|\ell|\le T\).
For the choices in \cref{prop:wasserstein-window-recovery}, the
\cref{lem:chen-stability} therefore places \(\sigma\) within
\(\epsilon/(4L)\) of \(\bar\rho\), while Jackson smoothing moves \(\rho\)
by at most another \(1/n\le\epsilon/(4L)\).
Thus the estimator is deliberately nonunique: no support separation or
parametric model is imposed, and every feasible LP solution gives the required
Wasserstein approximation.

\end{document}

%% file: annulus_recovery.tex
\begingroup
\fontencoding{OT1}\selectfont

\definecolor{annulusblue}{HTML}{35618D}
\definecolor{collargold}{HTML}{C18A20}
\definecolor{protectedpurple}{HTML}{7B3294}
\definecolor{otherred}{HTML}{C65332}
\definecolor{nodeteal}{HTML}{008579}

\begin{tikzpicture}[
    font=\fontsize{9pt}{11pt}\selectfont, >=Stealth,
    boundary/.style={annulusblue, line width=0.85pt},
    collar edge/.style={collargold!85, densely dashed, line width=0.45pt},
    protected root/.style={circle, fill=protectedpurple, draw=white,
        line width=0.5pt, inner sep=0pt, minimum size=5.6pt},
    sample/.style={rectangle, fill=nodeteal, draw=white,
        line width=0.5pt, inner sep=0pt, minimum size=5.2pt},
    root label/.style={text=protectedpurple, inner sep=2pt},
    other root/.style={circle, fill=otherred, draw=white,
        line width=0.5pt, inner sep=0pt, minimum size=5.6pt},
    guide/.style={black!45, line width=0.4pt},
    key text/.style={anchor=west, align=left, inner sep=0pt}
]
\def\Rin{0.90}
\def\Rout{2.40}
\def\q{0.84}
\pgfmathsetmacro{\Lvalue}{ln(\Rout/\Rin)}

\fill[annulusblue!9, even odd rule]
    (0,0) circle (\Rout) (0,0) circle (\Rin);
\foreach \R in {\Rin,\Rout} {
    \pgfmathsetmacro{\lower}{\q*\R}
    \pgfmathsetmacro{\upper}{\R/\q}
    \fill[collargold!16, even odd rule]
        (0,0) circle (\upper) (0,0) circle (\lower);
    \draw[collar edge] (0,0) circle (\lower);
    \draw[collar edge] (0,0) circle (\upper);
    \draw[boundary] (0,0) circle (\R);
}

\draw[guide, ->] (-3.08,0) -- (3.13,0) node[right, black!65] {$\operatorname{Re}w$};
\draw[guide, ->] (0,-3.07) -- (0,3.10) node[above, black!65] {$\operatorname{Im}w$};
\node[below left, black!65, inner sep=2pt] at (0,0) {$0$};

\foreach \R in {\Rin,\Rout} {
    \foreach \j in {0,...,4} {
        \node[sample] at ({72*\j}:\R) {};
    }
}

\foreach \a/\fraction in {1/0.23,2/0.50,3/0.74} {
    \pgfmathsetmacro{\rad}{\Rout*exp(-\Lvalue*\fraction)}
    \coordinate (p\a) at ({-360*\fraction}:\rad);
    \node[protected root] at (p\a) {};
}
\node[root label, anchor=east] at ([xshift=-5pt,yshift=5pt]p1) {$\mathrm e^{-s_\ell E_1}$};
\node[root label, anchor=south] at ([yshift=5pt]p2) {$\mathrm e^{-s_\ell E_2}$};
\node[root label, anchor=west] at ([xshift=5pt]p3) {$\mathrm e^{-s_\ell E_3}$};

\foreach \angle/\rad in {137/0.44,33/3.16,225/3.16} {
    \node[other root] at (\angle:\rad) {};
}
\node[otherred, anchor=south west, inner sep=2pt] at (33:3.16) {$\zeta$};

\node[annulusblue] at (135:1.65) {$A_\Omega$};
\draw[annulusblue, line width=0.45pt] (315:\Rin) -- (0.36,-0.38)
    node[above, inner sep=2pt] {$R_{\mathrm{in}}$};
\draw[annulusblue, line width=0.45pt] (243:\Rout) -- (-1.38,-2.69)
    node[below, inner sep=2pt] {$R_{\mathrm{out}}$};

\begin{scope}[xshift=8mm]
\node[key text, font=\fontsize{9pt}{11pt}\selectfont\bfseries] at (3.82,2.88) {Key};
\fill[collargold!16] (3.82,2.31) rectangle (4.20,2.55);
\draw[boundary] (3.82,2.43) -- (4.20,2.43);
\node[key text] at (4.40,2.43) {Annulus $A_\Omega$ and its boundaries};
\fill[collargold!16] (3.82,1.77) rectangle (4.20,2.01);
\draw[collar edge] (3.82,1.77) -- (4.20,1.77);
\draw[collar edge] (3.82,2.01) -- (4.20,2.01);
\node[key text] at (4.40,1.89) {Root-free collars (open)};
\node[protected root] at (4.01,1.35) {};
\node[key text] at (4.40,1.35) {Protected roots, $d=3$};
\node[other root] at (4.01,0.81) {};
\node[key text] at (4.40,0.81) {Other roots of $P(s_\ell,\cdot)$};
\node[sample] at (4.01,0.27) {};
\node[key text] at (4.40,0.27) {Quadrature nodes, $S=5$ per circle};
\node[key text, font=\fontsize{8pt}{9.5pt}\selectfont] at (4.40,-0.12)
    {$w_j=R\mathrm e^{\mathrm i \theta_j}$, $\quad \theta_j = 2\pi j /5,\ j =0,\ldots,4$};

\node[key text, font=\fontsize{9pt}{11pt}\selectfont\bfseries] at (3.82,-0.76)
    {Separation at either boundary $R$};
\fill[collargold!16] (5.34,-1.82) rectangle (7.14,-1.26);
\draw[collar edge] (5.34,-1.82) -- (5.34,-1.26);
\draw[collar edge] (7.14,-1.82) -- (7.14,-1.26);
\draw[boundary] (6.24,-1.82) -- (6.24,-1.26);
\draw[guide, ->] (3.92,-1.54) -- (8.76,-1.54);
\foreach \x in {4.49,8.08} {
    \node[circle, fill=black!60, inner sep=0pt, minimum size=4.6pt]
        at (\x,-1.54) {};
    \node[above, black!65, inner sep=3pt] at (\x,-1.49) {$\zeta$};
}
\node[above, collargold!90!black, font=\fontsize{8pt}{9.5pt}\selectfont] at (6.24,-1.26) {no roots};
\node[below, inner sep=3pt] at (5.34,-1.82) {$qR$};
\node[below, inner sep=3pt] at (6.24,-1.82) {$R$};
\node[below, inner sep=3pt] at (7.14,-1.82) {$R/q$};
\node[below, black!65, font=\fontsize{8pt}{9.5pt}\selectfont] at (8.53,-1.61) {$\log|w|$};
\end{scope}
\path[use as bounding box]
    ([xshift=-6pt,yshift=-6pt]current bounding box.south west)
    rectangle ([xshift=6pt,yshift=6pt]current bounding box.north east);
\end{tikzpicture}
\endgroup

%% file: W_circuit.tex
\begin{tikzpicture}[
    x=1cm, y=1cm, line width=0.6pt,
    wgate/.style={draw, fill=white, minimum width=1.15cm,
                  minimum height=0.65cm, inner sep=4pt},
    wtarget/.style={draw, fill=white, minimum width=0.65cm,
                    minimum height=0.65cm, inner sep=3pt}
]
\foreach \y/\register in {1.25/R,0/C,-1.25/S} {
        \draw (0,\y) -- (7.3,\y);
        \node[anchor=east] at (-0.12,\y) {$\mathsf{\register}$};
    }
    \draw (0.42,-0.10) -- (0.62,0.10);
    \node[above, font=\scriptsize] at (0.52,0.10) {$2m+2$};

    \node[wgate] at (1.90,0) {$U_{\mathrm{rot}}$};

\draw (4.10,0) -- (4.10,-1.25);
    \fill (4.10,0) circle (2.3pt);
    \node[above=4pt] at (4.10,0) {$\Pi_{\geq}$};
    \node[wtarget] at (4.10,-1.25) {$X$};

    \draw (6.10,0) -- (6.10,1.25);
    \fill (6.10,0) circle (2.3pt);
    \node[below=4pt] at (6.10,0) {$\Pi_{\leq}$};
    \node[wtarget] at (6.10,1.25) {$X$};

    \node[below, font=\small] at (4.10,-1.75) {$F_S$};
    \node[below, font=\small] at (6.10,-1.75) {$F_R$};
\end{tikzpicture}